\documentclass[final,onefignum,onetabnum]{siamart190516}

\usepackage{amssymb}
\usepackage{amsopn}
\usepackage{amsfonts}
\usepackage{graphicx}
\usepackage{enumerate}
\usepackage{bbm}
\usepackage{verbatim}
\usepackage{wasysym}

\newsiamremark{remark}{Remark}
\newsiamremark{hypothesis}{Hypothesis}
\crefname{hypothesis}{Hypothesis}{Hypotheses}
\newsiamthm{claim}{Claim}
\newsiamthm{fact}{Fact}
\newsiamthm{assumption}{Assumption}

\usepackage{hyperref}
\usepackage[capitalize, nameinlink]{cleveref}

\usepackage{appendix}

\crefname{appsec}{Appendix}{Appendices}

\def\Z{{\mathbb{Z}}}

\def\R{{\mathbb{R}}}
\def\C{{\mathbb{C}}}
\def\N{{\mathbb{N}}}

\DeclareMathOperator{\var}{Var} 
\DeclareMathOperator{\dist}{dist} 
\DeclareMathOperator{\diag}{diag} 

\DeclareMathOperator{\defin}{def} 
\DeclareMathOperator{\poly}{poly} 

\newcommand{\wrapp}[1]{\left({#1}\right)}
\newcommand{\wrapb}[1]{\left[{#1}\right]}
\newcommand{\wrapc}[1]{\left\{{#1}\right\}}

\newcommand{\abs}[1]{\left|{#1}\right|}
\newcommand{\norm}[1]{\left\|{#1}\right\|}

\allowdisplaybreaks[1]

\newcommand{\Var}{\mathrm{Var}}
\newcommand{\Cov}{\mathrm{Cov}}

\newcommand{\Ent}{\mathrm{Ent}}

\newcommand{\TV}[2]{\left\|#1 - #2\right\|_{\mathrm{TV}}}
\newcommand{\kl}[2]{D_{\mathrm{KL}} (#1 \parallel #2)}
\newcommand{\ceil}[1]{\left\lceil #1 \right\rceil}
\newcommand{\floor}[1]{\left\lfloor #1 \right\rfloor}

\DeclareMathOperator{\Spec}{Spec}

\newcommand{\fpras}{\mathsf{FPRAS}}

\newcommand{\fptas}{\mathsf{FPTAS}}
\DeclareMathOperator{\NP}{\mathsf{NP}}
\DeclareMathOperator{\RP}{\mathsf{RP}}

\newcommand{\eps}{\varepsilon}
\newcommand{\PP}{\mathbb{P}}
\newcommand{\EE}{\mathcal{E}}

\newcommand{\MM}{\Psi}

\DeclareMathOperator{\M}{\mathcal{M}}
\DeclareMathOperator{\SAW}{SAW}
\DeclareMathOperator{\T}{\mathbb{T}}

\newcommand{\SC}{\mathfrak{X}}
\newcommand{\bb}{b}

\newcommand{\seq}{\text{\scriptsize{\,$=$\,}}}

\newcommand{\sL}{{\sigma_\Lambda}}
\newcommand{\Cat}{C_1}

\newcommand{\Pgl}{P_{\textsc{gl}}}

\newcommand{\ii}{i}

\makeatletter
 \newcommand{\linkdest}[1]{\Hy@raisedlink{\hypertarget{#1}{}}}
\makeatother

\headers{Optimal Mixing of Glauber Dynamics}{Z. Chen, K. Liu, and E. Vigoda}

\title{Optimal Mixing of Glauber Dynamics: Entropy Factorization via High-Dimensional Expansion
\thanks{Submitted to the editors on \today. 
A preliminary short version of these results will appear in Proceedings of the 53rd Annual ACM Symposium on Theory of Computing (STOC 2021). 
\funding{Z. Chen and E. Vigoda are supported by NSF grant CCF-2007022. K. Liu is supported by NSF grant CCF-1907845 and ONR-YIP grant N00014-17-1-2429.}}
}
\author{
	Zongchen Chen
	\thanks{Georgia Institute of Technology, Atlanta, GA (\email{chenzongchen@gatech.edu}).}
	\and Kuikui Liu
	\thanks{University of Washington, Seattle, WA (\email{liukui17@cs.washington.edu}).}
	\and Eric Vigoda
	\thanks{University of California, Santa Barbara, CA
	(\email{vigoda@ucsb.edu}).}
}
\date{\today}

\begin{document}
	
\maketitle

\begin{abstract}
We prove an optimal mixing time bound for the single-site update Markov chain known as the Glauber dynamics or Gibbs sampling in a variety of settings. Our work presents an improved version of the spectral independence approach of Anari et al. (2020) and shows $O(n\log{n})$ mixing time on any $n$-vertex graph of bounded degree when the maximum eigenvalue of an associated influence matrix is bounded. As an application of our results, for the hard-core model on independent sets weighted by a fugacity $\lambda$, we establish $O(n\log{n})$ mixing time for the Glauber dynamics on any $n$-vertex graph of constant maximum degree $\Delta$ when $\lambda<\lambda_c(\Delta)$ where $\lambda_c(\Delta)$ is the critical point for the uniqueness/non-uniqueness phase transition on the $\Delta$-regular tree. More generally, for any antiferromagnetic 2-spin system we prove $O(n\log{n})$ mixing time of the Glauber dynamics on any bounded degree graph in the corresponding tree uniqueness region. Our results apply more broadly; for example, we also obtain $O(n\log{n})$ mixing for $q$-colorings of triangle-free graphs of maximum degree $\Delta$ when the number of colors satisfies $q > \alpha \Delta$ where $\alpha \approx 1.763$, and $O(m\log{n})$ mixing for generating random matchings of any graph with bounded degree and $m$ edges.

Our approach is based on two steps. First, we show that the approximate tensorization of entropy (i.e., factorizing entropy into single vertices), which is a key step for establishing the modified log-Sobolev inequality in many previous works, can be deduced from entropy factorization into blocks of fixed linear size. Second, we adapt the local-to-global scheme of Alev and Lau (2020) to establish such block factorization of entropy in a more general setting of pure weighted simplicial complexes satisfying local spectral expansion; this also substantially generalizes the result of Cryan et al. (2019). 


\end{abstract}

\begin{keywords}
  approximate counting, Markov chain Monte Carlo, Glauber dynamics, spectral independence, high-dimensional expanders, entropy factorization
\end{keywords}

\begin{AMS}
  60J10, 68Q87, 68W20 
\end{AMS}




\section{Introduction}

This paper establishes a well-known conjecture that a popular Markov chain known as the Glauber dynamics
converges very quickly to its stationary distribution in the tree uniqueness region, i.e., decay
of correlations region.   The Glauber dynamics is the quintessential
example of a local Markov chain, and its convergence rate is of great interest due to its simplicity and wide
applicability.  

Our setting is the general framework of spin systems.  Spin systems capture
many combinatorial models of interest, including the hard-core model on weighted independent
sets, the Ising model, and colorings, and are equivalent to undirected graphical models.  
For integer $q\geq 2$, a $q$-state spin system is defined by a $q\times q$ interaction matrix $A$.
For a given graph $G=(V,E)$ with $n=|V|$ vertices, 
the configurations of the model are the collection $\Omega$ of assignments $\sigma:V\rightarrow [q]$
of spins to the vertices of the graph.  Each configuration $\sigma\in\Omega$ has an associated weight $w(\sigma)$
defined by the pairwise interactions weighted by the interaction matrix $A$, see \cref{sec:main-result}
for a detailed definition.

The Gibbs distribution $\mu$ is the probability distribution over the collection $\Omega$ of configurations and is
defined as $\mu(\sigma)=w(\sigma)/Z$ where $Z=\sum_{\sigma} w(\sigma)$ is the normalizing factor
known as the partition function.  Approximately sampling from the Gibbs distribution is polynomial-time
equivalent to approximating the partition function~\cite{JVV,SVV}. Given an $\eps>0$ and $\delta>0$, an
$\fpras$ for the partition function outputs a $(1\pm\eps)$-relative approximation of the partition function
with probability $\geq 1-\delta$, whereas an $\fptas$ is the deterministic analog (i.e., it achieves $\delta=0$).

The canonical example of a spin system in statistical physics is the Ising model.
The Ising model is a 2-spin system (i.e., $q=2$); the 
spin space is denoted as $\{+,-\}$ and the configurations of the model are the $2^n$ assignments of spins $\{+,-\}$ to the vertices of the
underlying graph.  In the simpler case without an external field
the Ising model has a single parameter $\beta>0$ corresponding to the (logarithm of the) inverse temperature.
A configuration $\sigma\in\Omega$ has weight $w(\sigma) = \beta^ {m(\sigma)}$ where 
$m(\sigma)= |\{(u,v)\in E: \sigma(u)=\sigma(v)\}|$ is the
number of monochromatic edges in $\sigma$.  When $\beta>1$ then the model is {\em ferromagnetic} as
the two fully monochromatic configurations have maximum weight, whereas when $\beta<1$ then the
model is {\em antiferromagnetic}.  

The hard-core model is a natural combinatorial example of an antiferromagnetic 2-spin system.  The
model is parameterized by a fugacity $\lambda>0$.  For a graph $G=(V,E)$, configurations of the model
are the collection $\Omega$ of independent sets of $G$, and the weight
of an independent set $\sigma$ is $w(\sigma)=\lambda^{|\sigma|}$.

In general, a 2-spin system is defined by three parameters $\beta,\gamma\geq 0$ and $\lambda>0$.  A spin configuration $\sigma\in\{0,1\}^V$ is assigned weight: 
$w(\sigma)=\beta^{m_1(\sigma)}\gamma^{m_0(\sigma)}\lambda^{n_1(\sigma)},$
where, for $s\in\{0,1\}$,  $m_s(\sigma)$ is the number of edges where 
both endpoints receive spin $s$ and $n_s(\sigma)$ is the number of vertices assigned spin $s$.  Note the Ising model corresponds to the
case $\beta=\gamma$ where $\lambda$ is the external field, and 
the hard-core model corresponds to $\beta=0,\gamma=1$.
The model is ferromagnetic when $\beta\gamma>1$ and antiferromagnetic when $\beta\gamma<1$ (the model is trivial when $\beta\gamma=1$).

The Glauber dynamics is a simple Markov chain $(X_t)$ designed for sampling from the Gibbs distribution~$\mu$.  
The transitions $X_t\rightarrow X_{t+1}$ update a randomly chosen vertex as follows:
(i) select a vertex $v$ uniformly at random; (ii) for all $u\neq v$, set $X_{t+1}(u)=X_t(u)$; and (iii)
choose $X_{t+1}(v)$ from the marginal distribution for the spin at $v$ conditional on the configuration $X_{t+1}(N(v))$ on
the neighbors $N(v)$ of $v$.  It is straightforward to verify that the chain is ergodic (in the cases considered here, see the definition
of totally-connected in \cref{sec:main-result})
and the unique stationary distribution is the Gibbs distribution. 

The {\em mixing time} is the number of transitions, for the worst initial state $X_0$,
to guarantee that $X_t$ is within total variation distance $\leq 1/4$ of the Gibbs distribution; for a formal statement, see \cref{eqn:mixing}.
We say the chain is {\em rapidly mixing} when the mixing time is polynomial in $n=|V|$.
Hayes and Sinclair~\cite{HayesSinclair} established that the mixing time of the Glauber dynamics is $\Omega(n\log{n})$
for a family of bounded-degree graphs, and hence we say that the Glauber dynamics has optimal mixing time
when the mixing time is $O(n\log{n})$.

The computational complexity of approximating the partition function 
is closely connected to statistical physics phase transitions.  For $\Delta\geq 3$, consider the tree $T_\ell$ 
of height $\ell$ where all of the internal vertices have degree $\Delta$, and let $r$ denote its root. 
The uniqueness/non-uniqueness phase transition captures whether the leaves influence the root, in the limit as the
height grows.  

The uniqueness/non-uniqueness phase transition 
is nicely illustrated for the Ising model
which has two extremal boundaries: the all $+$ boundary and all $-$ boundary.  For $s\in\{+,-\}$, let $p_\ell^s$ denote the
marginal probability that the root has spin $+$ in the Gibbs distribution on $T_\ell$ conditional on all leaves having spin $s$.  
The model is in the uniqueness phase iff 
$\lim_{\ell\rightarrow\infty} p_{\ell}^+ = \lim_{\ell\rightarrow\infty} p_{\ell}^-$.
For the Ising model (without an external field) 
the uniqueness/non-uniqueness phase transition occurs at $\beta_c(\Delta) = (\Delta-2)/\Delta$ for the antiferromagnetic case and $\overline{\beta}_c(\Delta)=\Delta/(\Delta-2)$ for the ferromagnetic case. For the hard-core model, the critical fugacity is 
$\lambda_c(\Delta):=(\Delta-1)^{\Delta-1}/(\Delta-2)^{\Delta}$.
This phase transition on the $\Delta$-regular tree is connected 
to the complexity of approximating the partition function on
graphs of maximum degree $\Delta$.

For the hard-core model, for constant $\Delta$, for any $\delta>0$, 
Weitz~\cite{Wei06} presented an $\fptas$ for the partition function on
graphs of maximum degree $\Delta$ when $\lambda<(1-\delta)\lambda_c(\Delta)$.
In contrast, when $\lambda>\lambda_c$, Sly~\cite{Sly10} (see also~\cite{SS14,GSV16}) showed that, 
unless $\NP=\RP$, there is no $\fpras$ for
approximating the partition function on graphs of maximum degree $\Delta$.
Li, Lu, and Yin~\cite{LLY13} generalized Weitz's correlation
decay algorithmic approach
to all antiferromagnetic $2$-spin systems when the system is \emph{up-to-$\Delta$ unique} (see \cref{def:2spin-uniqueness}).
One important caveat to these correlation decay approaches is that the  
running time depends exponentially on $\log{\Delta}$ and~$1/\delta$.

Despite the algorithmic successes of the correlation decay approach,
establishing rapid mixing of the Glauber dynamics in the same tree uniqueness region was a vexing open problem. 
Anari, Liu, and Oveis Gharan~\cite{ALO20} introduced the spectral independence approach based on the theory of high-dimensional expanders~\cite{KM17, DK17, KO18, Opp18, AL20}, and established rapid mixing of the Glauber dynamics
for the hard-core model on any graph of maximum degree $\Delta$ when $\lambda<(1-\delta)\lambda_c(\Delta)$
for $\delta>0$.  However, while the mixing time had polynomial dependence on $\Delta$, it also had doubly exponential dependence on $1/\delta$. In~\cite{CLV20} the authors established rapid mixing for all 
antiferromagnetic 2-spin systems when the system is up-to-$\Delta$-unique with gap $\delta$ and improved the mixing time
to an exponential dependence on $1/\delta$. Here, roughly speaking, up-to-$\Delta$ uniqueness with gap $\delta$ means (multiplicative) gap $\delta$ from the uniqueness threshold on the $\Delta$-regular tree for all $d\leq\Delta$; 
see \cref{def:2spin-uniqueness} for a precise statement, and \cite{LLY13} for more discussion.

In this work, we not only establish a fixed polynomial upper bound on the mixing time, but we also prove optimal
mixing of the Glauber dynamics.   Our approach holds for general spin systems.  The spectral independence 
approach, first introduced for 2-spins in \cite{ALO20} and subsequently extended to $q$-spins in \cite{CGSV20, FGYZ20}, considers the $qn\times qn$ influence matrix.  For spins $i,j\in [q]$ and vertices $u,v\in V$,
the entry $((u,i),(v,j))$ of the influence matrix measures the effect of vertex $u$ having spin $i$ on the marginal
probability that vertex $v$ has spin $j$, see \cref{def:influence-matrix} for a precise statement.
Here we prove that if the maximum eigenvalue of the influence matrix is upper bounded 
and the marginal probabilities are lower bounded then the mixing time is $O(n\log{n})$ where the only
dependence on $1/\delta$ and $\Delta$ is in the constant factor captured by the big-O notation.
Our main result is stated in \cref{thm:main} in \cref{sec:main-result} after
presenting the necessary definitions.

We establish optimal mixing time of $O(n\log{n})$ by proving that the Glauber dynamics
contracts relative entropy (with respect to the Gibbs distribution) at a constant rate.
This is analogous to establishing a modified log-Sobolev constant for the Glauber dynamics;
there are several recent results in other contexts also proving entropy decay for various
Markov chains~\cite{CGM19,CP20,BCPSV20}. 
In contrast, previous works utilizing the spectral independence approach~\cite{ALO20,CLV20} and related works
on high-dimensional expanders~\cite{KM17, DK17, Opp18,KO18,AL20} consider the spectral gap (or analogously, decay of variance); such an approach is unable to establish optimal mixing time.
Our proof approach is outlined in \cref{sec:entropy-outline}.

The application of our results is nicely illustrated for the 
particular case of antiferromagnetic $2$-spin systems.  
We prove $O(n\log{n})$ mixing time of the Glauber dynamics when the system
is up-to-$\Delta$-unique.  This is the same region where the 
correlation decay results of~\cite{LLY13} and the rapid
mixing results of~\cite{CLV20} hold, which matches the hardness results of~\cite{SS14}. 
Note, a mixing time of $O(n\log{n})$ implies an $\widetilde{O}(n^2)$ time $\fpras$ for approximating the partition function~\cite{SVV,Kolmogorov}.

\begin{theorem}\label{thm:2-spin}
For all $\Delta \ge 3$, all $\delta \in(0,1)$, and all parameters $(\beta,\gamma,\lambda)$ specifying an antiferromagnetic 2-spin system which is up-to-$\Delta$ unique with gap $\delta$, 
there exists $C = C(\Delta,\delta,\beta,\gamma,\lambda)$ such that 
for every $n$-vertex graph $G=(V,E)$ of maximum degree at most $\Delta$, 
the mixing time of the Glauber dynamics for the 2-spin system on $G$ with parameters $(\beta,\gamma,\lambda)$ is at most $C n\log n$.
\end{theorem}

For the case of the hard-core model our theorem yields the following result.
\begin{theorem}\label{thm:hard-core}
For all $\Delta \ge 3$ and all $\delta \in (0,1)$, 
there exists $C=C(\Delta,\delta)$ such that 
for every $n$-vertex graph $G=(V,E)$ of maximum degree at most $\Delta$ and every $\lambda \leq (1 - \delta)\lambda_{c}(\Delta)$, 
the mixing time of the Glauber dynamics for the hard-core model on $G$ with fugacity $\lambda$ is at most $C n\log{n}$.
\end{theorem}
\begin{remark}
For interested readers, the constant $C = C(\Delta,\delta)$ scales roughly as $\Delta^{O(\Delta^{2}/\delta)}$.
\end{remark}

For the case of the Ising model in both the antiferromagnetic and ferromagnetic case, our theorem yields optimal mixing whenever $\beta$ is between $\beta_{c}(\Delta) = \frac{\Delta-2}{\Delta}$ and $\overline{\beta}_{c}(\Delta) = \frac{\Delta}{\Delta - 2}$.
\begin{theorem}\label{thm:ising}
For all $\Delta \ge 3$ and all $\delta \in (0,1)$, 
there exists $C=C(\Delta,\delta)$ such that 
for every $n$-vertex graph $G=(V,E)$ of maximum degree at most $\Delta$, every $\beta \in [\frac{\Delta-2+\delta}{\Delta-\delta}, \frac{\Delta-\delta}{\Delta-2+\delta}]$, and every $\lambda >0$, 
the mixing time of the Glauber dynamics for the Ising model on $G$ with inverse temperature $\beta$ and external field $\lambda$ is at most $C n\log{n}$.
\end{theorem}
\begin{remark}
We can actually show that specifically for the Ising model, $C = \Delta^{O(1/\delta)}$ suffices when $n$ is large enough and so we obtain polynomial mixing time even when the graph has unbounded degree.
\end{remark}

Recall that the above results are tight as there is no efficient approximation algorithm 
in the tree non-uniqueness region which corresponds to $\lambda>\lambda_c(\Delta)$ for the
hard-core model and $\beta<\beta_c(\Delta)$ for the antiferromagnetic Ising model.
The only analog of the above results establishing optimal mixing time in the entire tree uniqueness 
region was the work of Mossel and Sly~\cite{MS13} for the ferromagnetic Ising model.  Their proof
utilizes the monotonicity properties of the ferromagnetic Ising model which allows the use of the censoring
inequality of Peres and Winkler~\cite{PW13}.
The algorithm of Jerrum and Sinclair~\cite{JS:ising} gives an $\fpras$ for the ferromagnetic Ising model for any $\beta$ and any $G$, but the polynomial exponent is a large constant.

Our results hold for multi-spin systems as well.  The most notable example of a multi-spin system is the $q$-colorings problem,
namely, proper vertex $q$-colorings.
Given a graph $G=(V,E)$ of maximum degree $\Delta$, can we approximate the number of $q$-colorings of $G$?
Jerrum~\cite{Jerrum} proved $O(n\log{n})$ mixing time of the Glauber dynamics whenever $q>2\Delta$.
This was further improved in~\cite{CDMPP,Vigoda} 
to $O(n^2)$ mixing time when $q>(11/6 - \eps)\Delta$ for some small $\eps>0$.
There are several further improvements with various assumptions on the girth or maximum degree, c.f. \cite{DFHV}.
On the hardness side, Galanis et al.~\cite{GSV15} proved that unless $\NP=\RP$ there is no $\fpras$ 
for approximating the number of $q$-colorings when $q$ is even and $q<\Delta$.

For triangle-free graphs, a recent pair of works~\cite{FGYZ20,CGSV20} extended the spectral independence 
approach to establish rapid mixing of the Glauber dynamics when $q>(\alpha^*+\delta)\Delta$ for any $\delta>0$
where $\alpha^*\approx 1.763$; however the polynomial exponent in the mixing time depends on $1/\delta$ in these results.
Using our main result we prove $O(n\log{n})$ mixing time of the Glauber dynamics under the same conditions.

\begin{theorem}\label{thm:coloring}
Let $\alpha^*\approx 1.763$ denote the unique solution to $x=\exp(1/x)$.
For all $\Delta \ge 3$ and all $\delta>0$, 
there exists $C=C(\Delta,\delta)$ such that 
for every $n$-vertex triangle-free graph $G=(V,E)$ of maximum degree at most $\Delta$ and every $q \ge (\alpha^*+\delta)\Delta$, 
the mixing time of the Glauber dynamics for sampling random $q$-colorings on $G$ is at most $C n\log{n}$.
\end{theorem}

We prove spectral independence bounds for the monomer-dimer model on all matchings of a graph; no nontrivial bounds were previously known. Given a graph $G=(V,E)$ and a fugacity $\lambda>0$, the Gibbs distribution $\mu$ for the monomer-dimer model is defined on the collection $\M$ of all matchings of $G$ where $\mu(M)=w(M)/Z$ for $w(M)=\lambda^{|M|}$.  The Glauber dynamics for the monomer-dimer model adds or deletes a random edge in each step.
In particular, from $X_t\in\M$, choose an edge $e$ uniformly at random from $E$ and let $X'=X_t\oplus e$. 
If $X'\in\M$ then let $X_{t+1}=X'$ with probability $w(X') / (w(X') + w(X_t))$ and otherwise let $X_{t+1}=X_t$. 

We prove $O(m\log n)$ mixing time for the Glauber dynamics 
for sampling matchings on bounded-degree graphs with $n$ vertices and $m$ edges.  A classical result of Jerrum and Sinclair \cite{JS89b} yields rapid mixing of the Glauber dynamics for any graph, but the best mixing time bound was $O(n^2m\log{n})$~\cite{Jbook}.

\begin{theorem}\label{thm:matching}
For all $\Delta \ge 3$ and all $\lambda > 0$, 
there exists $C = C(\Delta,\lambda)$ such that 
for every $n$-vertex graph $G=(V,E)$ of maximum degree at most $\Delta$, 
the mixing time of the Glauber dynamics for the monomer-dimer model on $G$ with fugacity $\lambda$ is at most $C m\log n$.
\end{theorem}

For general ferromagnetic 2-spin systems the existing picture is not as clear as for antiferromagnetic systems.
Our work extends to ferromagnetic 2-spin systems, proving $O(n\log{n})$ mixing time for the same range 
of parameters as the previously best known bounds~\cite{GL18,SS19,CLV20}. In particular, we recover Theorems 26 and 27 in \cite{CLV20} with $O(n\log n)$ mixing time.

Finally, we mention that our techniques imply asymptotically optimal bounds (up to constant factors) on both the standard and modified log-Sobolev constants of the Glauber dynamics for spin systems on bounded degree graphs in all of the regimes mentioned above. This also applies for certain problems where prior works have obtained rapid mixing via other techniques such as path coupling and canonical paths. 

\subsection{Result for General Spin Systems}
\label{sec:main-result}
Our main results will follow from a general statement regarding the Glauber dynamics for an arbitrary spin system satisfying marginal bounds and spectral independence. We first proceed with a few definitions. 

Let $q \ge 2$ be an integer and $[q] = \{1,\dots,q\}$. 
Given a graph $G = (V,E)$, we consider the $q$-spin system on $G$ parameterized by a symmetric interaction matrix $A \in \R_{\ge 0}^{q\times q}$ representing ``interaction strengths'' and a field vector $h \in \R_{>0}^q$ representing ``external fields''. 
A configuration $\sigma \in [q]^V$ is an assignment of spins to vertices. 
The \emph{Gibbs distribution} $\mu = \mu_{G,A,h}$ over all configurations is given by
\[
\mu (\sigma) = \frac{1}{Z_G(A,h)} \prod_{\{u,v\} \in E} A(\sigma_u, \sigma_v) \prod_{v \in V} h(\sigma_v), \qquad \forall \sigma \in [q]^V
\]
where 
\[
Z_G(A,h) = \sum_{\sigma\in [q]^V} \prod_{\{u,v\} \in E} A(\sigma_u, \sigma_v) \prod_{v \in V} h(\sigma_v)
\]
is called the \emph{partition function}. 
The hard-core model, Ising model, random colorings, and monomer-dimer model (equivalent to hard-core model on line graphs) all belong to the family of spin systems. 

Let $\mu$ be an arbitrary distribution over $[q]^V$. 
A configuration $\sigma \in [q]^V$ is said to be \emph{feasible} with respect to $\mu$ if $\mu(\sigma) > 0$. 
Let $\Omega = \Omega(\mu)$ denote the collection of all feasible configurations (we omit $\mu$ when it is clear from the context); namely, $\Omega$ is the support of $\mu$. 
Furthermore, for $\Lambda \subseteq V$ let $\Omega_\Lambda = \{\tau \in [q]^\Lambda: \mu_\Lambda(\tau) > 0\}$ denote the collection of all feasible (partial) configurations on $\Lambda$, with the convention that $\Omega_v = \Omega_{\{v\}}$ for a single vertex $v$. 
Observe that $\Omega_V = \Omega$. 
For any subset $\Lambda \subseteq V$ and boundary condition $\tau \in \Omega_\Lambda$, we often consider the conditional distribution $\mu_S^\tau(\cdot) = \mu(\cdot | \sL \seq \tau)$ over configurations on $S = V \setminus \Lambda$, and we shall write $\Omega^\tau_U$ for the set of feasible (partial) configurations on $U \subseteq S$ under this conditional measure.

For a subset $S \subseteq V$, 
the Hamming graph $\mathcal{H}_{S}$ is defined to be the graph with vertex set $[q]^S$ of all configurations on $S$ such that two configurations are adjacent iff they differ at exactly one vertex. 
A collection $\Omega_0 \subseteq [q]^S$ of configurations on $S$ is said to be connected if the induced subgraph $\mathcal{H}_{S}[\Omega_0]$ of $\mathcal{H}_{S}$ is connected. 
A distribution $\mu$ over $[q]^V$ is said to be \emph{totally-connected} 
if for every nonempty subset $S \subseteq V$ and every boundary condition $\tau \in \Omega_{V \setminus S}$, the set $\Omega^\tau_S$ is connected. 


\begin{assumption}
Throughout the paper, we always assume that the distribution $\mu$ we are interested in is totally-connected. 
\end{assumption}

We remark that all soft-constraint models (i.e., $A(i,j) > 0$ for all $i,j \in [q]$) satisfy this assumption and common hard-constraint models, including the hardcore model, $q$-colorings when $q \ge \Delta+2$, and matchings, all satisfy this assumption as well.

The \emph{Glauber dynamics}, also known as the \emph{Gibbs sampling}, is a simple, natural, and popular Markov chain for sampling from a distribution $\mu$ over $[q]^V$. The dynamics starts with some (possibly random) configuration $X_0$. For every $t \ge 1$, a new random configuration $X_{t+1}$ is generated from $X_t$ as follows: pick a coordinate $v \in V$ uniformly at random, set $X_{t+1}(u) = X_{t}(u)$ for all $u \in V \setminus \{v\}$, and sample $X_{t+1}(v)$ from the conditional distribution $\mu(\sigma_v \seq \cdot \mid \sigma_{V \setminus \{v\}} \seq X_t(V \setminus \{v\}))$. 
Denote the transition matrix of the Glauber dynamics by $\Pgl$. 
If $\mu$ is totally-connected, then the Glauber dynamics is ergodic (i.e., irreducible and aperiodic) and has stationary distribution $\mu$. 

Let $P$ be the transition matrix of an ergodic Markov chain $(X_t)$ on a finite state space $\Omega$ with stationary distribution $\mu$. 
For $t\ge 0$ and $\sigma \in \Omega$, let $P^t(\sigma,\cdot)$ denote the distribution of $X_t$ when starting the chain with $X_0 = \sigma$. 
For $\eps \in (0,1)$, the \emph{mixing time} of $P$ is defined as
\begin{equation}\label{eqn:mixing}
T_{\mathrm{mix}}(P, \eps) = \max_{\sigma \in \Omega} \min \left\{ t\ge 0: \TV{P^t(\sigma,\cdot)}{\mu} \le \eps \right\}. 
\end{equation}

We will require two conditions of the distribution $\mu$. 
The first is that the marginal probability of each vertex is bounded away from $0$ under any conditioning.

\begin{definition}[Bounded Marginals]
\label{def:bounded-marg}
We say a distribution $\mu$ over $[q]^V$ is \emph{$\bb$-marginally bounded} if 
for every $\Lambda \subsetneq V$ and $\tau \in \Omega_\Lambda$, 
it holds for every $v \in V \setminus \Lambda$ and $i \in \Omega_v^\tau$ that, 
\[
\mu(\sigma_v \seq i \mid \sL \seq \tau) \ge \bb.
\]
\end{definition}

The second condition is the notion of spectral independence, first introduced by \cite{ALO20} and later generalized to multi-spin systems in \cite{CGSV20,FGYZ20}. 
Here we use the definitions from \cite{CGSV20}. 

\begin{definition}[Influence Matrix]
\label{def:influence-matrix}
Given $\Lambda \subsetneq V$ and $\tau \in \Omega_\Lambda$, let 
\[
\widetilde{V}_\tau = \{(u,i): u\in V \setminus \Lambda, i \in \Omega_u^\tau\}. 
\]
For every $(u,i), (v,j) \in \widetilde{V}_\tau$ with $u\neq v$, we define the \emph{(pairwise) influence} of $(u,i)$ on $(v,j)$ conditioned on $\tau$ by 
\[
\MM_\mu^\tau\big( (u,i),(v,j) \big) = \mu(\sigma_v \seq j \mid \sigma_u \seq i, \sL \seq \tau) - \mu(\sigma_v \seq j \mid \sL \seq \tau).
\]
Furthermore, let $\MM_\mu^\tau\big( (v,i),(v,j) \big) = 0$ for all $(v,i), (v,j) \in \widetilde{V}_\tau$. 
We call $\MM_\mu^\tau$ the \emph{(pairwise) influence matrix} conditioned on $\tau$. 
\end{definition}

Note that all eigenvalues of the influence matrix $\MM_\mu^\tau$ are real; for instance, see \cite{ALO20,CGSV20,BCCPSV21}. 

\begin{definition}[Spectral Independence]
\label{def:spectral-ind}
We say a distribution $\mu$ over $[q]^V$ is \emph{$\eta$-spectrally independent} if 
for every $\Lambda \subsetneq V$ and $\tau \in \Omega_\Lambda$, 
the largest eigenvalue $\lambda_1(\MM_\mu^\tau)$ of the influence matrix $\MM_\mu^\tau$ satisfies
\[
\lambda_1 (\MM_\mu^\tau) \le \eta. 
\]
\end{definition}

The work of \cite{FGYZ20} defined another version of influence matrix by
\[
\Psi_\mu^\tau(u,v) = \max_{i,j \in \Omega_u^\tau} \TV{\mu( \sigma_v \seq \cdot \mid \sigma_u \seq i, \sL = \tau)}{\mu( \sigma_v \seq \cdot \mid \sigma_u \seq j, \sL = \tau)},
\]
and the spectral independence correspondingly. 
We remark that \cref{def:spectral-ind} is weaker than the notion of spectral independence given in \cite{FGYZ20}, and for all current applications as in \cite{CGSV20,FGYZ20} or here in this paper, both definitions work.

Our main result is that if the Gibbs distribution on a bounded-degree graph is both marginally bounded and spectrally independent, 
then the Glauber dynamics satisfies the modified log-Sobolev inequality with constant $\Omega(1/n)$ (see \cref{def:func-ineq}) and 
mixes in $O(n\log n)$ steps, where $n$ is the number of vertices of the graph.

\begin{theorem}\label{thm:main}
Let $\Delta \ge 3$ be an integer and $b,\eta >0$ be reals. 
Suppose that $G=(V,E)$ is an $n$-vertex graph of maximum degree at most $\Delta$ and $\mu$ is a totally-connected Gibbs distribution of some spin system on $G$. 
If $\mu$ is both $b$-marginally bounded and $\eta$-spectrally independent, 
then the Glauber dynamics for sampling from $\mu$ satisfies the modified log-Sobolev inequality with constant $\frac{1}{C_1 n}$ where 
\[
C_1 = \left( \frac{\Delta}{\bb} \right)^{O\left(\frac{\eta}{\bb^2}+1\right)}. 
\]
Furthermore, the mixing time of the Glauber dynamics satisfies  
\[
T_{\mathrm{mix}}(\Pgl, \eps) 
= \left( \frac{\Delta}{\bb} \right)^{O\left(\frac{\eta}{\bb^2}+1\right)} 
\times 
O\left( n\log\left(\frac{n}{\eps}\right) \right). 
\]
\end{theorem}
\begin{remark}
More specifically, when $n \ge \frac{24\Delta}{\bb^2}(\frac{4\eta}{\bb^2}+1)$ we can choose
\[
C_1 = \frac{18\log(1/b)}{b^4} \left( \frac{24\Delta}{\bb^2} \right)^{\frac{4\eta}{\bb^2}+1},
\]
and the mixing time is bounded by
\[
T_{\mathrm{mix}}(\Pgl, \eps) \le \ceil{ \frac{18\log(1/b)}{b^4} \left( \frac{24\Delta}{\bb^2} \right)^{\frac{4\eta}{\bb^2}+1} n \left( \log n + \log \log \frac{1}{b} + \log \frac{1}{2\eps^2} \right)}. 
\]
\end{remark}

Previous results \cite{ALO20, CLV20,CGSV20,FGYZ20} could obtain $\mathrm{poly}(\Delta) \times n^{O(\eta)}$ mixing but without the assumption of marginal boundedness. In the setting of spin systems, we always have $b$-marginal boundedness with $b$ depending only on the parameters $A,h$ of the spin system and the maximum degree $\Delta$ of the graph, and so our results supersede those of \cite{ALO20, CLV20,CGSV20,FGYZ20} in the bounded degree regime.

\begin{remark}
After the first version of this paper, the work \cite{BCCPSV21} reformulates the proof of \cref{thm:main} without using simplicial complexes; in particular, the constant $C_1$ is brought down to $C_1 = (\Delta/b)^{O((\eta/b) + 1)}$. The proof approach in this paper can also be modified to achieve the same bound, by considering \cref{lem:entvarrelate} specialized to simplicial complexes corresponding to spin systems. 
\end{remark}

\subsection{Result for General Simplicial Complexes}
\label{subsec:SC}
The recent work \cite{ALO20} studied spin systems, and more generally any distribution over $[q]^V$, in a novel way by viewing full and partial configurations as a high dimensional simplicial complex and utilizing tools such as high-dimensional expansion. 
Subsequent works \cite{CLV20,CGSV20,FGYZ20} follow the same path as well. 
In this paper we also study spin systems in the framework of simplicial complexes. 
Moreover, we obtain new bounds on the mixing time and modified log-Sobolev constant of the global down-up and up-down walks for arbitrary pure weighted simplicial complexes. 
Before presenting our results, we first review some standard notation.

A \textit{simplicial complex $\SC$} is a collection of subsets (called faces) of a ground set $\mathcal{U}$ which is downwards closed; 
that is, if $\sigma \in \SC$ and $\tau \subseteq \sigma$ then $\tau \in \SC$. 
The dimension of a face is its size,\footnote{This differs from the standard definition of ``dimension'' in the literature, which is the size of the subset minus $1$. In our setting, setting the dimension equal to the size of the set is more natural.} and 
the dimension of $\SC$ is defined to be the maximum dimension of its faces. 
We say an $n$-dimensional simplicial complex $\SC$ is \textit{pure} if every face is contained in a maximal face of size $n$.
We write $\SC(k)$ for the collection of faces of size $k$. 
For a $k$-dimensional face $\tau \in \SC(k)$, we can define a pure $(n-k)$-dimensional simplicial subcomplex $\SC_\tau$ by taking $\SC_\tau = \{\xi \subseteq \mathcal{U} \setminus \tau: \tau \cup \xi \in \SC \}$. 

For a pure $n$-dimensional simplicial complex $\SC$, consider a positive weight function $w:\SC(n) \rightarrow \R_{>0}$, which induces a distribution $\pi_{n}$ on $\SC(n)$ with $\pi_{n}(\sigma) \propto w(\sigma)$.
Furthermore, we can also define a distribution $\pi_{k}$ over $\SC(k)$ for each nonnegative integer $k < n$ via the following process: sample $\sigma$ from $\pi_{n}$, and select a uniformly random subset of size $k$. 
For $\tau \in \SC(k)$, the weight function $w$ induces the weights for the simplicial subcomplex $\SC_\tau$ by $w_\tau(\xi) = w(\tau \cup \xi)$ for each $\xi \in \SC_\tau(n-k)$. 
The distribution $\pi_{\tau, j}$ is also defined accordingly for each nonnegative integer $j \le n-k$.

As noticed in \cite{ALO20}, there is a natural way to represent every distribution $\mu$ over $[q]^V$ with $|V| = n$ as a pure $n$-dimensional weighted simplicial complex $(\SC = \SC^\Omega,\mu)$, which is defined as follows. 
The ground set of $\SC$ consists of pairs 
\[
\widetilde{V} = \{(v,i) : v \in V, i \in \Omega_v\}. 
\] 
The maximal faces of~$\SC$ consist of collections of $n$ pairs forming a valid configuration $\sigma \in \Omega$; i.e., every configuration $\sigma \in \Omega$ corresponds to a maximal face $\{(v,\sigma_v): v\in V\}$. 
The rest of $\SC$ is generated by taking downwards closure so that $\SC$ is pure by construction.
Namely, every $U \subseteq V$ and $\tau \in \Omega_U$ corresponds to a face $\{(v,\tau_v): v \in U\}$; we shall denote it by $(U,\tau)$ for simplicity. 
Note that the faces of intermediate dimension can be thought of as partial configurations. 
Now, if there is a weight function $w: \Omega \to \R_{>0}$ associated with $\mu$ such that $\mu(\sigma) \propto w(\sigma)$ for each $\sigma \in \Omega$, then it also gives a weight function $w: \SC(n) \to \R_{>0}$ by the one-to-one correspondence between $\Omega$ and $\SC(n)$, and thus induces the associated distribution $\pi_n$ on $\SC(n)$. 
Observe that $\pi_{n}$ is exactly the distribution~$\mu$. 
Moreover, for each $k < n$, the distribution $\pi_k$ on $\SC(k)$ is given by 
\[
\pi_k(U,\tau) = \frac{1}{\binom{n}{k}} \, \mu(\sigma_U \seq \tau)
\]
for every $U \subseteq V$ and $\tau \in \Omega_U$.

For simplicial complexes, the global down-up and up-down walks between faces of distinct dimensions have attracted a lot of attention in recent years \cite{KM17, DK17, KO18, Opp18, ALOV18ii, CGM19, AL20}.
For integers $0\le r < s \le n$, define the \emph{order-$(s,r)$ (global) down-up} walk with transition matrix denoted by $P^\vee_{s,r}$ to be the following random walk over $\SC(s)$: in each step we remove $s-r$ elements, chosen uniformly at random, from the current face $\sigma_t \in \SC(s)$ to obtain a face $\tau_t \in \SC(r)$, and then pick $\xi_{t+1} \in \SC_{\tau_t}(s-r)$ from the distribution $\pi_{\tau_t,s-r}$ and set $\sigma_{t+1} = \tau_t \cup \xi_{t+1}$. 
The stationary distribution of $P^\vee_{s,r}$ is $\pi_s$. 
In particular, observe that the Glauber dynamics for a distribution $\mu$ over $[q]^V$ is the same as the order-$(n,n-1)$ down-up walk for the weighted simplicial complex $(\SC,\mu)$. 
Similarly, the \emph{order-$(r,s)$ (global) up-down} walk with transition matrix $P^\wedge_{r,s}$ is a random walk over $\SC(r)$ with stationary distribution $\pi_r$: given the current face $\tau_t \in \SC(r)$, sample $\xi_{t+1} \in \SC_{\tau_t}(s-r)$ from $\pi_{\tau_t,s-r}$, set $\sigma_{t+1} = \tau_t \cup \xi_{t+1}$, and finally remove $s-r$ elements from $\sigma_{t+1}$ uniformly at random to obtain $\tau_{t+1} \in \SC(r)$. 


We establish the modified log-Sobolev inequality and give meaningful bounds on the mixing time for the down-up and up-down walks for arbitrary weighted simplicial complexes. 
Our proof utilizes the local-to-global scheme as in \cite{AL20} and establishes contraction of entropy extending the result of \cite{CGM19}. 
Before stating our main result, we first give the definitions of marginal boundedness and local spectral expansion for simplicial complexes. 
As we shall see from \cref{claim:bounded} and \cref{claim:spec-ind} below, our requirements of marginal boundedness and spectral independence in \cref{thm:main} is translated from the corresponding conditions needed for simplicial complexes.


\begin{definition}[Bounded Marginals]
We say a pure $n$-dimensional weighted simplicial complex $(\SC,w)$ is \emph{$(\bb_0,\dots, \bb_{n-1})$-marginally bounded} if for all $0\le k\le n-1$, every $\tau \in \SC(k)$, and every $i\in \SC_{\tau}(1)$, 
we have
\[
\pi_{\tau,1}(i) \ge \bb_k.
\]
\end{definition}



\begin{claim}
\label{claim:bounded}
If a distribution $\mu$ over $[q]^V$ is $\bb$-marginally bounded, then the pure weighted simplicial complex $(\SC,\mu)$ is $(\bb_0,\dots, \bb_{n-1})$-marginally bounded with 
$\bb_k = \frac{\bb}{n-k}$ for each $k$. 
\end{claim}

The proof of \cref{claim:bounded} can be found in \cref{subsubsec:proof-claim}. 


The global walks in simplicial complexes can be studied by decomposition into local walks which we define now.
For every $0 \le k \le n-2$ and every face $\tau \in \SC(k)$, the \emph{local walk} at $\tau$ with transition matrix $P_\tau$ is the following (non-lazy) random walk over $\SC_\tau(1)$: given the current element $i \in \SC_\tau(1)$, the next element is generated from the distribution $\pi_{\tau \cup \{i\} ,1}$. 
One can relate mixing properties of the local walks to the mixing properties of the global walks; see \cite{KO18, ALOV18ii, CGM19, AL20}. 
In nearly all prior works, such a relation was quantified using the spectral gap of the walks. 
Like in \cite{CGM19}, while our ultimate goal is to show the modified log-Sobolev inequality of the global walks, we will still need the notion of local spectral expansion for local walks. 
Let us now capture this idea using the following definition, taking after \cite{KM17, DK17, Opp18, KO18, AL20, KM20}.

\begin{definition}[Local Spectral Expansion \cite{AL20}]
We say a pure weighted $n$-dimensional simplicial complex $(\SC,w)$ is a \emph{$(\zeta_{0},\dots,\zeta_{n-2})$-local spectral expander} if for every $0 \leq k \leq n-2$ and every $\tau \in \SC(k)$, we have 
\[
\lambda_{2}(P_{\tau}) \leq \zeta_{k}. 
\]
\end{definition}

\begin{claim}
\label{claim:spec-ind}
If a distribution $\mu$ over $[q]^V$ is $\eta$-spectrally independent, then the weighted simplicial complex $(\SC,\mu)$ is a $(\zeta_{0},\dots,\zeta_{n-2})$-local spectral expander with 
$\zeta_k = \frac{\eta}{n-k-1}$ for each $k$. 
\end{claim}

\begin{proof}
This is Theorem 8 from \cite{CGSV20}. 
\end{proof}


We then show that for any pure weighted simplicial complexes, the modified log-Sobolev inequality (see \cref{def:func-ineq}) holds for down-up and up-down walks if the marginal probabilities of the simplicial complex are bounded away from zero and all local walks have good expansion properties. 
This also bounds the mixing times of these random walks. 

\begin{theorem}\label{thm:down-up-mixing}
Let $(\SC,w)$ be a pure $n$-dimensional weighted simplicial complex. 
If $(\SC,w)$ is $(\bb_0,\dots,\bb_{n-1})$-marginally bounded and has $(\zeta_0,\dots,\zeta_{n-2})$-local spectral expansion, 
then for every $0 \le r < s \le n$, 
both the order-$(s,r)$ down-up walk and the order-$(r,s)$ up-down walk satisfy the modified log-Sobolev inequality with constant $\kappa = \kappa (r,s)$ defined as
\[
\kappa = \frac{\sum_{k=r}^{s-1} \Gamma_k}{\sum_{k=0}^{s-1} \Gamma_k}
\]
where: $\Gamma_0 = 1$; for $1\le k \le s-1$, $\Gamma_k = \prod_{j = 0}^{k-1} \alpha_j$; and for $0 \le k \le s-2$, 
\[
\alpha_k = \max\left\{ 
1-\frac{4\zeta_k}{\bb_k^2(s-k)^2} , 
\frac{1-\zeta_k}{4+2\log(\frac{1}{2b_k b_{k+1}})}
\right\}.
\]
Furthermore, the mixing time of the order-$(s,r)$ down-up walk is bounded by
\begin{equation}\label{eqn:mix-down-up}
T_{\mathrm{mix}}(P^\vee_{s,r}, \eps) \le \ceil{\frac{1}{\kappa} \left( \log \log \frac{1}{\pi_s^*} + \log \frac{1}{2\eps^2} \right)}
\end{equation}
where $\pi_s^* = \min_{\sigma \in \SC(s)} \pi_s(\sigma)$. 
The mixing time of the order-$(r,s)$ up-down walk is also bounded by \cref{eqn:mix-down-up} with $\pi_s^*$ replaced by $\pi_r^*$. 
\end{theorem}

\cref{thm:down-up-mixing} generalizes both the result of \cite{CGM19} for simplicial complexes with respect to strongly log-concave distributions and the result of \cite{AL20} for the Poincar\'{e} inequality (i.e., bounding the spectral gap). 
It in some sense answers a question of \cite{CGM19} on local-to-global modified log-Sobolev inequalities in high-dimensional expanders, at least in the bounded marginals setting.

Even though \cref{thm:down-up-mixing} can give a bound on the mixing time of the Glauber dynamics, which is the order-$(n,n-1)$ down-up walk in the corresponding weighted simplicial complex, our main result \cref{thm:main} does \emph{not} follow directly from \cref{thm:down-up-mixing}. 
In fact, we will only consider the order-$(n,n-\ell)$ down-up walk for $\ell = \Theta(n)$, which corresponds to the heat-bath block dynamics that updates a uniformly random subset of $\ell$ vertices in every step. 
One of our main technical contributions is to compare this block dynamics with the single-site Glauber dynamics; we shall detail this in \cref{subsec:AT-UBF} below.
Nevertheless, we find \cref{thm:down-up-mixing} interesting of its own and possible for future applications in other problems.

{
In the following section, we formally define approximate tensorization of entropy (see \cref{def:at}) and the more general notion of uniform block factorization of entropy (see \cref{def:ubf}), and then provide a proof outline for our main results.  An outline of the paper is provided in~\cref{sec:paper-outline}.

}

\section{Proof Outline}\label{sec:outline}
In this section, we outline our proofs of \cref{thm:main,thm:down-up-mixing}. 

\subsection{Approximate Tensorization and Uniform Block Factorization}
\label{subsec:AT-UBF}

One way of establishing rapid mixing of the Glauber dynamics is to show that the Gibbs distribution satisfies the \emph{approximate tensorization of entropy}. 
This approach has been (implicitly) used in much of the literature to establish the log-Sobolev inequalities, from which one can deduce an optimal bound on the mixing time. 
Before giving the formal definition, we first review some standard definitions. 
 
Consider a distribution $\mu$ supported on $\Omega \subseteq [q]^V$. 
For every $f: \Omega \to \R_{\ge 0}$, we denote the expectation of $f$ under $\mu$ by $\mu(f) = \sum_{\sigma \in \Omega} \mu(\sigma) f(\sigma)$ and the entropy of $f$ by $\Ent_{\mu}(f) = \mu(f \log \frac{f}{\mu(f)})$. 
We often simply write $\Ent(f)$ for the entropy and drop the subscript $\mu$ when it is clear from the context. 
More generally, given $S \subseteq V$ and $\tau \in \Omega_{V \setminus S}$, for every $f: \Omega_S^\tau \to \R_{\ge 0}$ we use $\mu_S^\tau(f)$ to denote the expectation of $f$ under the conditional distribution $\mu_S^\tau$ and $\Ent_S^\tau(f)$ for the corresponding entropy. 
For most of the time we are actually given a function $f: \Omega \to \R_{\ge 0}$, and we will still write $\mu_S^\tau(f)$ and $\Ent_S^\tau(f)$ where we think of $f$ as restricted to the space $\Omega_S^\tau$ and implicitly assume that the configuration outside $S$ is given by $\tau$; i.e., for an argument $\sigma \in \Omega_S^\tau$ the value of $f$ is $f(\sigma\cup \tau)$. 
It is helpful to think of $\mu_S^\tau(f)$ and $\Ent_S^\tau(f)$ as a function of the boundary condition $\tau$. 
In this sense, the notation $\mu[\Ent_S(f)]$, for example, represents the expectation of the function $\Ent_S^\tau(f)$ where $\tau \in \Omega_{V \setminus S}$ is distributed as the marginal of $\mu$ on $V \setminus S$. 

The notion of approximate tensorization of entropy is formally defined as follows. 

\begin{definition}[Approximate Tensorization]
\label{def:at}
We say that a distribution $\mu$ over $[q]^V$ satisfies the \emph{approximate tensorization of entropy} (with constant $\Cat$) if for all $f: \Omega \to \R_{\ge 0}$ we have 
\begin{equation}\label{eq:approx-tensor}
\Ent (f) \le \Cat \sum_{v\in V} \mu[\Ent_v(f)].
\end{equation}
\end{definition}

Approximate tensorization can be understood as closeness of $\mu$ to a product distribution, or weak dependency of variables. 
In fact, if $\mu$ is exactly a product distribution (e.g., the Gibbs distribution on an empty graph), then approximate tensorization holds with constant $\Cat = 1$; e.g., see~\cite{Cesi01,CMT15}.
If $\mu$ satisfies approximate tensorization with a constant $C_1$ independent of $n$, then the Glauber dynamics for sampling from $\mu$ mixes in $O(n\log n)$ steps. 
In fact, given approximate tensorization, one can deduce tight bounds on all of the following quantities: the spectral gap, both standard and modified log-Sobolev constants, relative entropy decay rate, mixing time, and concentration bounds. See \cref{fact:at-mixing} for a detailed summary. 

In many cases, especially on the integer lattice $\Z^d$, log-Sobolev inequalities for the Glauber dynamics are established through the approximate tensorization of entropy, which is more intuitive and easier to handle; e.g., see~\cite{Martinelli-notes,GZ-notes,Cesi01,CP20}.
Despite the success on $\Z^d$, there is not much study for spin systems on bounded-degree graphs. 
The works of \cite{CMT15,Marton15} considered approximate tensorization for general discrete product spaces, and gave sufficient conditions to derive it; however, for spin systems these results do not cover the whole uniqueness region.

One can regard approximate tensorization of entropy as factorizing entropy into all single vertices. 
Motivated by tools from high dimensional simplicial complexes \cite{AL20,ALO20} and study on general block factorization of entropy \cite{CP20}, we consider in this paper a more general notion of entropy factorization, where the entropy is factorized into subsets of vertices of a fixed size. 
The formal definition is given as follows. 

\begin{definition}[Uniform Block Factorization]
\label{def:ubf}
We say that a distribution $\mu$ over $[q]^V$ satisfies the \emph{$\ell$-uniform block factorization of entropy} (with constant $C$) if for all $f: \Omega \to \R_{\ge 0}$ we have 
\begin{equation}\label{eq:block-factor}
\frac{\ell}{n} \, \Ent (f) \le C \cdot \frac{1}{\binom{n}{\ell}} \sum_{S\in \binom{V}{\ell}} \mu[\Ent_S(f)].
\end{equation}
\end{definition}

We remark that uniform block factorization of entropy is a special case of block factorization given by equation~(1.3) in \cite{CP20}; there, the entropy factorizes into arbitrary blocks with arbitrary weights. 
Also observe that $1$-uniform block factorization is the same as approximate tensorization of entropy. 
Just as the approximate tensorization corresponds to the single-site Glauber dynamics, 
the $\ell$-uniform block factorization corresponds to the heat-bath block dynamics where in each step a subset of vertices of size $\ell$ is chosen uniformly at random and gets updated. 
Moreover, similar results as in \cref{fact:at-mixing} can be deduced for this block dynamics. 

Our first key result is a reduction from approximate tensorization to uniform block factorization. 
For $\bb$-marginally bounded Gibbs distributions on graphs with maximum degree $\leq \Delta$, 
we show that approximate tensorization is implied by $\ell$-uniform block factorization for $\ell = \ceil{\theta n}$ and an appropriate constant $\theta$ depending on $\bb$ and $\Delta$. 
This is given by the following lemma. 

\begin{lemma}\label{lem:comparison}
Let $\Delta\ge 3$ be an integer and $\bb > 0$ be a real. 
Consider the Gibbs distribution $\mu$ on an $n$-vertex graph $G$ of maximum degree at most $\Delta$ and assume that $\mu$ is $\bb$-marginally bounded. 
Suppose there exist positive reals $\theta \le \frac{b^2}{12\Delta}$ and $C$ such that 
$\mu$ satisfies the $\ceil{\theta n}$-uniform block factorization of entropy with constant $C$. 
Then $\mu$ satisfies the approximate tensorization of entropy with constant 
\[
\Cat =  \frac{18\log(1/b)}{b^4}\, C. 
\] 
\end{lemma}

This lemma is proved in \cref{sec:approxtensoruniformblock}, it follows from a type of ``shattering'' result which is stated in~\cref{lem:probsizecomponent}.

\begin{remark}
The notion of approximate tensorization and uniform block factorization with respect to variance is also meaningful. 
In fact, for variance these definitions are equivalent to bounding the spectral gap of the corresponding chains. 
Moreover, \cref{lem:comparison} holds for variance as well, which can already provide a tight bound on the spectral gap of the Glauber dynamics combining results from \cite{ALO20,CLV20,CGSV20,FGYZ20}. We will discuss this in more detail in \cref{app:var}. 
\end{remark}



\subsection{Simplicial Complexes and Entropy Contraction}
\label{sec:entropy-outline}

Our next goal is to establish $\ell$-uniform block factorization of entropy for $\ell = \Theta(n)$, which relies on the spectral independence property.  
The following lemma holds for all distributions over $[q]^V$, not only Gibbs distributions. 

\begin{lemma}\label{lem:down-up-k}
Let $\bb,\eta >0$ be reals. 
Then for every real $\theta \in (0,1)$ and every integer $n \ge \frac{2}{\theta}(\frac{4\eta}{\bb^2}+1)$ the following holds. 

Let $V$ be a set of size $n$ and $\mu$ be a distribution over $[q]^V$. 
If $\mu$ is both $\bb$-marginally bounded and $\eta$-spectrally independent, 
then $\mu$ satisfies $\ceil{\theta n}$-uniform block factorization of entropy with constant 
\[
C = \left( \frac{2}{\theta} \right)^{\frac{4\eta}{\bb^2}+1}. 
\]
\end{lemma}

Recall that there is a natural correspondence between a distribution $\mu$ over $[q]^V$ and the weighted simplicial complex $(\SC,\mu)$. 
For general weighted simplicial complexes, one property studied in \cite{CGM19} is how the entropy of a function defined on faces contracts when it projects down from higher dimensions to lower. 
This can be captured by the definition below. 
For a pure $n$-dimensional weighted simplicial complex $(\SC,w)$ and a nonnegative integer $k < n$, let $P^\uparrow_k$ denote the $|\SC(k)| \times |\SC(k+1)|$ dimensional transition matrix corresponding to adding a random element $i \notin \tau$ to some $\tau \in \SC(k)$ where $i$ is distributed as $\pi_{\tau,1}$. 
Also for any $0\le r < s \le n$ and any function $f^{(s)}:\SC(s) \rightarrow \R_{\geq0}$, define $f^{(r)}: \SC(r) \rightarrow \R_{\geq0}$ by $f^{(r)} = P^{\uparrow}_{r} \cdots P^{\uparrow}_{s-1} f^{(s)}$.

\begin{definition}[Global Entropy Contraction]
We say a pure $n$-dimensional weighted simplicial complex $(\SC,w)$ satisfies the \emph{order-$(r,s)$ global entropy contraction} with rate $\kappa = \kappa(r,s)$ if for all $f^{(s)} : \SC(s) \rightarrow \R_{\geq0}$ we have
\[
\Ent_{\pi_r}(f^{(r)}) \le (1-\kappa)\, \Ent_{\pi_s}(f^{(s)}).
\]
\end{definition}

It turns out, 
as a remarkable fact, 
that uniform block factorization of entropy for a distribution $\mu$ over $[q]^V$ is \emph{equivalent} to global entropy contraction for the weighted simplicial complex $(\SC,\mu)$. 

\begin{lemma}\label{lem:equiv}
A distribution $\mu$ over $[q]^V$ satisfies the $\ell$-uniform block factorization of entropy with some constant $C$ 
if and only if
the corresponding weighted simplicial complex $(\SC,\mu)$ satisfies order-$(n-\ell,n)$ global entropy contraction with rate $\kappa$, where 
$C \kappa = \ell/n$. 
\end{lemma}

The proof of \cref{lem:equiv} can be found in \cref{subsubsec:proof-claim}. 
As a consequence, to prove \cref{lem:down-up-k}, it suffices to establish global entropy contraction for the weighted simplicial complex $(\SC,\mu)$. 

Just like approximate tensorization and uniform block factorization having many implications for the corresponding single-site and block dynamics (e.g., see \cref{fact:at-mixing}), the notion of global entropy contraction can provide for weighted simplicial complexes meaningful bounds on  
the spectral gap, modified log-Sobolev constant, relative entropy decay rate, mixing time, and concentration bounds; 
see \cref{fact:ent-contraction} for details. 
In Lemma 11 of \cite{CGM19}, the authors established order-$(r,s)$ global entropy contraction with rate $\kappa = \frac{s-r}{s}$ for simplicial complexes with respect to homogeneous strongly log-concave distributions. 
From this, they deduced the modified log-Sobolev inequality for the down-up and up-down walks and showed rapid mixing of it.

We then show that for an arbitrary weighted simplicial complex $(\SC,w)$, one can deduce global entropy contraction from local spectral expansion whenever the marginals of the induced distributions are nicely bounded. For this, we prove a local-to-global result for entropy contraction in the spirit of \cite{AL20}. If we additionally know that the marginals are nicely bounded, we can further reduce the local entropy contraction to local spectral expansion. 


\begin{lemma}\label{lem:ent-contraction}
Let $(\SC,w)$ be a pure $n$-dimensional weighted simplicial complex. 
Suppose that $(\SC,w)$ is $(\bb_0,\dots,\bb_{n-1})$-marginally bounded and has $(\zeta_0,\dots,\zeta_{n-2})$-local spectral expansion. 
Then for all $0 \le r < s \le n$, 
$(\SC,w)$ satisfies order-$(r,s)$ global entropy contraction with rate $\kappa = \kappa(r,s)$ given as in \cref{thm:down-up-mixing}. 
\end{lemma}

\cref{thm:down-up-mixing} follows immediately from \cref{lem:ent-contraction} and \cref{fact:ent-contraction}. 
We remark that \cref{lem:ent-contraction} recovers Lemma 11 of \cite{CGM19} for simplicial complexes corresponding to discrete log-concave distributions, 
since there one has $\zeta_k = 0$ for all $k$ as shown in \cite{ALOV18ii}.




We present next the proof of \cref{lem:down-up-k}, which follows directly from \cref{lem:equiv,lem:ent-contraction}.

\begin{proof}[Proof of \cref{lem:down-up-k}]
From \cref{claim:bounded} and \cref{claim:spec-ind} we know that the weighted simplicial complex $(\SC,\mu)$ corresponding to $\mu$ is $(\bb_0,\dots,\bb_{n-1})$-marginally bounded with $\bb_k = \frac{b}{n-k}$ and has $(\zeta_0,\dots,\zeta_{n-2})$-local spectral expansion with $\zeta_k = \frac{\eta}{n-k-1}$. 
Then, \cref{lem:ent-contraction} implies that $(\SC,\mu)$ satisfies order-$(n-\ell,n)$ global entropy contraction for $\ell = \ceil{\theta n}$ with rate
\[
\kappa = \frac{\sum_{k=n-\ell}^{n-1} \Gamma_k}{\sum_{k=0}^{n-1} \Gamma_k}
\]
where $\Gamma_0 = 1$, $\Gamma_k = \prod_{j = 0}^{k-1} \alpha_j$, and $$\alpha_k = \max\left\{1-\frac{4\eta}{b^2(n-k-1)} , \frac{1-\eta/(n-k-1)}{4+2\log((n-k)(n-k-1)/(2b^2))} \right\}.$$ 
Define an integer $R = \ceil{\frac{4\eta}{b^2}}$ and observe that $n \ge \ell \ge \theta n \ge 2R$ by our assumption that $n \ge \frac{2}{\theta}(\frac{4\eta}{\bb^2}+1)$. 
Thus, we have
\[
\alpha_k \ge \hat{\alpha}_k := \max\left\{1-\frac{R}{n-k-1} , 0 \right\}. 
\]
The use of $\hat{\alpha}_k$ instead of $\alpha_k$ is only to simplify the calculation later, where we get a telescoping product with integral $R$. 
Notice that $\kappa$, when viewed as a function of $\alpha_k$'s, is monotone increasing with each $\alpha_k$. 
Thus, to lower bound $\kappa$, we can plug in the lower bounds $\hat{\alpha}_k$'s and get
\[
\kappa \ge \frac{\sum_{k=n-\ell}^{n-1} \hat{\Gamma}_k}{\sum_{k=0}^{n-1} \hat{\Gamma}_k}
\]
where $\hat{\Gamma}_0 = 1$ and $\hat{\Gamma}_k = \prod_{j = 0}^{k-1} \hat{\alpha}_j$ for each $k \ge 1$. 
We will show that for every $0\le k \le n-1$ one actually has
\begin{equation}\label{eq:Gamma-hat}
\hat{\Gamma}_k = \frac{(n-k-1)(n-k-2) \cdots (n-k-R)}{(n-1)(n-2) \cdots (n-R)}. 
\end{equation}
For $k = 0$ we have $\hat{\Gamma}_0 = 1$ and \cref{eq:Gamma-hat} holds. 
For $1 \le j \le n-R-2$ we have
\[
\hat{\alpha}_j = \max\left\{\frac{n-j-1-R}{n-j-1} , 0 \right\} = \frac{n-j-1-R}{n-j-1}
\]
and thus for $1\le k \le n-R-1$
\[
\hat{\Gamma}_k = \prod_{j = 0}^{k-1}\frac{n-j-1-R}{n-j-1} = \frac{(n-k-1)(n-k-2) \cdots (n-k-R)}{(n-1)(n-2) \cdots (n-R)} .
\]
Finally, since $\hat{\alpha}_j = 0$ when $n-R-1 \le j \le n-2$, we have $\hat{\Gamma}_k = 0$ for $n-R \le k \le n-1$. 
Therefore, \cref{eq:Gamma-hat} is true for all $k$. 
It then follows that
\begin{align*}
\kappa &\ge \frac{\sum_{k=n-\ell}^{n-1} (n-k-1)(n-k-2) \cdots (n-k-R)}{\sum_{k=0}^{n-1} (n-k-1)(n-k-2) \cdots (n-k-R)}\\
&= \frac{\sum_{j=0}^{\ell-1} j(j-1) \cdots (j-R+1)}{\sum_{j=0}^{n-1} j(j-1) \cdots (j-R+1)}. 
\end{align*}
The following is a standard equality which can be proved by induction:
\[
\sum_{j=0}^{N-1} j(j-1) \cdots (j-R+1) = \frac{1}{R+1} N(N-1) \cdots (N-R). 
\]
Hence, we obtain
\[
\kappa \ge \frac{\ell(\ell-1) \cdots (\ell-R)}{n(n-1) \cdots (n-R)}. 
\]
Finally, we deduce from \cref{lem:equiv} that
\[
C \le \frac{\ell}{n} \cdot \frac{1}{\kappa} 
\le \frac{(n-1) \cdots (n-R)}{(\ell-1) \cdots (\ell-R)} 
\le \left( \frac{n-R}{\ell-R} \right)^R 
\le \left( \frac{2n}{\ell} \right)^R 
\le \left( \frac{2}{\theta} \right)^{\frac{4\eta}{b^2} + 1}
\]
where we use our assumption $\ell \ge \theta n \ge 2R$. 
\end{proof}

\subsection{Wrapping up}

Combining \cref{lem:down-up-k,lem:comparison}, we establish approximate tensorization of entropy with a constant independent of $n$, when the Gibbs distribution is marginally bounded and spectrally independent. 
This is stated in the following theorem. 

\begin{theorem}\label{thm:at}
Let $\Delta \ge 3$ be an integer and $b,\eta >0$ be reals. 
Suppose that $G=(V,E)$ is an $n$-vertex graph of maximum degree at most $\Delta$ and $\mu$ is a totally-connected Gibbs distribution of some spin system on $G$. 
If $\mu$ is both $b$-marginally bounded and $\eta$-spectrally independent and $n \ge \frac{24\Delta}{\bb^2}(\frac{4\eta}{\bb^2}+1)$, 
then $\mu$ satisfies the approximate tensorization of entropy with constant
\[
C_1 = \frac{18\log(1/b)}{b^4} \left( \frac{24\Delta}{\bb^2} \right)^{\frac{4\eta}{\bb^2}+1}.
\]
\end{theorem}

\cref{thm:main} then follows immediately from \cref{thm:at,fact:at-mixing}.

Our main results \cref{thm:2-spin,thm:hard-core,thm:ising,thm:coloring,thm:matching} will follow from \cref{thm:main} by establishing marginal boundedness and spectral independence for each model.  The detailed proofs
are contained in \cref{app:proof-main}, we include here a brief sketch.
The marginal boundedness is a trivial bound. 
The spectral independence was previously established for antiferromagnetic 2-spin systems including the hard-core model and the Ising model in the whole uniqueness region \cite{ALO20,CLV20}, and for random $q$-colorings when $q$ is sufficiently large \cite{CGSV20,FGYZ20}. 
For the monomer-dimer model, spectral independence is not known previously. Following the proof strategy of \cite{CLV20} and utilizing the two-step recursion from \cite{BGKNT07}, we show the following. 

\begin{theorem}\label{thm:spec-ind-matching}
Let $\Delta \ge 3$ be an integer and $\lambda > 0$ be a real. 
Then for every graph $G = (V,E)$ of maximum degree at most $\Delta$ with $m = |E|$, the Gibbs distribution $\mu$ of the monomer-dimer model on $G$ with fugacity $\lambda$ is $\eta$-spectrally independent for $\eta = \min\wrapc{2\lambda\Delta, O(\sqrt{\lambda\Delta})}$.
\end{theorem}

\subsection{Outline of Paper}
\label{sec:paper-outline}

The remainder of the paper is organized as follows. 
In \cref{sec:preliminaries}, we collect relevant preliminaries. 
In \cref{sec:approxtensoruniformblock}, we show how to reduce approximate tensorization to uniform block factorization with linear-sized blocks; specifically, we prove \cref{lem:comparison}. 
In \cref{sec:Proof-ent-contraction}, we reduce uniform block factorization and, more generally, global entropy contraction in the setting of weighted simplicial complexes to local entropy contraction; 
we then further reduce local entropy contraction to local spectral expansion when the simplicial complexes have bounded marginals and thus prove \cref{lem:ent-contraction}. 
In \cref{sec:monomerdimerspecind}, we bound the spectral independence of the monomer-dimer model on bounded degree graphs and prove \cref{thm:spec-ind-matching}. We finish off the proofs of our main mixing time results in \cref{app:proof-main}. Finally, we conclude with some open problems in \cref{sec:conclusion}, and discuss analogous results for variance in \cref{app:var}.

\section{Preliminaries}
\label{sec:preliminaries}

In this section we review some standard definitions.

In the following definition, we assume the underlying distribution $\mu$ is fixed and omit it from the subscript. 

\begin{definition}
Let $\Omega$ be a finite set and $\mu$ be a distribution over $\Omega$. 
For all functions $f,g: \Omega \to \R$: 
\begin{enumerate}[(a)]
\item The \emph{expectation} of $f$ is defined as $\mu(f) = \sum_{x \in \Omega} \mu(x) f(x)$;
\item The \emph{variance} of $f$ is defined as $\Var(f) = \mu[(f-\mu(f))^2] = \mu(f^2) - \mu(f)^2$;
\item The \emph{covariance} of $f$ and $g$ is defined as $\Cov(f,g) = \mu[(f-\mu(f))(g-\mu(g))] = \mu(fg) - \mu(f)\mu(g)$;
\item If $f \ge 0$, the \emph{entropy} of $f$ is defined as $\Ent(f) = \mu\left[f \log\left( \frac{f}{\mu(f)}\right)\right] = \mu(f\log f) - \mu(f) \log \mu(f)$ with the convention that $0\log 0 = 0$.
\end{enumerate}
\end{definition}


For two distributions $\mu, \nu$ over a finite set $\Omega$, the Kullback–Leibler divergence (KL divergence), also called relative entropy, is defined as
\[
\kl{\nu}{\mu} = \sum_{x \in \Omega} \nu(x) \log \left( \frac{\nu(x)}{\mu(x)} \right).
\]
Let $f = \nu/\mu$ be the relative density of $\nu$ with respect to $\mu$; i.e., $f(x) = \nu(x)/\mu(x)$ for all $x \in \Omega$. 
Then $\Ent(f) = \kl{\nu}{\mu}$. 
The following is a well-known fact; see, e.g., \cite{Duchi19}. 

\begin{fact}[Donsker-Varadhan's Variational Representation]\label{lem:donskervaradhan}
For two distributions $\mu, \nu$ over a finite set $\Omega$, the KL divergence admits the following variational formula:
\begin{align*}
    \kl{\nu}{\mu} = \sup_{f :\Omega \to \R} \wrapc{ \nu(f) - \log \mu(e^f)}. 
\end{align*}
\end{fact}

We then review some standard functional inequalities, and refer to \cite{BT06,MT06} for more backgrounds. 

\begin{definition}\label{def:func-ineq}
Let $\Omega$ be a finite set and $\mu$ be a distribution over $\Omega$. 
Let $P$ denote the transition matrix of an ergodic, reversible Markov chain on $\Omega$ with stationary distribution $\mu$. 
\begin{enumerate}[(a)]
\item The Dirichlet form of $P$ is defined as for every $f, g: \Omega \to \R$, 
\[
\EE_P(f,g) = \frac{1}{2} \sum_{x,y \in \Omega} \mu(x) P(x, y) (f(x) - f(y)) (g(x) - g(y)). 
\]
In particular, if $\Omega \subseteq [q]^V$ and $P = \Pgl$ is the Glauber dynamics for $\mu$, then we can write
\[
\EE_{\Pgl}(f,g) = \frac{1}{n} \sum_{v \in V} \mu[\Cov_v(f,g)],
\]
see, e.g., \cite{MSW03,CMT15}.

\item We say the \emph{Poincar\'{e} inequality} holds with constant $\lambda$ if for every $f: \Omega \to \R$, 
\[
\lambda \, \Var(f) \le \EE_{P}(f, f).
\]
The spectral gap of $P$ is $$\lambda(P) = \inf \Big\{ \frac{\EE_{P}(f, f)}{\Var(f)} \Bigm| f: \Omega \to \R, \Var(f) \neq 0 \Big\}.$$

\item We say the \emph{standard log-Sobolev inequality} holds with constant $\rho$ if for every $f: \Omega \to \R_{\ge 0}$, 
\[
\rho \, \Ent(f) \le \EE_P(\sqrt{f}, \sqrt{f}). 
\]
The standard log-Sobolev constant of $P$ is $$\rho(P) = \inf \Big\{ \frac{\EE_P(\sqrt{f}, \sqrt{f})}{\Ent(f)} \Big| f: \Omega \to \R_{\ge 0}, \Ent(f) \neq 0 \Big\}.$$

\item We say the \emph{modified log-Sobolev inequality} holds with constant $\rho_0$ if for every $f: \Omega \to \R_{\ge 0}$, 
\[
\rho_0 \, \Ent(f) \le \EE_P(f, \log f). 
\]
The modified log-Sobolev constant of $P$ is $$\rho_0(P) = \inf \Big\{ \frac{\EE_P(f, \log f)}{\Ent(f)} \Bigm| f: \Omega \to \R_{\ge 0}, \Ent(f) \neq 0 \Big\}.$$

\item We say \emph{the relative entropy decays} with rate $\alpha$ if for every distribution $\nu$ over $\Omega$, 
\[
\kl{\nu P}{\mu} \le \left( 1-\alpha \right) \kl{\nu}{\mu}.
\]
\end{enumerate}
\end{definition}

Next, we consider the case that $\Omega \subseteq [q]^V$ for a finite set $V$. 
Let $S \subseteq V$ and $\tau \in \Omega_{V \setminus S}$. 
Recall that for every function $f: \Omega \to \R_{\ge 0}$, we write $\mu_S^\tau(f)$ and $\Ent_S^\tau(f) = \Ent_{\mu_S^\tau}(f)$ for the expectation and entropy of $f$ under the conditional distribution $\mu_S^\tau(\cdot) = \mu(\sigma_S \seq \cdot \mid \sigma_{V \setminus S} = \tau)$, where $f = f_\tau$ is understood as a function of the configuration on $S$ with $\tau$ fixed outside $S$. 
We think of $\mu_S^\tau(f)$ and $\Ent_S^\tau(f)$ as functions of $\tau$, and we will use, for example, $\Ent[\mu_S(f)]$ to represent the entropy of $\mu_S^\tau(f)$ under $\mu$, and $\mu[\Ent_S(f)]$ for the expectation of $\Ent_S^\tau(f)$. 
We give below a useful property of the expectation and entropy; see, e.g., \cite{MSW03} for proofs. 

\begin{fact}
\label{fact:ent-decomp}
Let $S \subseteq V$ and $\tau \in \Omega_{V \setminus S}$. 
For every function $f: \Omega \to \R_{\ge 0}$, we have
\[
\mu(f) = \mu[\mu_S(f)] \qquad\text{and}\qquad \Ent(f) = \mu[\Ent_S(f)] + \Ent[\mu_S(f)]. 
\] 
\end{fact}

\subsection{Implications of Approximate Tensorization}

We summarize here a few corollaries of approximate tensorization of entropy for arbitrary distributions over discrete product spaces. 

\begin{fact}\label{fact:at-mixing}
Let $V$ be a set of size $n$ and $\mu$ be a distribution over $[q]^V$. 
If $\mu$ satisfies the approximate tensorization of entropy with constant $C_1$, then the Glauber dynamics for $\mu$ satisfies all of the following:
\begin{enumerate}[(1)]
\item \label{item:poincare} 
The Poincar\'{e} inequality holds with constant $\lambda = \frac{1}{C_1 n}$; 

\item \label{item:MLSI} 
The modified log-Sobolev inequality holds with constant $\rho_0 = \frac{1}{C_1 n}$; 

\item \label{item:RED} 
The relative entropy decays with rate $\alpha = \frac{1}{C_1 n}$; 

\item \label{item:mixing} 
The mixing time of the Glauber dynamics satisfies
\[
T_{\mathrm{mix}}(\Pgl, \eps) \le \ceil{C_1 n \left( \log \log \frac{1}{\mu_{\min}} + \log \frac{1}{2\eps^2} \right)}
\]
where $\mu_{\min} = \min_{\sigma \in \Omega} \mu(\sigma)$; 
If furthermore $\mu$ is $\bb$-marginally bounded, then we have $\mu_{\min} \ge b^n$ and thus 
\[
T_{\mathrm{mix}}(\Pgl, \eps) \le \ceil{C_1 n \left( \log n + \log \log \frac{1}{b} + \log \frac{1}{2\eps^2} \right)};
\]

\item \label{item:concentration}
For every $f:\Omega \to \R$ which is $c$-Lipschitz with respect to the Hamming distance on $[q]^{V}$ and every $a \ge 0$, we have the concentration inequality
\[
	\Pr_{\sigma \sim \mu} \wrapb{ \left| f(\sigma) - \mu(f) \right| \ge a } \le 2 \, \exp\wrapp{-\frac{a^2}{2c^{2} C_{1}n}}; 
\]

\item \label{item:SLSI} 
If furthermore $\mu$ is $\bb$-marginally bounded, then the standard log-Sobolev inequality holds with constant $\rho = \frac{1-2b}{\log(1/b-1)} \cdot \frac{1}{C_1 n}$ when $\bb < \frac{1}{2}$, or $\rho = \frac{1}{2C_1 n}$ when $\bb = \frac{1}{2}$. 
\end{enumerate}
\end{fact}


The $\ell$-uniform block factorization of entropy implies similar results for the heat-bath block dynamics that updates a random subset of vertices of size $\ell$ in each step.

The implications in \cref{fact:at-mixing} are all known and have been widely used, often implicitly. 
In the proof below, we give references where explicit statements or direct proofs are available. 

\begin{proof}[Proof of \cref{fact:at-mixing}]
(\ref{item:poincare}) and (\ref{item:MLSI}) are proved in \cite[Proposition 1.1]{CMT15}. 
To show (\ref{item:RED}), let $P_v$ be the transition matrix corresponding to updating the spin at $v$ conditioned on all other vertices. 
Thus, we have the decomposition
\[
\Pgl = \frac{1}{n}\sum_{v\in V} P_v.
\]
Let $f = \nu/\mu$ be the relative density of $\nu$ with respect to $\mu$. Then we get
\begin{align*}
\kl{\nu \Pgl}{\mu} 
&= D_{\mathrm{KL}} \left(\frac{1}{n} \sum_{v\in V} \nu P_v \Biggm \Vert \mu \right) 
\le \frac{1}{n} \sum_{v\in V} \kl{\nu P_v}{\mu}
= \frac{1}{n} \sum_{v\in V} \Ent(P_v f)\\
&= \frac{1}{n} \sum_{v\in V} \Ent[\mu_v (f)] 
= \frac{1}{n} \sum_{v\in V} \Ent(f) - \mu[\Ent_v(f)] 
\\
&= \Ent(f) - \frac{1}{n} \sum_{v\in V} \mu[\Ent_v(f)] \le \left( 1 - \frac{1}{C_1 n} \right) \Ent(f) 
\\&=  \left( 1 - \frac{1}{C_1 n} \right) \kl{\nu}{\mu}. 
\end{align*}
(\ref{item:mixing}) can be deduced from (\ref{item:RED}) as shown by \cite[Lemma 2.4]{BCPSV20}; see also \cite[Corollary 2.8]{BT06} for the continuous time setting. 
(\ref{item:concentration}) follows from (\ref{item:MLSI}) and \cite[Lemma 15]{CGM19}. 
Finally, (\ref{item:SLSI}) follows by an application of \cite[Theorem A.1]{DSC96}. 
\end{proof}


\section{Approximate Tensorization via Uniform Block Factorization}\label{sec:approxtensoruniformblock}


Fix a graph $G$ on $n$ vertices of maximum degree at most $\Delta$, and assume that $\mu$ is a $\bb$-marginally bounded Gibbs distribution defined on $G$ satisfying the $\ceil{\theta n}$-uniform block factorization of entropy with constant $C$ where $\theta \le b^2/(4e\Delta)$; 
i.e., for $\ell = \ceil{\theta n}$ and all $f: \Omega \to \R_{\ge 0}$ it holds that
\[
\frac{\ell}{n} \, \Ent(f) \le C \cdot \frac{1}{\binom{n}{\ell}} \sum_{S \in \binom{V}{\ell}} \mu[\Ent_S(f)]. 
\]
We will show that $\mu$ also satisfies the approximate tensorization of entropy with constant $\Theta(C)$, which establishes \cref{lem:comparison}.

The intuition behind our approach is that for $\ell$ as large as $\theta n$, if one picks a uniformly random subset $S \subseteq V$ satisfying $|S| = \ell$, then the induced subgraph $G[S]$ of $G$ on vertex set $S$ is disconnected into many small connected components, each of which has constant size in expectation and at most $O(\log n)$ with high probability. 
Since the conditional Gibbs distribution $\mu_S^\tau$ is a product distribution of each connected component, we can use entropy factorization for product distributions to reduce approximate tensorization on $G$ to that on small connected subgraphs of $G$. 
This allows us to upper bound the optimal approximate tensorization constant with a converging series. 

Towards fulfilling this intuition, for any $S \subseteq V$, let $\mathcal{C}(S)$ denote the set of connected components of $G[S]$, with each connected component being viewed as a subset of vertices of $S$. 
Note that $\mathcal{C}(S)$ is a partition of $S$. 
For any $v \in S$, let $S_{v}$ denote the (unique) connected component in $\mathcal{C}(S)$ containing~$v$; for $v \notin S$, take $S_{v} = \emptyset$. 
The following is a well-known fact regarding the factorization of entropy for product measures; see, e.g., \cite{Cesi01,CMT15}. 
\begin{lemma}
\label{lem:prod-factor}
For every subset $S \subseteq V$, every boundary condition $\tau \in \Omega_{V \setminus S}$, and every function $f: \Omega_S^\tau \to \R_{\ge 0}$, we have
\[
\Ent_S^\tau(f) \le \sum_{U \in \mathcal{C}(S)} \mu_S^\tau[\Ent_U(f)]. 
\]
\end{lemma}
Recall that $\Ent_U(f) = \Ent_U^{\phi}(f)$ is regarded as a function of the boundary condition $\phi \in \Omega_{S \setminus U}^\tau$ on $S \setminus U$, and $\mu_S^\tau[\Ent_U(f)]$ is the expectation of it under the conditional Gibbs measure $\mu_S^\tau$. 

We also need the following crude exponential upper bound on the approximate tensorization constant for a Gibbs distribution with bounded marginals. 

\begin{lemma}\label{lem:at-crude}
If $\mu$ is $\bb$-marginally bounded, 
then for every subset $U \subseteq V$, every boundary condition $\xi \in \Omega_{V \setminus U}$, and every function $f: \Omega_U^\xi \to \R_{\ge 0}$, we have
\[
\Ent_U^\xi(f) \le \frac{3|U|^2 \log(1/b)}{2b^{2|U|+2}} \sum_{v \in U} \mu_U^\xi[\Ent_v(f)]. 
\]
\end{lemma}

Finally, the lemma below shows that when a uniformly random and sufficiently small subset of vertices is selected, the size of the connected component containing a given vertex is small with high probability.
\begin{lemma}
\label{lem:probsizecomponent}
Let $G=(V,E)$ be an $n$-vertex graph of maximum degree at most $\Delta$.
Then for every $k \in \N^+$ we have 
\begin{align*}
    \PP_S(|S_{v}| = k) \leq \frac{\ell}{n} \cdot (2e\Delta \theta)^{k - 1},
\end{align*}
where the probability $\PP$ is taken over a uniformly random subset $S \subseteq V$ of size $\ell = \ceil{\theta n}$.
\end{lemma}




The proofs of \cref{lem:probsizecomponent,lem:at-crude} can be found in \cref{subsec:at-proof-lemmas}. 
We are now ready to prove \cref{lem:comparison}. 

\begin{proof}[Proof of \cref{lem:comparison}]
Combining everything in this section, we deduce that
\begin{align*}
    &\Ent(f) \\
    &\leq C \cdot \frac{n}{\ell} \cdot \frac{1}{\binom{n}{\ell}} \sum_{S \in \binom{V}{\ell}} \mu[\Ent_S(f)] \tag*{($\ell$-uniform block factorization)}\\
    &\leq C \cdot \frac{n}{\ell} \cdot \frac{1}{\binom{n}{\ell}} \sum_{S \in \binom{V}{\ell}} \sum_{U \in \mathcal{C}(S)} \mu[\Ent_U(f)] \tag*{(\cref{lem:prod-factor})} \\
    &\leq C \cdot \frac{n}{\ell} \cdot \frac{1}{\binom{n}{\ell}} \sum_{S \in \binom{V}{\ell}} \sum_{U \in \mathcal{C}(S)} \frac{3|U|^2 \log(1/b)}{2b^{2|U|+2}} \sum_{v \in U} \mu[\Ent_v(f)] \tag*{(\cref{lem:at-crude})} \\
    &= \frac{3C\log(1/b)}{2b^4}  \cdot \frac{n}{\ell} \,\sum_{v \in V} \mu[\Ent_v(f)]    \sum_{k=1}^{\ell} \PP_S(|S_v| = k) \cdot \frac{k^2}{b^{2(k-1)}}   \tag*{(rearranging)} \\
    &\le \frac{3C\log(1/b)}{2b^4}  \,\sum_{v \in V} \mu[\Ent_v(f)]    \sum_{k=1}^{\ell} k^2 \left( \frac{2e\Delta \theta}{\bb^2} \right)^{k-1}   \tag*{(\cref{lem:probsizecomponent})} \\
    &\le \frac{3C\log(1/b)}{2b^4}  \,\sum_{k=1}^{\ell} \frac{k^2}{2^{k-1}} \sum_{v \in V} \mu[\Ent_v(f)] \tag*{($\theta \le \frac{b^2}{12\Delta}$)}\\
    &\le \frac{18C\log(1/b)}{b^4} \,\sum_{v \in V} \mu[\Ent_v(f)].  \tag*{($\sum_{k=1}^\infty \frac{k^2}{2^{k-1}} = 12$)} \\
\end{align*}
This establishes the lemma. 
\end{proof}

\subsection{Proof of Technical Lemmas}
\label{subsec:at-proof-lemmas}

We first prove \cref{lem:at-crude} which gives a crude bound on the approximate tensorization constant for any subset and boundary condition. 

\begin{proof}[Proof of \cref{lem:at-crude}]
Fix a subset $U \subseteq V$ of size $k \ge 1$ and some boundary condition $\xi \in \Omega_{V \setminus U}$. 
Let $C_1 = C_1(U,\xi)$ be the optimal constant of approximate tensorization for $\mu_U^\xi$; 
hence, for every function $f: \Omega_U^\xi \to \R_{\ge 0}$ one has
\[
\Ent_U^\xi(f) \le C_1 \sum_{v\in U} \mu_U^\xi[\Ent_v(f)]. 
\]
Let $\lambda = \lambda(U,\xi)$ be the spectral gap of the Glauber dynamics for $\mu_U^\xi$, and let $\rho = \rho(U,\xi)$ be the standard log-Sobolev constant. 
Thus, for every function $f: \Omega_U^\xi \to \R_{\ge 0}$ it holds that
\begin{align*}
\lambda \, \Var_U^\xi(f) &\le \frac{1}{k} \sum_{v \in U} \mu_U^\xi[\Var_v(f)];\\
\rho \, \Ent_U^\xi(f) &\le \frac{1}{k} \sum_{v \in U} \mu_U^\xi[\Var_v(\sqrt{f})]. 
\end{align*}
Since $\Var_v(\sqrt{f}) \le \Ent_v(f)$ (see \cite{Sal97}), we have
\begin{equation}\label{eq:cat-slsc}
\Cat \le \frac{1}{\rho k}; 
\end{equation}
see also \cite[Proposition 1.1]{CMT15}. 
Next, \cite[Corollary A.4]{DSC96} gives a comparison between the standard log-Sobolev constant and the spectral gap:
\[
\rho \ge \frac{(1-2\mu^*)}{\log(1/\mu^*-1)} \lambda
\]
where 
$ \mu^* = \min_{\sigma \in \Omega_U^\xi} \mu_U^\xi(\sigma) $. 
Since $\mu$ is $\bb$-marginally bounded, we have $\mu^* \ge b^k$. 
Also, notice that $|\Omega_U^\xi| = 1$ and $|\Omega_U^\xi| = 2$ corresponds to trivial cases where we have $\Cat \le 1$, so we may assume that $|\Omega_U^\xi| \ge 3$ which makes $\mu^* \le 1/3$. 
It follows that
\begin{equation}\label{eq:slsc-gap}
\rho \ge \frac{\lambda}{3k \log(1/b)}. 
\end{equation}
Finally, Cheeger's inequality yields
\begin{equation}\label{eq:gap-cond}
\lambda \ge \frac{\Phi^2}{2}
\end{equation}
where $\Phi$ is the conductance of the Glauber dynamics defined by
\begin{align*}
\Phi &= \min_{\substack{\Omega_0 \subseteq \Omega_U^\xi\\ \mu_U^\xi(\Omega_0) \le \frac{1}{2}}} \Phi_{\Omega_0} \\
\Phi_{\Omega_0} &= \frac{\Pgl(\Omega_0, \Omega_U^\xi \setminus \Omega_0)}{\mu_U^\xi(\Omega_0)} = \frac{1}{\mu_U^\xi(\Omega_0)} \sum_{\sigma \in \Omega_0} \sum_{\tau \in \Omega_U^\xi \setminus \Omega_0} \mu_U^\xi(\sigma) \Pgl(\sigma, \tau). 
\end{align*}
Our assumption that $\mu$ is totally-connected guarantees $\Phi_{\Omega_0} > 0$ for every $\Omega_0 \subseteq \Omega_U^\xi$ with $\mu_U^\xi(\Omega_0) \le \frac{1}{2}$. 
Furthermore, since $\mu$ is $\bb$-marginally bounded, for every $\sigma \in \Omega_0$ and $\tau \in \Omega_U^\xi \setminus \Omega_0$ such that $\Pgl(\sigma, \tau)>0$ we have
\[
\mu_U^\xi(\sigma) \Pgl(\sigma, \tau) \ge b^k \cdot \frac{b}{k} = \frac{b^{k+1}}{k}. 
\]
This gives 
\begin{equation}\label{eq:cond-bound}
\Phi \ge \frac{2b^{k+1}}{k}.
\end{equation}
Combining \cref{eq:cat-slsc,eq:slsc-gap,eq:gap-cond,eq:cond-bound}, we finally conclude that
\[
C_1 \le \frac{3k^2 \log(1/b)}{2b^{2k+2}},
\]
as claimed. 
\end{proof}

Next we establish \cref{lem:probsizecomponent}. We use the following lemma concerning the number of connected induced subgraphs in a bounded degree graph. 
\begin{lemma}[{\cite[Lemma 2.1]{BCKL13}}]\label{lem:ctconnected}
Let $G=(V,E)$ be a graph with maximum degree at most $\Delta$, and $v \in V$. Then for every $k \in \N^+$, the number of connected induced subgraphs of $G$ containing $v$ with $k$ vertices is at most $(e\Delta)^{k-1}$.
\end{lemma}

We then prove \cref{lem:probsizecomponent}. 
\begin{proof}[Proof of \cref{lem:probsizecomponent}]
If $\mathcal{A}_{v}(k)$ denotes the collection of subsets of vertices $U \subseteq V$ such that $|U| = k$, $v \in U$, and $G[U]$ is connected, then by the union bound, we have
\begin{align*}
    \PP_S(|S_{v}| = k) &\leq \PP_S(\exists U \in\mathcal{A}_{v}(k) : U \subseteq S) \\
    &\leq \sum_{U \in \mathcal{A}_{v}(k)} \PP_{S}(U \subseteq S) \\
    &= |\mathcal{A}_{v}(k)| \cdot \frac{\ell}{n} \cdot \frac{\ell-1}{n-1} \dotsb \frac{\ell - k + 1}{n - k + 1} \\
    &\leq |\mathcal{A}_{v}(k)| \cdot \frac{\ell}{n} \cdot \wrapp{\frac{\ell-1}{n-1}}^{k-1}. 
\end{align*}
We may assume that $n \ge 2$ (when $n=1$ the lemma holds trivially), and thus
\[
\frac{\ell-1}{n-1} \le \frac{\theta n}{n-1} \le 2\theta.  
\]
The lemma then follows immediately from $|\mathcal{A}_{v}(k)| \leq (e\Delta)^{k-1}$ by \cref{lem:ctconnected}. 
\end{proof}

\section{Global Entropy Contraction via Local Spectral Expansion}
\label{sec:Proof-ent-contraction}

In this section, we prove \cref{lem:ent-contraction} by establishing global entropy contraction when the simplicial complex is a local spectral expander. 
We give preliminaries in \cref{subsec:SC-preliminaries} for simplicial complexes. 
In \cref{subsec:local-to-global} we present a very general local-to-global scheme for entropy contraction. 
Finally, in \cref{subsec:local-ent-spec} we reduce local entropy contraction to local spectral expansion. 

\subsection{Preliminaries for Simplicial Complexes}
\label{subsec:SC-preliminaries}

Following the definitions and notations from \cref{subsec:SC}, here we give a few more definitions, examples, and references, emphasizing the global and local walks in weighted simplicial complexes. 
The proofs of \cref{claim:bounded,lem:equiv} are included in \cref{subsubsec:proof-claim}.

Consider a pure $n$-dimensional weighted simplicial complex $(\SC,w)$. 
We say $\SC$ is \textit{$n$-partite} if the ground set $\mathcal{U}$ of $\SC$ admits a partition $\mathcal{U} = \mathcal{U}_{1} \cup \dots \cup \mathcal{U}_{n}$ such that every face of $\SC$ has at most one element from each part $\mathcal{U}_{1},\dots,\mathcal{U}_{n}$. 
For a distribution $\mu$ over $[q]^V$ where $|V| = n$, the corresponding weighted simplicial complex $(\SC,\mu)$ is $n$-partite.

\subsubsection{Global Walks}
For $0\le k \le n$, the distribution $\pi_k$ on $\SC(k)$ is given by for every $\tau \in \SC(k)$, 
\begin{align*}
    \pi_{k}(\tau) = \frac{1}{\binom{n}{k}} \sum_{\sigma \in \SC(n) : \sigma \supseteq \tau} \pi_{n}(\sigma)
\end{align*}
Moreover, we have the following equality:
\[
\pi_{k}(\tau) = \frac{1}{k+1} \sum_{\sigma \in \SC(k+1) : \sigma \supseteq \tau} \pi_{k+1}(\sigma). 
\] 
We may define simple random walks on $\SC(k)$ for sampling from the distributions $\pi_{k}$ given by ``down-up'' and ``up-down'' motions in $\SC$. 
These walks were first introduced in \cite{KM17}, and further studied in \cite{DK17, KO18} in the context of high-dimensional expanders. 
Recent works \cite{ALOV18ii, CGM19, ALO20, AL20, CLV20, FGYZ20, CGSV20} have leveraged these walks to study mixing times of Markov chains.

For each $0 \leq k \leq n-1$, define the \textit{order-$k$ (global) up operator} $P_{k}^{\uparrow} \in \R^{\SC(k) \times \SC(k+1)}$ and the \textit{order-$(k+1)$ (global) down operator} $P_{k+1}^{\downarrow} \in \R^{\SC(k+1) \times \SC(k)}$ as the row stochastic matrices with the following entries:
\begin{gather*}
    P_{k}^{\uparrow}(\tau, \sigma) = \begin{cases}
        \frac{\pi_{k+1}(\sigma)}{(k+1)  \pi_{k}(\tau)}, &\quad\text{if } \tau \subseteq \sigma \\
        0, &\quad\text{o.w.}
    \end{cases} \\
    P_{k+1}^{\downarrow}(\sigma,\tau) = \begin{cases}
        \frac{1}{k+1},&\quad\text{if } \tau \subseteq \sigma \\
        0, &\quad\text{o.w.}
    \end{cases}
\end{gather*}
In words, the action of $P_{k}^{\uparrow}$ can be described as starting with $\tau \in \SC(k)$ and adding a random element $i \notin \tau$ according to the conditional distribution $\pi_{\tau,1}(i) = \frac{\pi_{k+1}(\tau \cup \{i\})}{(k+1) \pi_{k}(\tau)}$. 
In particular, we have $\pi_k P_{k}^{\uparrow} = \pi_{k+1}$. 
Similarly, the action of $P_{k+1}^{\downarrow}$ can be described as starting with $\sigma \in \SC(k+1)$ and removing a uniformly random element, and we have $\pi_{k+1} P_{k+1}^{\downarrow} = \pi_k$. 

From here, for $0\le j < k \le n$ we may define the order-$(j,k)$ (global) up-down walk on $\SC(j)$ as $P_{j,k}^\wedge = P_j^\uparrow \cdots P_{k-1}^\uparrow P_k^\downarrow \cdots P_{j+1}^\downarrow$ and the order-$(k,j)$ (global) down-up walk on $\SC(k)$ as $P_{k,j}^\vee = P_k^\downarrow \cdots P_{j+1}^\downarrow P_j^\uparrow \cdots P_{k-1}^\uparrow$. 
The stationary distributions of $P_{j,k}^\wedge$ and $P_{k,j}^\vee$ are $\pi_j$ and $\pi_k$ respectively. 

For example, for any distribution $\mu$ over $[q]^V$, the Glauber dynamics for sampling from $\mu$ is precisely the order-$(n,n-1)$ down-up walk on the corresponding weighted simplicial complex $(\SC,\mu)$, and the heat-bath block dynamics which updates a random subset of $\ell$ vertices in each step is precisely the order-$(n,n-\ell)$ down-up walk. 

Let $0<k \le n$ and let $f^{(k)} : \SC(k) \to \R$ be an arbitrary function. Then for each $0\le j < k$ we define the function $f^{(j)}: \SC(j) \to \R$ by $f^{(j)} = P_j^\uparrow \cdots P_{k-1}^\uparrow f^{(k)}$. One can think of $f^{(j)}$ as the projection of $f^{(k)}$ onto $\SC(j)$. 

\subsubsection{Local Walks}

One of the beautiful properties of simplicial complexes is that they facilitate a nice decomposition of the global walks into local walks. 
Recall that for every face $\tau \in \SC(k)$ there is a $(n-k)$-dimensional weighted simplicial subcomplex $(\SC_{\tau}, w_{\tau})$ where $\SC_{\tau} = \{\xi  \subseteq \mathcal{U} \setminus \tau: \tau \cup \xi \in \SC \}$ and $w_\tau(\xi) = w(\tau \cup \xi)$ for each $\xi \in \SC_\tau$. 
The subcomplex is known as the \textit{link} of $\SC$ with respect to the face $\tau$. 
The induced distribution $\pi_{\tau,n-k}$ over $\SC_{\tau}(n-k)$ can be thought of as the distribution of a face $\sigma$ obtained from $\pi_{n}$ conditioned on the event that $\tau \subseteq \sigma$.

Again, consider as an example the weighted simplicial complex $(\SC,\mu)$ with respect to a distribution $\mu$ over $[q]^V$. 
For every $U \subseteq V$ of size $k$ and every $\tau \in \Omega_U$, one can check that the distribution $\pi_{(U,\tau), n-k}$ over maximal faces of the link $(\SC_{(U,\tau)},w_{(U,\tau)})$ is exactly the conditional distribution $\mu_{V \setminus U}^\tau$ over configurations on $V \setminus U$.

With the notion of link, we may now define the local walks rigorously. 
For each $0 \leq k \leq n-2$ and $\tau \in \SC(k)$, we define the \textit{local walk} at $\tau$ as $P_{\tau} = 2 P_{\tau,1,2}^{\wedge} - I$ where $P_{\tau,1,2}^{\wedge} = P_{\tau,1}^{\uparrow}P_{\tau,2}^{\downarrow}$ is the order-$(1,2)$ up-down walk on the link $(\SC_{\tau}, w_{\tau})$. 
More specifically, for $i,j \in \SC_{\tau}(1)$, we have
\begin{align*}
    P_{\tau}(i,j) = \begin{cases}
        \pi_{\tau \cup \{i\}, 1} (j) = \frac{\pi_{\tau,2}(\{i,j\})}{2 \pi_{\tau,1}(i)}, &\quad\text{if } \{i,j\} \in \SC_{\tau}(2) \\
        0, &\quad\text{o.w.}
    \end{cases}
\end{align*}
Note that the stationary distribution of $P_{\tau}$ is $\pi_{\tau,1}$. 
One way to think of $P_{\tau}$ is being the random walk matrix corresponding to the following weighted graph: take the vertices of the weighted graph to be the elements of $\SC_{\tau}(1)$, and take the edges to be the vertex-pairs in $\SC_{\tau}(2)$, with weights given by $\pi_{\tau,2}$.

We can establish mixing properties of global walks by decomposing them into local walks, which are usually easier to study; 
we refer the readers to \cite{KO18, ALOV18ii, CGM19, AL20} for more details. 

Let $0 < k \le n$ and let $f: \SC(k) \to \R$ be an arbitrary function. 
Then for every $0\le j < k$ and every $\tau \in \SC(j)$, we define the function $f_\tau^{(k-j)}: \SC_\tau(k-j) \to \R$ by $f_\tau^{(k-j)} (\xi) = f^{(k)} (\tau \cup \xi)$ for each $\xi \in \SC_\tau(k-j)$. 
One can think of $f_\tau^{(k-j)}$ as the restriction of $f^{(k)}$ to the link $\SC_\tau(k-j)$.

Finally, we mention that the notion of global and local walks has a similar flavor as the restriction-projection framework used in \cite{MR00, MR02, JSTV04} for decomposition of general Markov chains. 
For instance, the analog of the ``restriction chains'' in the setting of simplicial complexes are the Glauber dynamics of the conditional distributions. 
However, we note that these analogous ``restriction chains'' significantly overlap, in contrast to the partitioning condition in \cite{MR00, MR02, JSTV04}.

\subsubsection{Proofs of \texorpdfstring{\cref{claim:bounded,lem:equiv}}{Claim 1.11 and Lemma 2.6}}

\label{subsubsec:proof-claim}

We present here the proofs of \cref{claim:bounded,lem:equiv}.

\begin{proof}[Proof of \cref{claim:bounded}]
For all $0\le k \le n-1$, every $\tau \in \SC(k)$, and every $i \in \SC_\tau(1)$, we have
\[
\pi_{\tau,1}(i) = \frac{\pi_{k+1}(\tau \cup \{i\})}{\sum_{j \in \SC_\tau(1)} \pi_{k+1}(\tau \cup \{j\})} = \frac{\pi_{k+1}(\tau \cup \{i\})}{(k+1) \pi_k(\tau)}. 
\]
Now, consider the weighted simplicial complex $(\SC,\mu)$ corresponding to a distribution $\mu$ over $[q]^V$. 
Every face in $\SC(k)$ is in the form $(U,\tau)$ where $U \subseteq V$, $|U| = k$, and $\tau \in \Omega_U$, and every element in $\SC_\tau(1)$ is of the form $(v,i)$ for some $v \in V\setminus U$ and $i\in \Omega_v^\tau$. 
Hence, the equality above implies that
\begin{align*}
\pi_{(U,\tau),1}\big((v,i)\big) &= \frac{\pi_{k+1}\big( (U \cup \{v\}, \tau \cup \{i\}) \big)}{(k+1) \pi_k\big( (U,\tau) \big)}\\
&= \frac{\frac{1}{\binom{n}{k+1}} \, \mu(\sigma_U \seq \tau, \sigma_v \seq i)}{(k+1) \cdot \frac{1}{\binom{n}{k}} \, \mu(\sigma_U \seq \tau)}\\
&= \frac{1}{n-k} \, \mu(\sigma_v \seq i \mid \sigma_U = \tau)\\
&\ge \frac{b}{n-k}
\end{align*}
where the last inequality follows from the marginal boundedness of $\mu$. This shows the claim. 
\end{proof}

From the proof we see that $\pi_{(U,\tau),1}\big((v,i)\big) \le \frac{1}{n-k}$ always holds. 
This in fact is true for all weighted simplicial complexes. 

To prove \cref{lem:equiv}, we need the following entropy decomposition result in \cite{CGM19}.
\begin{lemma}[Entropy Decomposition in Simplicial Complexes, \cite{CGM19}]
\label{lem:entropydecomp}
Let $(\SC,w)$ be a pure $n$-dimensional weighted simplicial complex, and let $1 \leq j < k \leq n$. 
Then the following decomposition of entropy holds: 
for every $f^{(k)} : \SC(k) \rightarrow \R_{\geq0}$, 
\begin{align*}
    \Ent_{\pi_{k}}(f^{(k)}) &= \Ent_{\pi_{j}}(f^{(j)}) + \sum_{\tau \in \SC(j)} \pi_{j}(\tau) \cdot \Ent_{\pi_{\tau,k-j}}(f_{\tau}^{(k-j)}). 
\end{align*}
\end{lemma}

\begin{proof}
Since $\pi_{k}(f^{(k)}) = \pi_{j}(f^{(j)})$ and the equality is scale invariant, we may assume without loss that $\pi_{k}(f^{(k)}) = \pi_{j}(f^{(j)}) = 1$; hence $\Ent_{\pi_{k}}(f^{(k)}) = \pi_{k} (f^{(k)}\log f^{(k)})$ and $\Ent_{\pi_{j}}(f^{(j)}) = \pi_{j} (f^{(j)}\log f^{(j)})$. 
Using the law of conditional expectation, we have
\begin{align*}
    &\pi_{k}\wrapp{f^{(k)}\log f^{(k)}} 
    \\&= \sum_{\tau \in \SC(j)} \pi_{j}(\tau) \cdot \pi_{\tau,k-j}\wrapp{f_\tau^{(k-j)} \log f_\tau^{(k-j)}} \\
    &= \sum_{\tau \in \SC(j)} \pi_{j}(\tau) \cdot \underset{=f^{(j)}(\tau)}{\underbrace{\pi_{\tau,k-j}(f_\tau^{(k-j)})}} \log \underset{=f^{(j)}(\tau)}{\underbrace{\pi_{\tau,k-j}(f_\tau^{(k-j)})}} 
    + \sum_{\tau \in \SC(j)} \pi_{j}(\tau) \cdot \Ent_{\pi_{\tau,k-j}}(f_\tau^{(k-j)}) \\
    &= \pi_{j} \wrapp{f^{(j)} \log f^{(j)}} + \sum_{\tau \in \SC(j)} \pi_{j}(\tau) \cdot \Ent_{\pi_{\tau,k-j}}(f_\tau^{(k-j)})
\end{align*}
as claimed. 
\end{proof}

\begin{proof}[Proof of \cref{lem:equiv}]
Since there is a one-to-one correspondence between $\Omega$ and $\SC(n)$, every function $f: \Omega \to \R_{\ge 0}$ for the spin system is equivalent to a function $f^{(n)} : \SC(n) \to \R_{\ge 0}$ for the simplicial complex. 
Moreover, for every $0\le k \le n-1$, every $U \subseteq V$ of size $k$, and every $\tau \in \Omega_U$, the function $f_\tau: \Omega_{V\setminus U} \to \R_{\ge 0}$ (restriction of $f$ to configurations on $V\setminus U$ with $U$ fixed to be $\tau$) is the same as the function $f_{(U,\tau)}^{(n-k)} : \SC_{(U,\tau)}(n-k) \to \R_{\ge 0}$ for the link with respect to $(U,\tau) \in \SC(k)$. 
Thus, we can get
\begin{align*}
\frac{1}{\binom{n}{\ell}} \sum_{S\in \binom{V}{\ell}} \mu[\Ent_S(f)]
&= \sum_{S\in \binom{V}{\ell}} \sum_{\tau \in \Omega_{V \setminus S}} 
\frac{1}{\binom{n}{\ell}} \cdot \mu(\sigma_{V \setminus S} \seq \tau) \cdot \Ent_S^\tau (f)\\
&= \sum_{U\in \binom{V}{n-\ell}} \sum_{\tau \in \Omega_U} 
\frac{1}{\binom{n}{\ell}} \cdot \mu(\sigma_U = \tau) \cdot \Ent_{V \setminus U}^\tau(f_\tau)\\
&= \sum_{(U,\tau) \in \SC(n-\ell)} \pi_{n-\ell}(U,\tau) \cdot \Ent_{\pi_{(U,\tau),\ell}}(f_{(U,\tau)}^{(\ell)})\\
&= \Ent_{\pi_n}(f^{(n)}) - \Ent_{\pi_{n-\ell}}(f^{(n-\ell)}) 
\end{align*}
where the last equality follows from \cref{lem:entropydecomp}. 
The lemma then follows. 
\end{proof}

\subsubsection{Implications of Global Entropy Contraction}

We summarize here a few corollaries of global entropy contraction for arbitrary weighted simplicial complexes. 

\begin{fact}\label{fact:ent-contraction}
Let $(\SC,w)$ be a pure $n$-dimensional weighted simplicial complex. 
If $(\SC,w)$ satisfies the order-$(r,s)$ global entropy contraction with rate $\kappa$, then the order-$(s,r)$ down-up walk and order-$(r,s)$ up-down walk satisfy all of the following:
\begin{enumerate}[(1)]
\item \label{it:gap}
The Poincar\'{e} inequality holds with constant $\lambda = \kappa$; 

\item \label{it:MLSI}
The modified log-Sobolev inequality holds with constant $\rho_0 = \kappa$; 


\item \label{it:red}
The relative entropy decays with rate $\alpha = \kappa$; 

\item \label{it:mix}
The mixing time of the order-$(s,r)$ down-up walk satisfies
\[
T_{\mathrm{mix}}(P^\vee_{s,r}, \eps) \le \ceil{\frac{1}{\kappa} \left( \log \log \frac{1}{\pi^*_s} + \log \frac{1}{2\eps^2} \right)}
\]
where $\pi^*_s = \min_{\sigma \in \SC(s)} \pi_s(\sigma)$; 
The same mixing time bound holds for the order-$(r,s)$ up-down walk as well with $\pi^*_s$ replaced by $\pi^*_r$.

\item \label{it:concentration}
For every $f:\SC(s) \to \R$ which is $c$-Lipschitz w.r.t. the shortest path metric induced by the order-$(s,r)$ down-up walk on $\SC(s)$, and every $a \ge 0$, we have the concentration inequality
\[
	\Pr_{\sigma \sim \pi_s} \wrapb{ \left| f(\sigma) - \mu(f) \right| \ge a } \le 2 \, \exp\wrapp{-\frac{\kappa a^2}{2c^{2}}}. 
\]
\end{enumerate}
\end{fact}

These implications are shown in \cite{CGM19}. We remark that \cref{fact:at-mixing} can also be viewed as a consequence of \cref{fact:ent-contraction} by \cref{lem:equiv}. 
Technically we can also show a standard log-Sobolev inequality under the assumption of marginal boundedness, like (\ref{item:SLSI}) of \cref{fact:at-mixing}. 
However, the constant is complicated to state in this setting and so we omit it here. 

\begin{proof}
(\ref{it:gap}) can be proved by the linearization argument in the proof of \cite[Proposition 1.1]{CMT15}. 
The remaining results are all shown in \cite{CGM19}: 
(\ref{it:red}) follows from \cite[Corollary 13]{CGM19}; 
(\ref{it:MLSI}) follows by (\ref{it:red}) and \cite[Theorem 7]{CGM19}; 
(\ref{it:mix}) can be deduced from (\ref{it:red}) and the direct proof of \cite[Corollary 8]{CGM19};  
finally, (\ref{it:concentration}) follows from (\ref{it:MLSI}) and \cite[Lemma 15]{CGM19}. 
\end{proof}

\subsection{Local-to-Global Entropy Contraction}
\label{subsec:local-to-global}

As mentioned earlier, one of the beautiful properties of simplicial complexes is that they admit a local-to-global phenomenon, where mixing properties of the local walks $P_{\tau}$ imply mixing properties for the global walks in a quantitative sense. For instance, $( \frac{C}{n-1},\frac{C}{n-2},\dots,C )$-local spectral expansion implies $\Omega(n^{-(1+C)})$ spectral gap for the Glauber dynamics \cite{AL20} (although we note weaker results of this type were previously known due to \cite{DK17, KO18}). It turns out this local-to-global spectral result is completely equivalent to the property that local contraction of variance implies global contraction of variance; see \cref{app:var} for more discussion. The goal of this section is to prove an analogous local-to-global result for entropy. 
First, we formalize our notion of local entropy contraction.

\begin{definition}[Local Entropy Contraction]\label{def:localentcontraction}
We say a pure $n$-dimensional weighted simplicial complex $(\SC,w)$ satisfies \emph{$(\alpha_0,\dots,\alpha_{n-2})$-local entropy contraction} 
if for every (global) function $f^{(n)}: \SC(n) \to \R_{\ge 0}$, every $0\le k \le n-2$, and every $\tau \in \SC(k)$, we have
\begin{align*}
    \Ent_{\pi_{\tau,2}} (f_{\tau}^{(2)}) \geq (1 + \alpha_{k}) \, \Ent_{\pi_{\tau,1}} (f_{\tau}^{(1)}).
\end{align*}
\end{definition}
Recall that $f_{\tau}^{(2)}(\xi) = f^{(k+2)}(\tau \cup \xi) = (P_{k+2}^\uparrow \cdots P_{n-1}^\uparrow f^{(n)} ) (\tau \cup \xi)$ for each $\xi \in \SC_{\tau}(2)$, and $f_{\tau}^{(1)}$ is defined similarly. 

One might want to replace the conditions in \cref{def:localentcontraction} with ``for every $0\le k \le n-2$, every $\tau \in \SC(k)$, and every (local) function $f_\tau^{(2)}: \SC_\tau(2) \to \R_{\ge 0}$'', resulting in a stronger notion of local entropy contraction. 
Indeed, this is the case for local spectral expansion or local variance contraction (see \cref{app:var}). 
In contrast, our \cref{def:localentcontraction} only considers local functions originating from global functions. 
This is sufficient since in the end we are only interested in the order-$(n,k)$ down-up walk $P_{n,n-\ell}^{\vee}$ and only need to consider global functions $f^{(n)}: \SC(n) \to \R_{\ge 0}$. 
Moreover, using this weaker notion makes it easier for us to deduce local entropy contraction from local spectral expansion in \cref{subsec:local-ent-spec}.

With the notion of local entropy contraction, we are able to prove the following local-to-global entropy contraction result.
\begin{theorem}[Local-to-Global Entropy Contraction]
\label{thm:local-global}
Suppose a pure weighted $n$-dimensional simplicial complex $(\SC,w)$ satisfies $(\alpha_0,\dots,\alpha_{n-2})$-local entropy contraction. Then for every function $f^{(n)}: \SC(n) \to \R_{\ge 0}$ and all $1\le k \le n-1$, we have
\begin{equation}\label{eq:ent-recursion}
\frac{\Ent_{\pi_{k+1}}(f^{(k+1)})}{\sum_{i=0}^{k} \Gamma_i} 
\ge 
\frac{\Ent_{\pi_j}(f^{(k)})}{\sum_{i=0}^{k-1} \Gamma_i} 
\end{equation}
where $\Gamma_{i} = \prod_{j=0}^{i-1} \alpha_j$ for $1\le i\le n-1$ and $\Gamma_{0} = 1$. 
In particular, $(\SC,w)$ satisfies the order-$(k,n)$ global entropy contraction with rate 
\[
\kappa = \frac{\sum_{i=k}^{n-1} \Gamma_i}{\sum_{i=0}^{n-1} \Gamma_i}. 
\]
\end{theorem}
\begin{remark}
After we posted a preliminary version of this paper, it was brought to our attention that Guo-Mousa had also independently obtained this result \cite{GM20}. The analogous result for variance, which we present in \cref{app:var}, was also independently obtained by Kaufman-Mass \cite{KM20}. Finally, we mention the recent work \cite{AASV21} which established local-to-global arguments for any $f$-divergence in a very general setting.
\end{remark}

In \cite{CGM19}, the authors showed that the simplicial complexes with respect to homogeneous strongly log-concave distributions satisfy $(1,\dots,1)$-local entropy contraction; see Lemma 10 of \cite{CGM19}. 
That is, $\alpha_k = 1$ for all $k$ and thus $\Gamma_k = 1$ for all $k$. 
Then \cref{thm:local-global} implies that for every $k$ we have
\[
\Ent_{\pi_{k+1}}(f^{(k+1)}) \ge \frac{k+1}{k} \, \Ent_{\pi_k}(f^{(k)}),
\]
which recovers Lemma 11 of \cite{CGM19}. 

\begin{proof}[Proof of \cref{thm:local-global}]
We prove \cref{eq:ent-recursion} by induction. When $k=1$, \cref{eq:ent-recursion} is equivalent to
\[
\Ent_{\pi_2} (f^{(2)}) \ge (1+\alpha_0) \, \Ent_{\pi_1} (f^{(1)}), 
\]
which holds by assumption. Now suppose \cref{eq:ent-recursion} holds for $k-1$; i.e., 
\[
\frac{\Ent_{\pi_{k}}(f^{(k)})}{\sum_{i=0}^{k-1} \Gamma_i} 
\ge 
\frac{\Ent_{\pi_{k-1}}(f^{(k-1)})}{\sum_{i=0}^{k-2} \Gamma_i}. 
\] 
By \cref{lem:entropydecomp} we have
\begin{align*}
\Ent_{\pi_{k+1}}(f^{(k+1)}) - \Ent_{\pi_{k-1}}(f^{(k-1)})
&= \sum_{\tau \in \SC(k-1)} \pi_{k-1}(\tau) \cdot \Ent_{\pi_{\tau,2}}(f_{\tau}^{(2)})\\
&\ge (1+\alpha_{k-1}) \sum_{\tau \in \SC(k-1)} \pi_{k-1}(\tau) \cdot \Ent_{\pi_{\tau,1}}(f_{\tau}^{(1)})\\
&= (1+\alpha_{k-1}) \left( \Ent_{\pi_{k}}(f^{(k)}) - \Ent_{\pi_{k-1}}(f^{(k-1)}) \right). 
\end{align*}
It follows that
\begin{align*}
\Ent_{\pi_{k+1}}(f^{(k+1)}) 
&\ge (1+\alpha_{k-1}) \, \Ent_{\pi_{k}}(f^{(k)}) - \alpha_{k-1} \, \Ent_{\pi_{k-1}}(f^{(k-1)})\\
&\ge (1+\alpha_{k-1}) \, \Ent_{\pi_{k}}(f^{(k)}) - \alpha_{k-1} \cdot \frac{\sum_{i=0}^{k-2} \Gamma_i}{\sum_{i=0}^{k-1} \Gamma_i} \cdot \Ent_{\pi_{k}}(f^{(k)}) \tag{Induction}\\
&= \left( 1 + \alpha_{k-1} \cdot \frac{\Gamma_{k-1}}{\sum_{i=0}^{k-1} \Gamma_i} \right)  \Ent_{\pi_{k}}(f^{(k)})\\
&= \frac{\sum_{i=0}^{k} \Gamma_i}{\sum_{i=0}^{k-1} \Gamma_i} \cdot \Ent_{\pi_{k}}(f^{(k)}). 
\end{align*}
This proves the theorem. 
\end{proof}

\subsection{Local Entropy Contraction via Local Spectral Expansion}
\label{subsec:local-ent-spec}
Thanks to \cref{thm:local-global}, it remains to show local entropy contraction in order to prove \cref{lem:ent-contraction}. 
In general, proving local entropy contraction is a difficult task. Our goal here is to show that one can deduce local entropy contraction from local spectral expansion if we further enforce certain bounds on the marginals of the distribution.


\begin{theorem}\label{thm:bdd+expander=ent}
Suppose that $(\SC,w)$ is a pure $n$-dimensional weighted simplicial complex that is $(\bb_0,\dots,\bb_{n-1})$-marginally bounded and has $(\zeta_0, \dots, \zeta_{n-2})$-local spectral expansion. 
Then $\SC$ satisfies $(\alpha_0,\dots,\alpha_{n-2})$-local entropy contraction where for $0\le k \le n-2$, 
\begin{equation}\label{eq:alpha}
\alpha_k = \max\left\{ 
1-\frac{4\zeta_k}{\bb_k^2(n-k)^2} , 
\frac{1-\zeta_k}{4+2\log(\frac{1}{2b_k b_{k+1}})}
\right\}. 
\end{equation}
\end{theorem}
\begin{remark}
The purpose of the second bound is to obtain some weak control over the local entropy contraction when the first bound becomes vacuous (e.g. when $n - k$ is sufficiently small). However, we technically do not need it in the proof of our main mixing results, since we only require local entropy contraction when $n - k \geq \Omega(n0$ and $n$ is sufficiently large. Nevertheless, we record it here to demonstrate that there is a more general connection between local spectral expansion and local entropy contraction.
\end{remark}

With \cref{thm:bdd+expander=ent}, we are able to prove \cref{lem:ent-contraction}. 

\begin{proof}[Proof of \cref{lem:ent-contraction}]
First notice that it suffices to prove the lemma for the case $s = n$, since when $s < n$ we can consider the $s$-dimensional simplicial complex $\SC' = \{\sigma \in \SC: |\sigma| \le s\}$ instead. 
When $s = n$, the lemma follows immediately from \cref{thm:local-global,thm:bdd+expander=ent}. 
\end{proof}

The rest of this section aims to prove \cref{thm:bdd+expander=ent}. 
We will show separately the two bounds in \cref{eq:alpha} on the rate of local entropy contraction, and we will refer to them as \emph{the first bound} and \emph{the second bound}. 
The first bound is more subtle and indicates that $\alpha_k = 1 - \Theta(\zeta_k)$ for the weighted simplicial complex $(\SC,\mu)$ corresponding to a marginally bounded distribution $\mu$ over $[q]^V$ (see \cref{claim:bounded}). 
The second bound is crude but may still be helpful when the first bound is bad.

\subsubsection{Proof of the First Bound}

We break the proof of it into the following lemmas, which are the essence of our approach.

First we generalize Lemma 10 of \cite{CGM19} as follows. 
\begin{lemma}
\label{lem:localspecentvar}
Suppose the local walk $P_{\emptyset}$ with stationary distribution $\pi_{1}$ satisfies $\lambda_{2}(P_{\emptyset}) \leq \zeta$. Then for every $f^{(2)} : \SC(2) \rightarrow \R_{\geq0}$ and taking $f^{(1)} = P_{1}^{\uparrow}f^{(2)}$, we have
\begin{align*}
    \Ent_{\pi_{2}} (f^{(2)}) - 2\, \Ent_{\pi_{1}} (f^{(1)}) \geq - \zeta \cdot  \frac{\var_{\pi_{1}} (f^{(1)}) }{\pi_{1} (f^{(1)}) }. 
\end{align*}
\end{lemma}

The following lemma shows that for marginally bounded simplicial complexes, for any global function $f^{(n)}$, the induced local functions $f_\tau^{(1)}$ are ``balanced'' in the sense that the values of $f_\tau^{(1)}$ cannot be too large compared to the expectation of it. 
\begin{lemma}\label{lem:bddentvarcondition}
If $\SC$ is a $(b_0,\dots, b_{n-1})$-marginally bounded simplicial complex, 
then for every function $f^{(n)}: \SC(n) \to \R_{\ge 0}$, all $0\le k\le n-1$, every $\tau \in \SC(k)$ such that $f^{(k)}(\tau) > 0$, and every $i\in \SC_{\tau}(1)$, 
we have
\[
\frac{f_{\tau}^{(1)}(i)}{\pi_{\tau,1} (f_{\tau}^{(1)})} \le \frac{1}{b_k(n-k)}. 
\]
\end{lemma}

Finally, we show that for ``balanced'' functions the entropy and variance differ only by a constant factor after normalization. 
\begin{lemma}\label{lem:entvarrelate}
Let $\pi$ be a distribution over a finite set $\Omega$, and let $f :\Omega \rightarrow \R_{\geq0}$ such that $\pi(f) > 0$. 
If $f(x) \le c \cdot \pi(f)$ for all $x \in \Omega$, then
\[
\frac{\Var_\pi(f)}{\pi(f)} \le 4c^2 \, \Ent_\pi(f). 
\]
\end{lemma}

Note that $c\ge 1$ and we always have the inequality $\Ent_\pi(f) \le \frac{\Var_\pi(f)}{\pi(f)}$.

We then show how to use these three lemmas to prove the first bound in \cref{eq:alpha}.
\begin{proof}[Proof of the first bound in \cref{eq:alpha}]
Let $0 \leq k \leq n-2$, $\tau \in \SC(k)$ and $f^{(n)} : \SC(n) \rightarrow \R_{\geq0}$ be arbitrary. Then applying, \cref{lem:localspecentvar} to the link $(\SC_{\tau}, w_{\tau})$ and using $\lambda_{2}(P_{\tau}) \leq \zeta_{k}$, we have the inequality
\begin{align}\label{eq:localentvarineq}
    \Ent_{\pi_{\tau,2}} (f_{\tau}^{(2)}) - 2 \, \Ent_{\pi_{\tau,1}} (f_{\tau}^{(1)}) \geq -\zeta_{k} \cdot \frac{\var_{\pi_{\tau,1}} (f_{\tau}^{(1)})}{\pi_{\tau,1} (f_{\tau}^{(1)})}
\end{align}
where $f_{\tau}^{(2)}:\SC_{\tau}(2) \rightarrow \R_{\geq0}$ and $f_{\tau}^{(1)} : \SC_{\tau}(1) \rightarrow \R_{\geq0}$ are derived from $f^{(n)}$. By \cref{lem:bddentvarcondition}, we have $f_{\tau}^{(1)}(i) \leq \frac{1}{b_{k}(n-k)} \cdot \pi_{\tau,1} (f_{\tau}^{(1)})$ for all $i \in \SC_{\tau}(1)$. It follows by \cref{lem:entvarrelate} that
\begin{align}\label{eq:bddlocalentvar}
    \frac{\var_{\pi_{\tau,1}} (f_{\tau}^{(1)})}{\pi_{\tau,1} (f_{\tau}^{(1)})} \leq \frac{4}{b_{k}^{2}(n-k)^2}  \cdot \Ent_{\pi_{\tau,1}} (f_{\tau}^{(1)}). 
\end{align}
Hence, it follows by combining \cref{eq:localentvarineq,eq:bddlocalentvar} that
\begin{align*}
    \Ent_{\pi_{\tau,2}} (f_{\tau}^{(2)}) - 2\, \Ent_{\pi_{\tau,1}} (f_{\tau}^{(1)}) \geq -\frac{4\zeta_{k}}{b_{k}^{2}(n-k)^2} \cdot \Ent_{\pi_{\tau,1}} (f_{\tau}^{(1)}). 
\end{align*}
Rearranging, we obtain
\begin{align*}
    \Ent_{\pi_{\tau,2}} (f_{\tau}^{(2)}) \geq \wrapp{1 + \wrapp{1 - \frac{4\zeta_{k}}{b_{k}^{2}(n-k)^2}}} \Ent_{\pi_{\tau,1}} (f_{\tau}^{(1)}). 
\end{align*}
As this holds for all $\tau \in \SC(k)$ and all $0 \leq k \leq n-2$, we obtain the first bound. 
\end{proof}

It remains to prove \cref{lem:localspecentvar,lem:bddentvarcondition,lem:entvarrelate}. 
We note that these lemmas are logically independent of each other.


\begin{proof}[Proof of \cref{lem:localspecentvar}]
First, observe that the desired inequality is scale invariant, and hence we may assume without loss that $\pi_{2} (f^{(2)}) = \pi_{1} (f^{(1)}) = 1$. 
We shall write $ij$ to represent $\{i,j\} \in \SC(2)$ for simplicity. 
Let us rewrite $2\, \Ent_{\pi_{1}} (f^{(1)})$ in a form which is more convenient to compare with $\Ent_{\pi_{2}} (f^{(2)})$. Observe that
\begin{align*}
    \Ent_{\pi_{1}} (f^{(1)}) &= \sum_{i \in \SC(1)} \pi_{1}(i) f^{(1)}(i) \log f^{(1)}(i) \\
    &= \sum_{i \in \SC(1)} \pi_{1}(i) \wrapp{\sum_{j \in \SC(1) : ij \in \SC(2)} \frac{\pi_{2}(ij)}{2 \pi_{1}(i)} \cdot f^{(2)}(ij)} \log f^{(1)}(i) \\
    &= \sum_{i \in \SC(1)} \sum_{j \in \SC(1) : ij \in \SC(2)} \frac{\pi_{2}(ij)}{2} \cdot f^{(2)}(ij) \log f^{(1)}(i) \\
    &= \frac{1}{2}\sum_{ij \in \SC(2)} \pi_{2}(ij) \cdot f^{(2)}(ij) \log \big( f^{(1)}(i) f^{(1)}(j) \big).
\end{align*}
By the inequality $a\log \frac{a}{b} \geq a - b$ for any $a \ge 0$ and $b > 0$, we can get
\begin{align*}
    &\Ent_{\pi_{2}} (f^{(2)}) - 2\, \Ent_{\pi_{1}} (f^{(1)}) 
    \\&= \sum_{ij \in \SC(2)} \pi_{2}(ij) \cdot f^{(2)}(ij) \wrapp{\log f^{(2)}(ij) - \log \big( f^{(1)}(i) f^{(1)}(j) \big)} \\
    &\geq \sum_{ij \in \SC(2)} \pi_{2}(ij) \cdot \wrapp{f^{(2)}(ij) - f^{(1)}(i)f^{(1)}(j)} \\
    &= \underset{=\pi_{2} (f^{(2)}) = 1}{ \underbrace{ \sum_{ij \in \SC(2)} \pi_{2}(ij) f^{(2)}(ij) } }
    - \sum_{ij \in \SC(2)} \pi_{2}(ij) f^{(1)}(i) f^{(1)}(j) \\
    &= 1 - (f^{(1)})^{\top} W f^{(1)}
\end{align*}
where $W \in \R_{\geq0}^{\SC(1) \times \SC(1)}$ is the symmetric matrix with entries $W(i,j) = \frac{\pi_{2}(ij)}{2}$ whenever $ij \in \SC(2)$, and $W(i,j) = 0$ otherwise. 
Note that $W = \diag(\pi_{1}) P_{\emptyset}$ since recall that $P_{\emptyset}$ has entries $P_{\emptyset}(i,j) = \frac{\pi_{2}(ij)}{2 \pi_{1}(i)}$ whenever $ij \in \SC(2)$, and $P_{\emptyset}(i,j) = 0$ otherwise. Now, we analyze $W$ spectrally assuming knowledge of the spectral gap of $P_{\emptyset}$. Observe that since $\pi_{1}(f^{(1)}) = 1$, we have that
\begin{align*}
    &1 - (f^{(1)})^{\top}Wf^{(1)} \\
    &= \underset{- \var_{\pi_{1}}(f^{(1)})}{\underbrace{\pi_{1}\wrapb{f^{(1)}}^{2} - \pi_{1}\wrapb{(f^{(1)})^{2}}}} + \underset{= \mathcal{E}_{P_{\emptyset}}(f^{(1)},f^{(1)})}{\underbrace{\pi_{1}\wrapb{(f^{(1)})^{2}} - (f^{(1)})^{\top}\diag(\pi_{1})P_{\emptyset}f^{(1)}}} \\
    &= \mathcal{E}_{P_{\emptyset}}(f^{(1)}, f^{(1)}) - \var_{\pi_{1}}(f^{(1)}) \\
    &\geq (1 - \zeta) \cdot \var_{\pi_{1}}(f^{(1)}) -  \var_{\pi_{1}}(f^{(1)}) \tag{Spectral Gap or Poincar\'{e} Inequality} \\
    &= - \zeta \cdot \frac{\var_{\pi_{1}}(f^{(1)})}{\pi_{1}(f^{(1)})} \tag{$\pi_{1}(f^{(1)}) = 1$} \\
\end{align*}
and we are done.
\end{proof}

\begin{proof}[Proof of \cref{lem:bddentvarcondition}]

Without loss of generality, we may assume $\pi_{n} (f^{(n)}) = 1$, which also implies that $\pi_{k}(f^{(k)}) = 1$ for all $k$. 
It follows that if we define $\nu_{k} = f^{(k)}\pi_{k}$ for $0\le k \le n$, then $\nu_{k}$ is a distribution on $\SC(k)$. 
For intuition, note that $\nu_{k} = \nu_{n}P_{n}^{\downarrow}\dotsb P_{k+1}^{\downarrow}$, and hence we can regard these distributions as from the weighted simplicial complex $(\SC,\nu_n)$. 
We then have
\begin{align*}
	\frac{f_\tau^{(1)}(i)}{\pi_{\tau,1}  (f_\tau^{(1)})  } 
	= \frac{f^{(k+1)}(\tau \cup \{i\})}{f^{(k)}(\tau)} 
	= \frac{ \frac{\nu_{k+1}(\tau \cup \{i\})}{(k+1) \, \nu_{k}(\tau)} }{ \frac{\pi_{k+1}(\tau \cup \{i\})}{(k+1) \, \pi_{k}(\tau)} } 
	= \frac{\nu_{\tau,1}(i)}{\pi_{\tau,1}(i)}.
\end{align*}
Trivially we have $\nu_{\tau,1}(i) \le \frac{1}{n-k}$; see the proof of \cref{claim:bounded} in \cref{subsubsec:proof-claim} and the remark after it. 
Furthermore, by assumption, we have $\pi_{\tau,1}(i) \ge \bb_k$. The claim follows.
\end{proof}

\begin{proof}[Proof of \cref{lem:entvarrelate}]
Without loss of generality, we may assume that $\pi(f) = 1$. 
Then, $f(x) \le c$ for all $x \in \Omega$. 
Let $\nu = f\pi$, so $f$ is the relative density of $\nu$ with respect to $\pi$. 
We need to show that 
\[
\Var_\pi(f) \le 4c^2 \, \Ent_\pi(f). 
\]
\cref{lem:donskervaradhan} implies that for every function $g: \Omega \to \R$, we have
\[
\nu(g) \le \kl{\nu}{\pi} + \log \pi(e^g) = \Ent_\pi(f) + \log \pi(e^g). 
\]
Let $g = t(f - 1)$ for some parameter $t>0$ to be determined. 
Then 
\[
\nu(g) = t \,  (\nu(f) - 1) = t \, (\pi(f^2) - 1) = t \, \Var_\pi(f). 
\]
Hence, we obtain that
\[
\Var_\pi(f) \le 
\frac{1}{t} \, \Ent_\pi(f)
+ \frac{1}{t} \, \log \pi\left( e^{t (f-1)} \right).
\]
This is known as the entropy inequality \cite{MSW03}. 

Notice that $c \ge 1$ always and $\Ent_\pi(f) \le \Var_\pi(f)$ when $\pi(f) = 1$ (see, e.g., \cite{Sal97}). 
Consider first the case that $1\le c \le 2$. 
We shall pick 
\[
t = \sqrt{\frac{\Ent_\pi(f)}{\Var_\pi(f)}} \le 1.
\]
Then $t(f - 1) \le c-1 \le 1$. 
Since $e^x \le 1+x+x^2$ when $x\le 1$, we get
\begin{align*}
\log \pi\left( e^{t (f-1)} \right) 
&\le \log \pi\left( 1 + t (f-1) + t^2 (f-1)^2 \right)\\
&= \log \left( 1 + t^2 \, \Var_\pi(f) \right)\\
&\le t^2 \, \Var_\pi(f).
\end{align*}
It follows that
\[
\Var_\pi(f) \le 
\frac{1}{t} \, \Ent_\pi(f)
+ t \, \Var_\pi(f). 
\]
For our choice of $t$, we obtain
\[
\Var_\pi(f) \le 4 \, \Ent_\pi(f) \le 4c^2 \, \Ent_\pi(f). 
\]

Next, consider the case that $c>2$. This time we pick
\[
t = \sqrt{\frac{\Ent_\pi(f)}{\Var_\pi(f)}} \cdot \frac{2\ln c}{c} \le \frac{2\ln c}{c}.
\]
Then $t(f - 1) \le 2\ln c$. 
For all $x \le 2\ln c$, it holds that $e^x \le 1+x+ (\frac{c}{2\ln c})^2 x^2$. 
Hence, we get
\begin{align*}
\log \pi\left( e^{t (f-1)} \right) 
&\le \log \pi\left( 1 + t (f-1) + \left( \frac{c}{2\ln c} \right)^2 t^2 (f-1)^2 \right)\\
&= \log \left( 1 + t^2  \left( \frac{c}{2\ln c} \right)^2  \Var_\pi(f)  \right)\\
&\le t^2  \left( \frac{c}{2\ln c} \right)^2  \Var_\pi(f).
\end{align*}
We then deduce that
\begin{align*}
\Var_\pi(f) &\le 
\frac{1}{t} \, \Ent_\pi(f)
+ t \left( \frac{c}{2\ln c} \right)^2  \Var_\pi(f) \\ 
\Longrightarrow \quad 
\Var_\pi(f)  &\le \left( \frac{c}{\ln c} \right)^2  \Ent_\pi(f)
\le 4c^2 \, \Ent_\pi(f).
\end{align*}
This establishes the lemma. 
\end{proof}

\subsubsection{Proof of the Second Bound}\label{sec:localentdecayconstantsize}
Here we prove the second bound in \cref{eq:alpha}. 
We do this by reducing entropy contraction to bounding the standard log-Sobolev constant. 
Since the marginals are bounded, a comparison inequality between the standard log-Sobolev constant and spectral gap then finishes the proof.

We will show that for every (local) function $f^{(2)} : \SC(2) \to \R_{\ge 0}$, we have
\begin{equation}\label{eq:0alpha0}
\Ent_{\pi_2} (f^{(2)}) \ge (1+\alpha_0) \, \Ent_{\pi_1} (f^{(1)}), 
\end{equation}
where 
\[
\alpha_0 = \frac{1-\zeta_0}{4+2\log(\frac{1}{2\bb_0\bb_1})}. 
\]
This establishes the second bound for the case $k = 0$ and $\tau = \emptyset$. 
For arbitrary $0\le k \le n-2$ and arbitrary $\tau \in \SC(k)$, we can just consider the link $(\SC_\tau,w_\tau)$ at $\tau$ instead and achieve the results.

Recall that for a reversible Markov chain $P$ with stationary distribution $\pi$ on a finite state space $\Omega$, the standard logarithmic Sobolev constant is defined as
\[
\rho(P) = \inf \left\{\frac{\EE_{P}(\sqrt{f}, \sqrt{f})}{\Ent_{\pi}(f)} \biggm| f: \Omega \to \R_{\ge 0}, \, \Ent_\pi(f) \neq 0 \right\}. 
\]


We prove \cref{eq:0alpha0} by the following two lemmas. First we relate $\alpha_0$ with the standard log-Sobolev constant of the order-$(2,1)$ down-up walk $P_{2,1}^\vee$.

\begin{lemma}
\label{lem:logsobentdecay}
For every (local) function $f^{(2)} : \SC(2) \to \R_{\ge 0}$, we have
\begin{align*}
	\Ent_{\pi_1}(f^{(1)}) \leq \left( 1 - \rho(P_{2,1}^{\vee}) \right) \, \Ent_{\pi_2}(f^{(2)}). 
\end{align*}
\end{lemma}

The following lemma is equation (3.9) from \cite{DSC96} which compares the standard log-Sobolev constant and the spectral gap. 
\begin{lemma}[\cite{DSC96}] 
\label{lem:logsobspecgap}
For every reversible Markov chain $P$ with stationary distribution $\pi$ on a finite state space $\Omega$, we have the inequality
\begin{align*}
    \rho(P) \geq \frac{\lambda(P)}{2+\log (1 / \pi^*)}
\end{align*}
where $\pi^* = \min_{x \in \Omega} \pi(x)$.
\end{lemma}

We are now ready to prove the second bound of \cref{eq:alpha}. 
\begin{proof}[Proof of the second bound in \cref{eq:alpha}]
From \cref{lem:logsobentdecay} we deduce that for every (local) function $f^{(2)} : \SC(2) \to \R_{\ge 0}$, 
\begin{equation}\label{eq:fst}
\Ent_{\pi_2}(f^{(2)}) 
\ge \frac{1}{1 - \rho(P_{2,1}^{\vee})}  \, \Ent_{\pi_1}(f^{(1)}) 
\ge \left( 1 + \rho(P_{2,1}^{\vee}) \right) \Ent_{\pi_1}(f^{(1)}). 
\end{equation}
Meanwhile, \cref{lem:logsobspecgap} gives
\begin{equation}\label{eq:scd}
\rho(P_{2,1}^{\vee}) \ge \frac{\lambda(P_{2,1}^{\vee})}{2+\log (1 / \pi_2^*)} 
\ge \frac{1-\zeta_0}{4 + 2\log (1 / \pi_2^*)},
\end{equation}
where the last inequality follows from
\[
\lambda(P_{2,1}^{\vee}) = \lambda(P_{1,2}^{\wedge}) = 1 - \lambda_2(P_{1,2}^{\wedge}) = \frac{1}{2} \left( 1 - \lambda_2(P_\emptyset) \right) \ge \frac{1-\zeta_0}{2}. 
\]
Also, the marginal boundedness of $(\SC,w)$ implies that for every $\{i,j\} \in \SC(2)$, 
\begin{equation}\label{eq:trd}
\pi_2(\{i,j\}) = 2 \cdot \pi_1(i) \cdot \pi_{\{i\},1}(j) \ge 2 \bb_0 \bb_1. 
\end{equation}
Combining \cref{eq:fst,eq:scd,eq:trd}, we get \cref{eq:0alpha0}. 
The bounds on local entropy contraction rate for any $0\le k \le n-2$ and $\tau \in \SC(k)$ follows by considering the link $(\SC_\tau,w_\tau)$. 
\end{proof}

Let us now prove \cref{lem:logsobentdecay}. We follow the proof of \cite[Proposition 6]{Mic97}, using the following technical lemma.
\begin{lemma}[{\cite[Lemma 5]{Mic97}}]\label{lem:miclotechnical}
For real numbers $t \geq 0$ and $s \geq -t$, we have the inequality
\begin{align*}
    (t+s)\log(t+s) \geq t\log t + s(1 + \log t) + (\sqrt{t + s} - \sqrt{t})^{2}. 
\end{align*}
\end{lemma}

\begin{proof}[Proof of \cref{lem:logsobentdecay}]
For convenience, we write $P^\uparrow = P_1^\uparrow$, $P^\downarrow = P_2^\downarrow$, and $P^\vee = P_{2,1}^\vee = P_2^\downarrow P_1^\uparrow$ in the proof. 
Since $\pi_{1}(f^{(1)}) = \pi_{2} (f^{(2)})$ and the inequality is scale invariant, we may assume these expectations are $1$. 
Towards proving the desired contraction inequality, we first prove the following intermediate inequality: for all $\ii \in \SC(1)$,
\begin{equation}\label{eq:logsobentdecaytechnical}
    \left(P^{\uparrow}f^{(2)} \log f^{(2)}\right)(\ii) \geq \left(P^{\uparrow}f^{(2)}\right)(\ii) \cdot \log \left(P^{\uparrow}f^{(2)}\right)(\ii) 
    + \left(P^{\uparrow}f^{(2)}\right)(\ii) - \left(P^{\uparrow}\sqrt{f^{(2)}}\right)(\ii)^{2}. 
\end{equation}
Let us first see how to use this inequality to prove the desired contraction inequality. Observe that
\begin{align*}
    &\Ent_{\pi_1}(f^{(1)}) 
    \\={}& \Ent_{\pi_1}\left(P^{\uparrow}f^{(2)}\right) \\
    ={}& \sum_{\ii \in \SC(1)} \pi_1(\ii) \cdot \left(P^{\uparrow}f^{(2)}\right)(\ii) \cdot \log \left(P^{\uparrow}f^{(2)}\right)(\ii) \\
    \underset{ \text{(\ref{eq:logsobentdecaytechnical})} }{\leq}&
    \sum_{\ii \in \SC(1)} \pi_1(\ii) \left(P_{1}^{\uparrow}f^{(2)}\log f^{(2)}\right)(\ii) 
    - \sum_{\ii \in \SC(1)} \pi_1(\ii) \left(\left(P^{\uparrow}f^{(2)}\right)(\ii) - \left(P^{\uparrow}\sqrt{f^{(2)}}\right)(\ii)^{2}\right) \\
    ={}& \left\langle \mathbf{1}, P^{\uparrow} f^{(2)} \log f^{(2)} \right\rangle_{\pi_1} - \left\langle \mathbf{1}, P^{\uparrow} f^{(2)} \right\rangle_{\pi_1} + \left\langle P^{\uparrow} \sqrt{f^{(2)}}, P^{\uparrow} \sqrt{f^{(2)}} \right\rangle_{\pi_1} \\
    ={}& \left\langle \mathbf{1}, f^{(2)} \log f^{(2)} \right\rangle_{\pi_2} - \left\langle \mathbf{1}, f^{(2)} \right\rangle_{\pi_2} + \left\langle \sqrt{f^{(2)}}, P^{\vee} \sqrt{f^{(2)}} \right\rangle_{\pi_2} \\
    ={}& \Ent_{\pi_2}(f^{(2)}) - \mathcal{E}_{P^{\vee}}\left(\sqrt{f^{(2)}}, \sqrt{f^{(2)}}\right) \\
    \leq{}& \left(1 - \rho(P^{\vee})\right) \, \Ent_{\pi_2}(f^{(2)}). 
\end{align*}
All that remains is to prove \cref{eq:logsobentdecaytechnical}. For every $\ii \in \SC(1)$, taking $t = \left(P^{\uparrow}f^{(2)}\right)(\ii)$, we have that
\begin{align*}
    &\left(P^{\uparrow}f^{(2)}\log f^{(2)}\right)(\ii) \\
    &= \sum_{\sigma \in \SC(2)} P^{\uparrow}(\ii,\sigma) f^{(2)}(\sigma) \log f^{(2)}(\sigma) \\
    &= \sum_{\sigma\in \SC(2)} P^{\uparrow}(\ii,\sigma) \left( t + f^{(2)}(\sigma) - t \right) \log \left( t + f^{(2)}(\sigma) - t \right) \\&\geq \sum_{\sigma \in \SC(2)} P^{\uparrow}(\ii,\sigma) \left(t\log t + (f^{(2)}(\sigma) - t)(1 + \log t) + \left(\sqrt{f^{(2)}(\sigma)} - \sqrt{t}\right)^{2}\right) \\
    &= \left(P^{\uparrow}f^{(2)}\right)(\ii)\log \left(P^{\uparrow}f^{(2)}\right)(\ii) + \sum_{\sigma \in \SC(2)} P^{\uparrow}(\ii,\sigma) \underset{\text{Expand}}{\underbrace{\left(\sqrt{f^{(2)}(\sigma)} - \sqrt{\left(P^{\uparrow}f^{(2)}\right)(\ii)}\right)^{2}}} \\
    &= \left(P^{\uparrow}f^{(2)}\right)(\ii)\log \left(P^{\uparrow}f^{(2)}\right)(\ii) + \underset{(*)}{\underbrace{2\left(P^{\uparrow}f^{(2)}\right)(\ii) - 2\sqrt{\left(P^{\uparrow}f^{(2)}\right)(\ii)} \cdot \left(P^{\uparrow}\sqrt{f^{(2)}}\right)(\ii)}}
\end{align*}
where in the inequality, we used \cref{lem:miclotechnical}. Let us now lower bound $(*)$. We observe that
\begin{align*}
    (*) - \left(\left(P^{\uparrow}f^{(2)}\right)(\ii) - \left(P^{\uparrow}\sqrt{f^{(2)}}\right)(\ii)^{2}\right) = \left(\sqrt{\left(P^{\uparrow}f^{(2)}\right)(\ii)} - \left(P^{\uparrow}\sqrt{f^{(2)}}\right)(\ii)\right)^{2} \geq 0. 
\end{align*}
\cref{eq:logsobentdecaytechnical} then follows and we are done.
\end{proof}

\section{Spectral Independence for the Monomer-Dimer Model}\label{sec:monomerdimerspecind}
Our goal in this section is to obtain strong spectral independence bounds for the monomer-dimer model; in particular, we prove \cref{thm:spec-ind-matching}. Fix a graph $G=(V_{G},E_{G})$ and a positive real number $\lambda > 0$. We define the monomer-dimer model on $G$ to be the distribution $\mu_{G}$ on $2^{E_{G}}$ supported on all matchings of $G$, where $\mu_{G}(M) \propto \lambda^{|M|}$. We note this model may be identified with the hard-core model on the line graph $L(G)$ with parameter $\lambda$, and hence, may also be viewed as a 2-spin system with spins in $\{0,1\}$. As above, we also think of the states as being assignments $\sigma:E \rightarrow \{0,1\}$ such that $\{e \in E : \sigma(e) = 0\}$ is a matching.

For the purposes of this section, we employ an equivalent notion of pairwise influence, originally defined in \cite{ALO20} and also used in \cite{CLV20}, which is more convenient for the monomer-dimer model since there are only two spin values. We define
\begin{align*}
    \mathcal{I}_{G}^{\tau}(e \rightarrow f) = \mu_{G}(\sigma_{f} = 0 \mid \sigma_{e} = 0, \sigma_{\Lambda} = \tau) - \mu_{G}(\sigma_{f} = 0 \mid \sigma_{e} = 1, \sigma_{\Lambda} = \tau)
\end{align*}
for every $\Lambda \subseteq E_{G}$, every feasible boundary condition $\tau : \Lambda \rightarrow \{0,1\}$, and every pair of distinct edges $e,f \notin \Lambda$; we take $\mathcal{I}_{G}^{\tau}(e \rightarrow f) = 0$ if $e = f$. 
We say the Gibbs distribution $\mu_{G}$ of the monomer-dimer model is $\eta$-spectrally independent if $\lambda_{1}(\mathcal{I}_{G}^{\tau}) \leq \eta$ for all $\Lambda \subseteq E$ and all feasible $\tau : \Lambda \rightarrow \{0,1\}$, where $\mathcal{I}_{G}$ is the associated matrix of pairwise influences. We note this notion of spectral independence is equivalent to \cref{def:spectral-ind}; see \cite{ALO20}. 

We prove the following. 
Note that \cref{thm:spec-ind-matching} is a consequence of it as the maximum eigenvalue of a matrix is upper bounded by the maximum absolute row sum. 

\begin{theorem}
\label{thm:matchingsspecind}
Fix an integer $\Delta \geq 3$, and a positive real number $\lambda > 0$. Then for every graph $G=(V_{G},E_{G})$ of maximum degree at most $\Delta$ with $m = |E|$, every $\Lambda \subseteq E_{G}$, and every feasible boundary condition $\sigma_{\Lambda} : \Lambda \rightarrow \{0,1\}$ on $\Lambda \subseteq V$ with $\abs{\Lambda} = k$, the Gibbs distribution $\mu = \mu_{G,\lambda}$ of the monomer-dimer model on $G$ with fugacity $\lambda$ is $(\eta_{0},\dots,\eta_{m-2})$-spectrally independent with
\[ 
    \eta_{k} \leq \min\wrapc{2\lambda\Delta, O(\sqrt{\lambda\Delta})}.
\]
\end{theorem}
\begin{remark}
We note that \cite{AASV21} independently gave an $O(1)$-spectral independence bound for vertex-to-vertex influences using Hurwitz stability which is incomparable to our result here.
\end{remark}

Our proof follows the strategy used in \cite{ALO20, CLV20, CGSV20, FGYZ20}. Specifically, we prove \cref{thm:matchingsspecind} in two steps. In the first step, we prove a reduction for bounding the total influence of an edge in $G$ to the total influence of an edge in the associated tree of self-avoiding walks in $G$. To do this, we extend known results \cite{God93} on the univariate matching polynomial, following a similar but simpler argument used in \cite{CLV20}. In the second step, we bound the total influence of an edge in any tree of maximum degree at most $\Delta$ by leveraging the associated tree recursions. We formalize these in the following two intermediate theorems.
\begin{theorem}[Reduction from Graphs to Trees]\label{thm:matchingsgraphtreereduction}
Fix a graph $G=(V_{G},E_{G})$, $r \in V_{G}$, $e \sim r$ incident to $r$, $f \in E_{G}$, and $\lambda > 0$. Then there exists a tree $T=T_{\SAW}(G,r) = (V_{T},E_{T})$ such that the following inequality holds:
\begin{align*}
    \sum_{f \in E_{G}: f \neq e} \abs{\mathcal{I}_{G}(e \rightarrow f)} \leq \sum_{g \in E_{T}: g \neq e} \abs{\mathcal{I}_{T}(e \rightarrow g)}.
\end{align*}
\end{theorem}
\begin{theorem}[Total Edge Influence in Trees]\label{thm:matchingstreetotalinf}
Let $T = (V_{T},E_{T})$ be any tree of maximum degree $\leq \Delta$, $e \in E_{T}$ be any edge, and fix $\lambda > 0$. Then we have the bound
\begin{align*}
    \sum_{f \in E_{T} : f \neq e} \abs{\mathcal{I}_{T}(e \rightarrow f)} \leq \min\wrapc{2\lambda\Delta, 2\sqrt{1 + \lambda\Delta}}.
\end{align*}
\end{theorem}
\begin{remark}
We note that this bound on the total influence of an edge is tight for the infinite $\Delta$-regular tree. However, it turns out that bounding the maximum eigenvalue of the influence matrix using the total influence of an edge is not tight for the infinite $\Delta$-regular tree. 
\end{remark}
Assuming the truth of these two theorems, we now give a straightforward proof of \cref{thm:matchingsspecind}.
\begin{proof}[Proof of \cref{thm:matchingsspecind}]

Fix $G$, $\Lambda \subseteq E_{G}$ and $\tau : \Lambda \rightarrow \{0,1\}$. Let $H=(V_{H},E_{H})$ be the graph obtained from $G$ by deleting all edges $e \in \Lambda$ such that $\tau(e) = 1$, and deleting all edges $f \in \Lambda$ along with edges incident to them such that $\tau(f) = 0$. Observe that $H$ is a subgraph of $G$ with maximum degree at most $\Delta$, and crucially, the conditional distribution $\mu_{G}^{\tau}$ is precisely $\mu_{H}$. By \cref{thm:matchingsgraphtreereduction,thm:matchingstreetotalinf}, we have the bound
\begin{align*}
    \lambda_{1}(\mathcal{I}_{G}^{\tau}) = \lambda_{1}(\mathcal{I}_{H}) \leq \min\{2\lambda\Delta, 2\sqrt{1 + \lambda\Delta}\}.
\end{align*}
As $G,\Lambda,\tau$ were arbitrary, the claim follows.
\end{proof}
All that remains is to prove \cref{thm:matchingsgraphtreereduction,thm:matchingstreetotalinf}. We do this in \cref{subsec:matchingsgraphtreereduction,subsec:matchingslevelone}, respectively, noting that the arguments are completely independent of one another.

\subsection{Reducing Influences in Graphs to Influences in Trees: Proof of \texorpdfstring{\cref{thm:matchingsgraphtreereduction}}{Theorem 6.2}}
\label{subsec:matchingsgraphtreereduction}
Fix a graph $G=(V_{G},E_{G})$ with maximum degree $\leq\Delta$, and a vertex $r \in V_{G}$. Let $T_{\SAW}(G,r)$ denote the self-avoiding walk tree in $G$ rooted at $r$; in the context of matchings, this is known as the ``path tree'' \cite{God93}, and we refer to \cite{God93} for formal definitions. 
Note that we do not impose any boundary conditions on $T_{\SAW}(G,r)$ like in \cite{Wei06}. For every vertex $u \in V_{G}$, we write $C(u)$ to be the set of copies of $u$ in $T$. Similarly, for every edge $e \in E_{G}$, we write $C(e)$ to be the set of copies of $e$ in $T$.

We prove the following more fine-grained relationship between pairwise influences in $G$ and pairwise influences in $T = T_{\SAW}(G,r)$.
\begin{proposition}[Influence in $G$ to Influence in $T_{\SAW}(G,r)$]\label{prop:graphtotreeinf}
For every graph $G=(V_{G},E_{G})$, $r \in V_{G}$, $e \sim r$, $f \in E_{G}$ with $f \neq e$, and edge activities $x_{e} \geq 0$, if we let $T = T_{\SAW}(G,r) = (V_{T},E_{T})$, then we have the identity
\begin{align*}
    \mathcal{I}_{G}(e \rightarrow f) = \sum_{f' \in C(f)} \mathcal{I}_{T}(e \rightarrow f').
\end{align*}
\end{proposition}
We note that by the Triangle Inequality, \cref{prop:graphtotreeinf} immediately implies \cref{thm:matchingsgraphtreereduction}. Hence, it suffices to prove \cref{prop:graphtotreeinf}, which we do by generalizing properties of the univariate matching polynomial.

Define the following multivariate edge-matching polynomial.
\begin{align*}
    \M_{G}(x_{e} : e \in E_{T}) = \sum_{M \subseteq E \text{ matching}} \prod_{e \in M} x_{e}.
\end{align*}
$\M_{G}$ is also the partition function of the monomer-dimer model on $G$ with edge activities $x_{e} \geq 0$. Furthermore, if $r \in V_{G}$ is arbitrary, and we denote $T = T_{\SAW}(G,r)$, then define
\begin{align*}
    \overline{\M}_{T}(x_{e} : e \in E_{G}) \overset{\defin}{=} \M_{T}(\overline{x}_{f} : f \in E_{T})
\end{align*}
where $\overline{x}_{f} = x_{e}$ for all $f \in C(e)$ and all $e \in E_{G}$. We note that while $\M_{G}$ is always multiaffine, $\overline{\M}_{T}$ is not. Furthermore, $\M_{G}$ is not homogeneous. Finally, note that the degree of any edge $e \sim r$ incident to $r$ is $1$ in $\overline{\M}_{T}$ since no self-avoiding walk can reuse $e$ after using $e$ to leave $r$. We will crucially need the following decomposition of $\M_{G}$.
\begin{lemma}\label{lem:matchingpolyrecursion}
For every graph $G=(V_{G},E_{G})$ and any vertex $v \in V_{G}$, we have the identity
\begin{align*}
    \M_{G}(x) = \M_{G - r}(x) + \sum_{v \sim r} x_{rv} \M_{G-r-v}(x).
\end{align*}
\end{lemma}
\begin{proof}
Group the matchings for which $r$ is not saturated in the term $\M_{G-r}(x)$. Similarly, group the matchings for which a fixed edge $e = \{r,v\}$ incident to $r$ is selected in $\M_{G-r-v}(x)$.
\end{proof}

We prove the following, a univariate analog of which was already proved in \cite{God93}.
\begin{lemma}\label{lem:multivardivis}
For every graph $G=(V_{G},E_{G})$ and $r \in V_{G}$, taking $T = T_{\SAW}(G,r)$, we have the identity
\begin{align*}
    \frac{\M_{G}(x)}{\M_{G-r}(x)} = \frac{\overline{\M}_{T}(x)}{\overline{\M}_{T - r}(x)}.
\end{align*}
Furthermore, we may write $\overline{\M}_{T}(x) = \M_{G}(x) \cdot q(x)$ for some polynomial $q$ which does not depend on $x_{e}$ for any $e \sim r$.
\end{lemma}
First, let us see how to use \cref{lem:multivardivis} to prove \cref{prop:graphtotreeinf}.
\begin{proof}[Proof of \cref{prop:graphtotreeinf}]
Fix $r \in V_{G}$, and write $T = T_{\SAW}(G,r)$. By \cref{lem:multivardivis}, we have that $\overline{\M}_{T}(x) = \M_{G}(x) \cdot q(x)$ for a polynomial $q$ which does not depend on $x_{e}$ for all $e \sim r$. It follows that if $e \sim r$, then
\begin{align*}
    \mu_{T}(\sigma_{e} = 0) = (x_{e}\partial_{x_{e}}\log \overline{\M}_{T})(x) = (x_{e}\partial_{x_{e}}\log \M_{G})(x) = \mu_{G}(\sigma_{e} = 0).
\end{align*}
It then also follows that for any $e \sim r$ and any edge $f \in E_{G}$, we have the identity
\begin{align*}
    (x_{f}x_{e}\cdot \partial_{x_{f}}\partial_{x_{e}} \log \overline{\M}_{T})(x) = (x_{f}x_{e}\cdot \partial_{x_{f}}\partial_{x_{e}} \log \M_{G})(x).
\end{align*}
Now, let us understand the left-hand and right-hand sides separately as influences. For the right-hand side, we have that
\begin{align*}
    &(x_{f}x_{e}\cdot \partial_{x_{f}}\partial_{x_{e}} \log \M_{G})(x) \\
    &= x_{e}x_{f}\cdot \partial_{x_{f}} \frac{(\partial_{x_{e}}\M_{G})(x)}{\M_{G}(x)} \\
    &= x_{f}x_{e} \cdot\wrapp{\frac{(\partial_{x_{f}}\partial_{x_{e}}\M_{G})(x)}{\M_{G}(x)} - \frac{(\partial_{x_{f}}\M_{G})(x) \cdot (\partial_{x_{e}}\M_{G})(f)}{\M_{G}(x)^{2}}} \\
    &= \mu_{G}(\sigma_{e} = 0, \sigma_{f} = 0) - \mu_{G}(\sigma_{e} = 0) \cdot \mu_{G}(\sigma_{f} = 0) \\
    &= \mu_{G}(\sigma_{e} = 0) \cdot \wrapp{\mu_{G}(\sigma_{f} = 0 \mid \sigma_{e} = 0) - \mu_{G}(\sigma_{f} = 0)} \\
    &= \mu_{G}(\sigma_{e} = 0) \cdot \mu_{G}(\sigma_{e}=1) \cdot \mathcal{I}_{G}(e \rightarrow f).
\end{align*}
For the left-hand side, we have by the Chain Rule that
\begin{align*}
    (x_{f}x_{e}\cdot \partial_{x_{f}}\partial_{x_{e}} \log \overline{\M}_{T})(x) &= x_{f}x_{e} \cdot \partial_{x_{f}} \frac{(\partial_{x_{e}}\overline{\M}_{T})(x)}{\overline{\M}_{T}(x)} \\
    &= \sum_{f' \in C(f)} x_{f}x_{e} \cdot\partial_{\overline{x}_{f'}} \frac{(\partial_{x_{e}}\M_{T})(\overline{x})}{\M_{T}(\overline{x})} \Bigg|_{\overline{x} = x} \\
    &= \sum_{f' \in C(f)} \mu_{T}(\sigma_{e} = 0) \cdot \mu_{T}(\sigma_{e} = 1) \cdot \mathcal{I}_{T}(e \rightarrow f').
\end{align*}
Since $\mu_{G}(\sigma_{e} = 0) = \mu_{T}(\sigma_{e} = 0)$ and $\mu_{G}(\sigma_{e} = 1) = \mu_{T}(\sigma_{e} = 1)$, the claim follows.
\end{proof}
All that remains is to prove \cref{lem:multivardivis}.
\begin{proof}[Proof of \cref{lem:multivardivis}]
We go by induction on the graph. First, we note that the claim is trivial in the case where $G$ itself is a tree, since then $T = G$ and $\M_{G} = \overline{\M}_{T}$ (i.e. $q$ is identically $1$). This forms our base case. 
Now, 
by \cref{lem:matchingpolyrecursion} we may write
\begin{align*}
    \M_{G}(x) = \M_{G - r}(x) + \sum_{v \sim r} x_{rv} \M_{G-r-v}(x)
\end{align*}
where we note the polynomials $\M_{G-r}(x), \M_{G-r-v}(x)$ do not depend on any $x_{e}$ for $e \sim r$. 
Therefore, we deduce that
\begin{align*}
    \frac{\M_{G}(x)}{\M_{G-r}(x)} &= 1 + \sum_{v \sim r} x_{rv} \cdot \frac{\M_{G-r-v}(x)}{\M_{G-r}(x)} \\
    &= 1 + \sum_{v \sim r} x_{rv} \cdot \frac{\overline{\M}_{T_{\SAW}(G-r, v) - v}(x)}{\overline{\M}_{T_{\SAW}(G-r,v)}(x)} \tag{Induction} \\
    &= 1 + \sum_{v \sim r} x_{rv} \cdot \frac{\overline{\M}_{T_{\SAW}(G,r) - r - v}(x)}{\overline{\M}_{T_{\SAW}(G,r) - r}(x)}
\end{align*}
For the last step, note that $T_{\SAW}(G-r,v)$ is the subtree of $T_{\SAW}(G,r)$ rooted at $v$ since $v$ is adjacent to $r$; the matching polynomials of the other subtrees all ``cancel''. Therefore, we have the following: 
\begin{align*}
    &1 + \sum_{v \sim r} x_{rv} \cdot \frac{\overline{\M}_{T_{\SAW}(G,r) - r - v}(x)}{\overline{\M}_{T_{\SAW}(G,r) - r}(x)} \\
    ={}& \frac{\overline{\M}_{T_{\SAW}(G,r)-r}(x) + \sum_{v \sim r} x_{rv} \overline{\M}_{T_{\SAW}(G,r) - r - v}(x)}{\overline{\M}_{T_{\SAW}(G,r)-r}(x)} \\
    ={}& \frac{\overline{\M}_{T_{\SAW}(G,r)}(x)}{\overline{\M}_{T_{\SAW}(G,r)-r}(x)}.
\end{align*}
This proves the first claim. For the second claim, we go by induction again. The base case where $G$ is a tree is again immediate. For the inductive step, we have
\begin{align*}
    \overline{\M}_{T_{\SAW}(G,r)}(x) &= \M_{G}(x) \cdot \frac{\overline{\M}_{T_{\SAW}(G,r)-r}(x)}{\M_{G-r}(x)} \\
    &= \M_{G}(x) \cdot \frac{\prod_{v \sim u} \overline{\M}_{T_{\SAW}(G-r, v)}(x)}{\M_{G-r}(x)} \\
    &= \M_{G}(x) \cdot q(x)
\end{align*}
where in the second step we use that deleting $r$ from $T_{\SAW}(G,r)$ disconnects the subtrees $T_{\SAW}(G-r,v)$, and in the final step we use that $\M_{G-r}(x)$ divides each $\overline{\M}_{T_{\SAW}(G-r,v)}(x)$ by the induction hypothesis. This shows the lemma. 
\end{proof}

\subsection{The Total Influence in a Tree: Proof of \texorpdfstring{\cref{thm:matchingstreetotalinf}}{Theorem 6.3}}
\label{subsec:matchingslevelone}
At this point, we can forget self-avoiding walk trees, and just focus on the special case where $G$ itself is a tree $T$. Throughout, we assume our tree $T$ has maximum degree at most $\Delta$. If $e \in T$ with endpoints $r_{1},r_{2}$, then we may view $T$ as two trees $T(r_{1}), T(r_{2})$ on disjoint sets of vertices which are connected by the edge $e$, with $T(r_{1})$ being rooted at $r_{1}$ and $T(r_{2})$ being rooted at $r_{2}$. If $v$ is a vertex in $T(r_{1})$ (resp. $T(r_{2})$), we write $T(v)$ for the subtree of $T(r_{1})$ (resp. $T(r_{2})$) rooted at $v$. We also write $L_{v}(k)$ for the set of descendants of $v$ (in $T(v)$) at distance exactly $k$ from $v$.

We will let $\mu_{T}(r)$ to denote the marginal probability that the vertex $r \in V_{T}$ is saturated in a random matching drawn from the Gibbs distribution, and let $\mu_{T}(\overline{r}) := 1 - \mu_{T}(r)$. Similarly, for an edge $e \in E_{T}$, we write $\mu_{T}(e)$ for the marginal probability that $e$ is in a random matching drawn from the Gibbs distribution, and let $\mu_{T}(\overline{e}) := 1 - \mu_{T}(e)$. Note that since at most one edge incident to $r$ may be in any given matching, we have
\begin{align}\label{eq:matching-vertex-saturation-edge-marginal-sum}
    \mu_{T}(r) = \sum_{e \sim r} \mu_{T}(e).
\end{align}

We now prove several intermediate technical results which we will use to deduce \cref{thm:matchingstreetotalinf}. To state them, we will need the following recursion for the probabilities $\mu_{G}(\overline{r})$ in the monomer-dimer model on $G$ with fugacity $\lambda \geq 0$:
\begin{align}\label{eq:matchings-tree-prob-recursion}
    \mu_{G}(\overline{r}) = F_{\lambda}(\mu_{G-r}(\overline{v}) : v \in L_{r}(1)) \quad\text{where}\quad F_{\lambda}(p) = F_{d,\lambda}(p) := \frac{1}{1 + \lambda \sum_{i=1}^{d} p_{i}}.
\end{align}
We note this is immediate from \cref{lem:matchingpolyrecursion} and using that $\mu_{G}(\overline{r}) = \frac{\M_{G-r}(\lambda)}{\M_{G}(\lambda)}$ where we take $\lambda = \lambda\mathbf{1}$.

\begin{proposition}\label{prop:potentialbound}
Consider the potential function $\Phi(x) := \log x$. Then for every tree $T=(V_{T},E_{T})$, every edge $e = \{r_{1},r_{2}\} \in E_{T}$, every positive integer $k \geq 1$, and every $r \in e = \{r_{1},r_{2}\}$, we have the inequality
\begin{align*}
    &\sum_{f \in E_{T(r)} : \dist(e,f) = k} \abs{\mathcal{I}_{T}(e \to f)} \\
    &\leq \min\wrapc{\wrapp{\frac{\lambda\Delta}{1 + \lambda\Delta}}^{k}, \wrapp{\frac{\lambda\Delta}{1+\lambda\Delta}}^{k \bmod 2} \cdot \wrapp{\sup_{y} \norm{\nabla (\Phi \circ F_{\lambda}^{\circ 2} \circ \Phi^{-1})(y)}_{1}}^{\floor{k/2}}}.
\end{align*}
\end{proposition}
Given \cref{prop:potentialbound}, we will also need a bound on the gradient norm. Conveniently, this gradient norm was already analyzed in \cite{BGKNT07} to establish the correlation decay property.
\begin{lemma}[{\cite[Lemma 3.3]{BGKNT07}}]\label{lem:gradientnormbound}
Consider the potential function $\Phi(x) = \log x$. Then we have the following bound on the norm of the gradient for the two-step log-marginal recursion:
\begin{align*}
    \sup_{y} \norm{\nabla (\Phi \circ F_{\lambda}^{\circ 2} \circ \Phi^{-1})(y)}_{1} \leq 1 - \frac{2}{\sqrt{1 + \lambda \Delta} + 1}.
\end{align*}
\end{lemma}
We now show how to use intermediate technical results to prove \cref{thm:matchingstreetotalinf}. Immediately following, we will prove \cref{prop:potentialbound}.

\begin{proof}[Proof of \cref{thm:matchingstreetotalinf}]
By combining \cref{prop:potentialbound,lem:gradientnormbound}, we have the following two bounds
\begin{align*}
    \sum_{f \in E_{T} : f \neq e} \abs{\mathcal{I}_{T}(e \rightarrow f)} &= \wrapp{\sum_{f \in E_{T(r_{1})}} + \sum_{f \in E_{T(r_{2})}}} \abs{\mathcal{I}_{T}(e \rightarrow f)} \\
    &\leq 2\sum_{k=1}^{\infty} \wrapp{1 - \frac{2}{\sqrt{1 + \lambda \Delta} + 1}}^{\floor{k/2}} \cdot \wrapp{\frac{\lambda\Delta}{1+\lambda\Delta}}^{k \bmod 2} \\
    &= \frac{\lambda\Delta}{1+\lambda\Delta} \cdot \wrapp{\sqrt{1 + \lambda\Delta} + 1} + \wrapp{\sqrt{1+\lambda\Delta}-1} \\
    &\leq 2\sqrt{1 + \lambda\Delta} \\
    \sum_{f \in E_{T} : f \neq e} \abs{\mathcal{I}_{T}(e \rightarrow f)} &\leq 2\sum_{k=1}^{\infty} \wrapp{\frac{\lambda\Delta}{1+\lambda\Delta}}^{k} = 2 \cdot \frac{\lambda\Delta}{1+\lambda\Delta} \cdot \frac{1}{1 - \frac{\lambda\Delta}{1 + \lambda\Delta}} = 2\lambda\Delta.
\end{align*}
This proves the theorem. 
\end{proof}

\subsubsection{Proof of \texorpdfstring{\cref{prop:potentialbound}}{Proposition 6.8}}
The crux of the proof rests on the following two lemmas. The first is a convenient factorization of the pairwise influence in trees, which was already observed in prior work \cite{ALO20, CLV20} in the context of vertex-spin systems.
\begin{lemma}[Factorization of Pairwise Influence in Trees]\label{lem:matchings-tree-influence-product}
Fix two edges $e,f \in E_{T}$. Let $e = e_{1},r = u_{1},e_{2},u_{2},\dots,u_{k},e_{k+1} = f$ be the unique path in $T$ from $e$ to $f$, where edge $e_{i}$ connects vertices $u_{i-1}$ and $u_{i}$. Then we have
\begin{align*}
    \mathcal{I}_{T}(e \rightarrow f) = \prod_{i=1}^{k} \mathcal{I}_{T}(e_{i} \rightarrow e_{i+1}).
\end{align*}
\end{lemma}
\begin{proof}[Proof of \cref{lem:matchings-tree-influence-product}]
It suffices to show that if $g$ is any edge on the unique path from $e$ to $f$, then $\mathcal{I}_{T}(e \rightarrow f) = \mathcal{I}_{T}(e \rightarrow g) \cdot \mathcal{I}_{T}(g \rightarrow f)$; the full claim then follows by induction. This simpler identity follows immediately from the fact that conditioning on $g$ disconnects $e$ from $f$ so that they become independent.
\end{proof}
\begin{lemma}\label{lem:matchings-edge-marginal-product-vertex-marginals}
Let $T$ be a tree rooted at $r$ and let $v \in L_{r}(1)$. Then we have the identity
\begin{align*}
    \lambda \cdot \mu_{T}(\overline{r}) \cdot \mu_{T(v)}(\overline{v}) = \mu_{T}(\{r,v\}).
\end{align*}
\end{lemma}
\begin{proof}
More generally, if $G$ is any graph (not necessarily a tree), and $e = \{r,v\}$, then
\begin{align*}
    \mu_{G}(\{r,v\}) = \lambda \cdot \mu_{G}(\overline{r} \wedge \overline{v}) = \lambda \cdot \mu_{G}(\overline{r}) \cdot \mu_{G - r}(\overline{v}),
\end{align*}
because any matching in $G$ containing $\{r,v\}$ is precisely a matching in $G - r - v$ adjoined with the edge $\{r,v\}$ (which introduces the multiplicative factor of $\lambda$ in the right-hand side). If $G$ is a tree $T$, then $T - r$ is the union of subtrees which are disconnected from each other. Hence, $\mu_{T-r}(\overline{v}) = \mu_{T(v)}(\overline{v})$.
\end{proof}

We are now ready to prove \cref{prop:potentialbound}. First, observe that if $e,f$ are neighboring edges in $T$, then $\mathcal{I}_{T}(e \to f) = - \mu_{T - e}(f)$ due to the hard constraints of the monomer-dimer model. Hence,
\begin{align*}
    \sum_{f \in E_{T(r)} : \dist(e,f) = k} \abs{\mathcal{I}_{T}(e \to f)} &= \sum_{f \in E_{T(r)} : \dist(e,f) = k} \prod_{i=1}^{k} \abs{\mathcal{I}_{T}(e_{i} \to e_{i+1})} \tag{\cref{lem:matchings-tree-influence-product}} \\
    &= \sum_{f \in E_{T(r)} : \dist(e,f) = k} \prod_{i=1}^{k} \mu_{T - e_{i}}(e_{i+1}) \\
    &= \sum_{u_{i+1} \in L_{u_{i}}(1), \forall i\in[k]} \prod_{i=1}^{k} \mu_{T(u_{i})}(\{u_{i},u_{i+1}\}) \\
    &= \sum_{u_{i+1} \in L_{u_{i}}(1), \forall i\in[k]} \prod_{i=1}^{k} \lambda \cdot \mu_{T(u_{i})}(\overline{u_{i}}) \cdot \mu_{T(u_{i+1})}(\overline{u_{i+1}}), \tag{\cref{lem:matchings-edge-marginal-product-vertex-marginals}}
\end{align*}
where again we write $e = e_{1},r = u_{1},e_{2},u_{2},\dots,u_{k},e_{k+1} = f$ for the unique path in $T$ from $e$ to $f$, for each $f$.

Note that for the first bound $\wrapp{\frac{\lambda\Delta}{1+\lambda\Delta}}^{k}$, one could have stopped at the penultimate step. To see this, observe that by directly applying the tree recursion $F_{\lambda}$ from \cref{eq:matchings-tree-prob-recursion},
\begin{align*}
    \mu_{T}(\overline{r}) = \frac{1}{1 + \lambda \sum_{v \in L_{r}(1)} \mu_{T(v)}(\overline{v})} \geq \frac{1}{1 + \lambda \Delta}
\end{align*}
since $\mu_{T(v)}(\overline{v}) \leq 1$ trivially for all $v \in L_{r}(1)$. Hence $\mu_{T}(r) \leq \frac{\lambda \Delta}{1 + \lambda \Delta}$. This upper bound applies to general rooted trees, in particular, the root $u_{i}$ of $T(u_{i})$ for each $i=1,\dots,k$. 
It follows from induction that
\begin{align*}
    &\sum_{u_{i+1} \in L_{u_{i}}(1), \forall i\in[k]} \prod_{i=1}^{k} \mu_{T(u_{i})}(\{u_{i},u_{i+1}\}) \\
    &= \sum_{u_{i+1} \in L_{u_{i}}(1), \forall i\in[k-1]} \prod_{i=1}^{k-1} \mu_{T(u_{i})}(\{u_{i},u_{i+1}\}) \underset{= \mu_{T(u_{k})}(u_{k})}{\underbrace{\sum_{u_{k+1} \in L_{u_{k}}(1)} \mu_{T(u_{k})}(\{u_{k},u_{k+1}\})}} \\
    &\leq \frac{\lambda\Delta}{1+\lambda\Delta} \sum_{u_{i+1} \in L_{u_{i}}(1), \forall i\in[k-1]} \prod_{i=1}^{k-1} \mu_{T(u_{i})}(\{u_{i},u_{i+1}\}) \\
    &\leq \wrapp{\frac{\lambda\Delta}{1+\lambda\Delta}}^{k}.
\end{align*}
This gives the first bound in \cref{prop:potentialbound}, which is already sufficient for \cref{thm:matching}. For completeness, we now prove the second bound, with a similar but more technical argument, which is better when $\lambda\Delta$ is sufficiently large.

It turns out, one may view the preceding simple analysis as a ``one-step'' analysis, in the sense that we only applied the tree recursion $F_{\lambda}$ one step at a time. We will establish the second upper bound via a ``two-step analysis''.

We have shown the identity above
\begin{align}\label{eq:matchings-infprod-vertex-marginals}
    \sum_{f \in E_{T(r)} : \dist(e,f) = k} \abs{\mathcal{I}_{T}(e \to f)} = \sum_{u_{i+1} \in L_{u_{i}}(1), \forall i\in[k]} \prod_{i=1}^{k} \lambda \cdot \mu_{T(u_{i})}(\overline{u_{i}}) \cdot \mu_{T(u_{i+1})}(\overline{u_{i+1}}).
\end{align}
Let us now reinterpret the terms $\lambda \cdot \mu_{T(u_{i})}(\overline{u_{i}}) \cdot \mu_{T(u_{i+1})}(\overline{u_{i+1}})$ appearing in the right-hand side of \cref{eq:matchings-infprod-vertex-marginals} as \emph{derivatives} of the tree recursion $F_{\lambda}$ after composing with our potential $\Phi$. Composing with the recursion $F_{\lambda}$ yields the following recursion for the logarithm of the marginals $y_{i} = \log p_{i} = \Phi(p_{i})$:
\begin{align*}
    (\Phi \circ F_{\lambda} \circ \Phi^{-1})(y) &= -\log\wrapp{1 + \lambda \sum_{i=1}^{d} \exp(y_{i})}.
\end{align*}
Differentiating and applying the Inverse Function Theorem, we obtain the the identities
\begin{align*}
    (\partial_{p_{i}}F_{\lambda})(p) &= -\frac{\lambda}{\wrapp{1 + \lambda \sum_{j=1}^{d}p_{j}}^{2}} = -\lambda F_{\lambda}(p)^{2}; \\
    \partial_{y_{i}}(\Phi \circ F_{\lambda} \circ \Phi^{-1})(y) &= \frac{(\Phi' \circ F_{\lambda})(p)}{\Phi'(p_{i})} \cdot (\partial_{p_{i}}F_{\lambda})(p) = -\frac{p_{i}}{F_{\lambda}(p)} \cdot \lambda F_{\lambda}(p)^{2} \\
    &= -\lambda \cdot p_{i} \cdot F_{\lambda}(p).
\end{align*}
Since for any $y = \log p$, we have
\begin{align*}
    \norm{\nabla (\Phi \circ F_{\lambda} \circ \Phi^{-1})(y)}_{1} &= F_{\lambda}(p) \cdot \lambda \sum_{i=1}^{d} p_{i} \\
    &= \frac{\lambda \sum_{i=1}^{d} p_{i}}{1 + \lambda \sum_{i=1}^{d} p_{i}} \\
    &= 1 - F_{\lambda}(p) \\
    &\leq \frac{\lambda\Delta}{1+\lambda\Delta},
\end{align*}
we see immediately see via the same kind of inductive argument that the right-hand side of \cref{eq:matchings-infprod-vertex-marginals} is upper bounded by $\wrapp{\frac{\lambda\Delta}{1+\lambda\Delta}}^{k}$. However, we can do better. Following \cite{BGKNT07}, we consider the two-step recursion $F_{\lambda}^{\circ 2}$, which admits faster contraction rates when $\lambda\Delta \geq \frac{3}{4}$ by \cref{lem:gradientnormbound}. First, we calculate that
\begin{align*}
    (\Phi \circ F_{\lambda}^{\circ 2} \circ \Phi^{-1})(y) &= -\log \wrapp{1 + \lambda \sum_{i=1}^{d} \frac{1}{1 + \lambda \sum_{j=1}^{d} \exp(y_{ij})}} \\
    \partial_{y_{ij}}(\Phi \circ F_{\lambda}^{\circ 2} \circ \Phi^{-1})(y) &= \frac{(\Phi' \circ F_{\lambda}^{\circ 2})(p)}{\Phi'(p_{ij})} \cdot (\partial_{p_{ij}} F_{\lambda}^{\circ 2})(y) \\
    &= \frac{p_{ij}}{F_{\lambda}^{\circ 2}(p)} \cdot ((\partial_{p_{i}}F_{\lambda}) \circ F_{\lambda})(p) \cdot (\partial_{p_{j}}F_{\lambda})(p_{i}) \\
    &= \frac{p_{ij}}{F_{\lambda}^{\circ 2}(p)} \cdot \wrapp{-\lambda (F_{\lambda} \circ F_{\lambda})(p)^{2}} \cdot \wrapp{-\lambda F_{\lambda}(p_{i})^{2}} \\
    &= \wrapp{\lambda \cdot p_{ij} \cdot F_{\lambda}(p_{i})} \cdot \wrapp{\lambda \cdot F_{\lambda}(p_{i}) \cdot F_{\lambda}^{\circ 2}(p)},
\end{align*}
where we imagine that $j$ is a child of $i$ which is a child of the root. With this, we can then peel off two levels at a time. In particular, by induction the right-hand side of \cref{eq:matchings-infprod-vertex-marginals} satisfies
\begin{align*}
    &\sum_{u_{i+1} \in L_{u_{i}}(1), \forall i\in[k]} \prod_{i=1}^{k} \lambda \cdot \mu_{T(u_{i})}(\overline{u_{i}}) \cdot \mu_{T(u_{i+1})}(\overline{u_{i+1}}) \\
    &= \sum_{u_{i+1} \in L_{u_{i}}(1), \forall i\in[k-2]} \prod_{i=1}^{k-2} \lambda \cdot \mu_{T(u_{i})}(\overline{u_{i}}) \cdot \mu_{T(u_{i+1})}(\overline{u_{i+1}}) \\
    &\qquad \times \underset{=\norm{\nabla (\Phi \circ F_{\lambda}^{\circ 2} \circ \Phi^{-1}(y)}_{1} \text{ where } y \text{ is the logarithm of the } \mu_{T(u_{k+1})}(\overline{u_{k+1}})}{\underbrace{\wrapp{\sum_{\substack{u_{k} \in \\ L_{u_{k-1}}(1)}} \sum_{\substack{u_{k+1} \in \\ L_{u_{k}}(1)}} \wrapp{\lambda\mu_{T(u_{k-1})}(\overline{u_{k-1}})\mu_{T(u_{k})}(\overline{u_{k}})} \wrapp{\lambda\mu_{T(u_{k})}(\overline{u_{k}})\mu_{T(u_{k+1})}(\overline{u_{k+1}})}}}} \\
    &\leq \wrapp{\sup_{y} \norm{\nabla (\Phi \circ F_{\lambda}^{\circ 2} \circ \Phi^{-1})(y)}_{1}} \cdot \sum_{\substack{u_{i+1} \in L_{u_{i}}(1) \\ \forall i\in[k-2]}} \prod_{i=1}^{k-2} \lambda \cdot \mu_{T(u_{i})}(\overline{u_{i}}) \cdot \mu_{T(u_{i+1})}(\overline{u_{i+1}}) \\
    &\leq \wrapp{\sup_{y} \norm{\nabla (\Phi \circ F_{\lambda}^{\circ 2} \circ \Phi^{-1})(y)}_{1}}^{\floor{k/2}} \cdot \wrapp{\frac{\lambda\Delta}{1+\lambda\Delta}}^{k \bmod 2}
\end{align*}
as desired.

\section{Proofs of Main Results}
\label{app:proof-main}
In this section we prove our main results \cref{thm:2-spin,thm:hard-core,thm:ising,thm:coloring,thm:matching}. 

By \cref{thm:main}, to establish optimal mixing time bound it suffices to show marginal boundedness and spectral independence for the corresponding Gibbs distribution.

We first consider antiferromagnetic 2-spin systems. 
Let $\beta,\gamma,\lambda$ be reals such that $0\le \beta \le \gamma$, $\gamma > 0$, $\beta \gamma < 1$ and $\lambda > 0$ so the triple $(\beta,\gamma,\lambda)$ specifies parameters of an antiferromagnetic 2-spin system. 
We state here the formal definition of up-to-$\Delta$ uniqueness with gap $\delta$ given in \cite{LLY13}.
\begin{definition}[Up-to-$\Delta$ uniqueness with gap $\delta$, \cite{LLY13}]
\label{def:2spin-uniqueness}
For each $1\le d < \Delta$ define
\[
f_d(R) = \lambda \left( \frac{\beta R + 1}{R+\gamma} \right)^d
\]
and denote the unique fixed point of $f_d$ by $R^*_d$. 
We say the parameters $(\beta,\gamma,\lambda)$ are \emph{up-to-$\Delta$ unique with gap $\delta$} if $|f'_d(R^*_d)| < 1-\delta$ for all $1\le d < \Delta$. 
\end{definition}

\begin{proof}[Proof of \cref{thm:2-spin}]
The proof of Theorem 3 from \cite{CLV20} showed that for antiferromagnetic $2$-spin systems that are up-to-$\Delta$ unique with gap $\delta$, the Gibbs distribution $\mu$ is $O(1/\delta)$-spectrally independent. 
Also, by considering the worst configuration of the neighborhood for soft-constraint models (i.e., $0 < \beta \le \gamma$) or $2$-hop neighborhood for hard-constraint models (i.e., $0 = \beta < \gamma$), one can check that $\mu$ is $\bb$-marginally bounded for some constant $\bb = \bb(\Delta,\beta,\gamma,\lambda)$. 
The theorem then follows from \cref{thm:main}. 
\end{proof}

Though in general the constant $C = C(\Delta,\delta,\beta,\gamma,\lambda)$ for bounding the mixing time depends on the parameters $(\beta,\gamma,\lambda)$ of the model, in most applications such as the hard-core model (\cref{thm:hard-core}) and the Ising model (\cref{thm:ising}) we can make the constant $C$ independent of all parameters. 
This is achieved by considering separately when the parameters are pretty far away from the uniqueness threshold, in which case we can deduce rapid mixing under the \emph{Dobrushin uniqueness condition}~\cite{Dob70}, see also~\cite{BD97i,Weitz2}. 

\begin{lemma}
Consider an arbitrary distribution $\mu$ over $[q]^V$. 
For two distinct vertices $u,v \in V$, define
\[
R(u,v) = \max_{\substack{\tau, \xi \in \Omega_{V \setminus \{v\}} \\ \mathrm{Dif}(\tau, \xi) = \{u\}}} \TV{\mu(\sigma_v \seq \cdot \mid \sigma_{V \setminus \{v\}} \seq \tau)}{\mu(\sigma_v \seq \cdot \mid \sigma_{V \setminus \{v\}} \seq \xi)}
\]
where $\mathrm{Dif}(\tau,\xi) = \{w \in V: \tau_w \neq \xi_w\}$. 
If there exists $c \in (0,1)$ such that for every vertex $v \in V$ we have
\[
\sum_{u \in V \setminus \{v\}} R(u,v) \le 1-c
\]
(in which case we say the \emph{Dobrushin uniqueness condition} holds with constant $c$), then the mixing time of the Glauber dynamics for sampling from $\mu$ satisfies
\[
T_{\mathrm{mix}}(\Pgl, \eps) \le \frac{n}{c} \log\left( \frac{n}{\eps} \right). 
\]
\end{lemma}

We present next the proofs of \cref{thm:hard-core,thm:ising}. 

\begin{proof}[Proof of \cref{thm:hard-core}]
By \cref{thm:2-spin}, for every $\lambda \le (1-\delta) \lambda_c(\Delta)$ there exists $C = C(\Delta,\delta,\lambda)$ such that the Glauber dynamics mixes in $C n \log(n/\eps)$ steps. 
Meanwhile, it is easy to check that, when $\lambda \le \frac{1}{2\Delta}$ the Dobrushin uniqueness condition holds with $c = 1/2$, and thus the mixing time is upper bounded by $2n \log(n/\eps)$. 
If we take
\[
C' = C'(\Delta,\delta) := \max\left\{ 2, \sup_{\frac{1}{2\Delta} < \lambda \le (1-\delta) \lambda_c(\Delta)} C(\Delta,\delta,\lambda) \right\},
\]
then the mixing time of the Glauber dynamics is at most $C' n\log(n/\eps)$, as claimed. 
\end{proof}

\begin{proof}[Proof of \cref{thm:ising}]
Consider first the antiferromagnetic Ising model ($\beta = \gamma < 1$) and by symmetry we may assume $\lambda \le 1$. 
It is shown in \cite{CLV20} that the Gibbs distribution $\mu$ is $O(1/\delta)$-spectrally independent in this case, and by considering the worst neighborhood configuration one can check that $\mu$ is $\bb$-marginally bounded for
\[
\bb = \min \left\{ \frac{\lambda \beta^\Delta}{\lambda \beta^\Delta + 1}, \frac{\lambda^{-1} \gamma^\Delta}{\lambda^{-1} \gamma^\Delta + 1} \right\} = \frac{\lambda \beta^\Delta}{\lambda \beta^\Delta + 1} > \frac{\lambda \left(\frac{\Delta-2}{\Delta}\right)^\Delta}{\lambda \left(\frac{\Delta-2}{\Delta}\right)^\Delta + 1} \ge \frac{\lambda}{28}. 
\]
Thus, when $\lambda \ge 1/500$, \cref{thm:main} implies that the mixing time of the Glauber dynamics is at most $\Delta^{O(1/\delta)} n \log(n/\eps)$ for large enough $n$. 
Meanwhile, if $\lambda < 1/500$ then one can check that the Dobrushin uniqueness condition holds with $c = 1/2$, and thus the mixing time is upper bounded by $2n \log(n/\eps)$. 
This proves the theorem for the antiferromagnetic case. 

Next, consider the ferromagnetic Ising model ($\beta = \gamma > 1$). Assume $\lambda \ge 1$ for convenience. 
The Gibbs distribution $\mu$ is $O(1/\delta)$-spectrally independent by Theorem 26 of \cite{CLV20} and $\bb$-marginally bounded for
\[
\bb = \min \left\{ \frac{1}{\lambda \beta^\Delta + 1}, \frac{1}{\lambda^{-1} \gamma^\Delta + 1} \right\} = \frac{1}{\lambda \beta^\Delta + 1} > \frac{1}{\lambda (\frac{\Delta}{\Delta-2})^\Delta + 1} \ge \frac{1}{28 \lambda }. 
\]
If $\lambda \le 500$ the mixing time is $\le \Delta^{O(1/\delta)} n \log(n/\eps)$ by \cref{thm:main}, and if $\lambda > 500$ the mixing time is $\le 2n \log(n/\eps)$ by the Dobrushin uniqueness condition. This shows the ferromagnetic case, and completes the proof of the theorem. 
\end{proof}

For random colorings, we can use the same argument. 
\begin{proof}[Proof of \cref{thm:coloring}]
\cite{FGYZ20} showed that the uniform distribution $\mu$ of colorings is $O(1/\delta)$-spectrally independent under our assumption. 
(Note that the notion of spectral independence in \cite{FGYZ20} implies the one in \cite{CGSV20} which is \cref{def:spectral-ind}; see Lemma 3.6 of \cite{FGYZ20} and Theorem 8 of \cite{CGSV20}; also, \cite{FGYZ20} gave a better bound on the spectral independence constant and applicable to a slightly larger parameter region). 
Also, the proof of Lemma 3 from \cite{GKM15} can be adapted to show that $\mu$ is $\Omega(1/q)$-marginally bounded. 
Hence, \cref{thm:main} implies that the mixing time of the Glauber dynamics is at most $C n\log(n/\eps)$ for some $C = C(\Delta,\delta,q)$. 
Notice that when $q \ge 3\Delta$, the Dobrushin uniqueness condition holds with $c = 1/2$ and thus the mixing time is at most $2n \log(n/\eps)$. 
By taking
\[
C' = C'(\Delta, \delta) := \max \left\{ 2, \max_{(\alpha^* + \delta) \Delta \le q < 3\Delta} C(\Delta,\delta,q) \right\},
\]
we get an upper bound $C' n \log(n/\eps)$ for the mixing time. 
\end{proof}

Finally, we give the proof for the monomer-dimer model. 
\begin{proof}[Proof of \cref{thm:matching}]
Notice that the monomer-dimer model on $G$ is equivalent to the hard-core model on the line graph of $G$; so \cref{thm:main} is still applicable. 
\cref{thm:spec-ind-matching} shows that the Gibbs distribution $\mu$ of the monomer-dimer model is $\eta$-spectrally independent for 
\[
\eta = \min\wrapc{2\lambda\Delta, 2\sqrt{1 + \lambda\Delta}}.
\]
Meanwhile, by considering the worst configuration on the 2-hop neighborhood one can show that $\mu$ is $\bb$-marginally bounded for some $\bb = \bb(\Delta,\lambda)$. 
Thus, the theorem then follows from \cref{thm:main}. 
\end{proof}

\section{Open Problems}\label{sec:conclusion}
\begin{itemize}
    \item Can we improve the approximate tensorization constant $C_1$ in \cref{thm:at} and the mixing time bound in \cref{thm:main} with a better dependence on the maximum degree $\Delta$ and on the spectral independence $\eta$? 
    For example, for the hard-core model when $\lambda\le (1-\delta)\lambda_c(\Delta)$, 
    currently our mixing time bound scales as $\Delta^{O(\Delta^2/\delta)} \times O(n\log n)$. 
    Can we improve it and get $\poly(\Delta, 1/\delta) \, n \log n$? This has recently been partially resolved in \cite{CFYZ21}, who showed a lower bound of $\frac{1}{C(\delta)n}$ on the spectral gap of the Glauber dynamics, where $C(\delta)$ is a constant depending only on $\delta > 0$.
    
    \item One can show that the spectral independence of the monomer-dimer model on the infinite $\Delta$-regular tree $\T_{\Delta}$ is exactly $\frac{2x}{1 - x}$ where
    \begin{align*}
        x = \frac{1}{\Delta-1} \wrapp{1 - \frac{2}{\sqrt{1 + 4\lambda(\Delta-1)} + 1}}
    \end{align*}
    is the (unsigned) pairwise influence between edges of $\T_{\Delta}$ sharing a vertex. Note that for $\Delta = 2$, this is $\sqrt{1 + 4\lambda} - 1 = \Theta(\sqrt{\lambda})$, while for $\Delta \geq 3$, the spectral independence $O(1/\Delta)$ independent of $\lambda$.
    
    This suggests that while the bound in \cref{thm:matchingsspecind} on the total influence of an edge is tight, the bound in \cref{thm:spec-ind-matching} on the maximum eigenvalue obtained by controlling the $\infty$-norm of the influence matrix is not tight, in contrast to the upper bound in \cite{CLV20} for vertex two-spin systems (which has a matching lower bound \cite{ALO20}). It would be interesting to obtain improved bounds on the spectral independence for the monomer-dimer model.
    
\end{itemize}
\vfill

\pagebreak
\bibliographystyle{siamplain}
\bibliography{GD.bib}
\pagebreak

\begin{appendices}
\crefalias{section}{appsec}
\crefalias{subsection}{appsec}

\section{Factorization and Contraction of Variance}
\label{app:var}

In this paper we studied the approximate tensorization and uniform block factorization of entropy for spin systems, and local-to-global entropy contraction for weighted simplicial complexes. 
There are analogs of these notions in terms of \emph{variance} as well, which is closely related to the spectral gap of the corresponding Markov chains. 
Moreover, most of our main technical contributions also hold in the variance setting. 
In this appendix we summarize a few definitions and theorems that we can get for variance using the techniques from this paper. The results for simplicial complexes were also independently discovered by \cite{KM20}.

Let $\mu$ be a distribution over a finite set $\Omega$. 
Recall that for every function $f : \Omega \to \R$, the variance of $f$ is defined as $\Var(f) = \mu[(f-\mu(f))^2] = \mu(f^2) - \mu(f)^2$. 
If $P$ is a reversible Markov chain on $\Omega$ with stationary distribution $\mu$, then the spectral gap of $P$ is defined as 
\[
\lambda(P) = \inf \left\{ \frac{\EE_{P}(f, f)}{\Var(f)} \biggm| f: \Omega \to \R, \, \Var(f) \neq 0 \right\}. 
\]


\subsection{Spin Systems}

For spin systems, variance and entropy share a lot of common properties.  
For example, the decomposition result \cref{fact:ent-decomp} and the tensorization lemma for product measures \cref{lem:prod-factor} both hold for variance as well; we refer to \cite{MSW03} for statements. 

We can also define the notions of approximate tensorization and uniform block factorization of variance.

\begin{definition}[Approximate Tensorization of Variance]
We say that a distribution $\mu$ over $[q]^V$ satisfies the \emph{approximate tensorization of variance} (with constant $\Cat$) if for all $f: \Omega \to \R$ we have 
\begin{equation*}\label{eq:var-approx-tensor}
\Var (f) \le \Cat \sum_{v\in V} \mu[\Var_v(f)].
\end{equation*}
\end{definition}

\begin{definition}[Uniform Block Factorization of Variance]
We say that a distribution $\mu$ over $[q]^V$ satisfies the \emph{$\ell$-uniform block factorization of variance} (with constant $C$) if for all $f: \Omega \to \R$ we have 
\begin{equation*}\label{eq:var-block-factor}
\frac{\ell}{n} \, \Var (f) \le C \cdot \frac{1}{\binom{n}{\ell}} \sum_{S\in \binom{V}{\ell}} \mu[\Var_S(f)].
\end{equation*}
\end{definition}

Moreover, it is known that approximate tensorization of variance is equivalent to the Poincar\'{e} inequality, i.e., a bound on the spectral gap. 
We refer to \cite{CMT15,CP20} for more backgrounds. 

\begin{fact}\label{fact:gap-at}
A distribution $\mu$ over $[q]^V$ satisfies the approximate tensorization of variance with constant $\Cat$ if and only if the spectral gap of the Glauber dynamics for $\mu$ satisfies $\lambda(\Pgl) \ge \frac{1}{\Cat n}$. 
\end{fact}

Similarly, the uniform block factorization of variance is the same as a bound on the spectral gap of the block dynamics which updates a random subset of vertices in each step.

One of our main results, \cref{lem:comparison}, which deduces approximate tensorization from uniform block factorization, is also true for variance.

\begin{lemma}\label{lem:var-comparison}
Let $\Delta\ge 3$ be an integer and $\bb > 0$ be a real. 
Consider the Gibbs distribution $\mu$ on an $n$-vertex graph $G$ of maximum degree at most $\Delta$ and assume that $\mu$ is $\bb$-marginally bounded. 
Suppose there exist positive reals $\theta \le \frac{b^2}{12\Delta}$ and $C$ such that 
$\mu$ satisfies the $\ceil{\theta n}$-uniform block factorization of variance with constant $C$. 
Then $\mu$ satisfies the approximate tensorization of variance with constant $\Cat =  O (C) $
where $O(\cdot)$ hides a constant factor depending only on $\bb$. 
\end{lemma}

More importantly, one can modify previous results and proofs in \cite{ALO20,CLV20,CGSV20,FGYZ20} and \cref{thm:spec-ind-matching} to show that the optimal uniform block factorization constant is $O(1)$ for spectrally independent distributions, as we can see in the next subsection. 
Thus, combining \cref{lem:var-comparison,fact:gap-at} we can already get an optimal $\Omega(1/n)$ bound on the spectral gap of the Glauber dynamics.

\subsection{Simplicial Complexes}

For simplicial complexes, variance also has similar properties as entropy; e.g., as mentioned in \cite{CGM19}, the decomposition result (see \cref{lem:entropydecomp}) holds for variance as well.

We can also define the notions of global and local variance contraction in weighted simplicial complexes. 

\begin{definition}[Global Variance Contraction]
We say a pure $n$-dimensional weighted simplicial complex $(\SC,w)$ satisfies the \emph{order-$(r,s)$ global variance contraction} with rate $\kappa = \kappa(r,s)$ if for all $f^{(s)} : \SC(s) \rightarrow \R$ we have
\[
\Var_{\pi_r}(f^{(r)}) \le (1-\kappa)\, \Var_{\pi_s}(f^{(s)}).
\]
\end{definition}

It is known that variance contraction is equivalent to the Poincar\'{e} inequality; see, e.g., \cite[Lemma 1.13]{MT06} and the references therein. 

\begin{fact}\label{fact:var-contraction-gap}
A pure $n$-dimensional weighted simplicial complex $(\SC,w)$ satisfies the order-$(r,s)$ global variance contraction with rate $\kappa$ if and only if the spectral gap of the order-$(s,r)$ down-up walk satisfies $\lambda(P_{s,r}^\vee) \ge \kappa$. 
\end{fact}

The reason behind \cref{fact:var-contraction-gap} is that $P_{s,r}^\vee = (P_s^\downarrow \cdots P_{r+1}^\downarrow) (P_r^\uparrow \cdots P_{s-1}^\uparrow)$ is the product of two operators $(P_s^\downarrow \cdots P_{r+1}^\downarrow)$ and $(P_r^\uparrow \cdots P_{s-1}^\uparrow)$ that are adjoint to each other. 
Also note that $\lambda(P_{r,s}^\wedge) = \lambda(P_{s,r}^\vee)$ and so the same holds for the up-down walk as well.

The local variance contraction is defined in the following way. 

\begin{definition}[Local Variance Contraction]\label{def:var-localentcontraction}
We say a pure $n$-dimensional weighted simplicial complex $(\SC,w)$ satisfies \emph{$(\alpha_0,\dots,\alpha_{n-2})$-local variance contraction} 
if for every $0\le k \le n-2$ and every $\tau \in \SC(k)$, it holds for all $f_\tau^{(2)} : \SC_\tau(2) \to \R$ that
\begin{align*}
    \Var_{\pi_{\tau,2}} (f_{\tau}^{(2)}) \geq (1 + \alpha_{k}) \, \Var_{\pi_{\tau,1}} (f_{\tau}^{(1)}).
\end{align*}
\end{definition}

We remark that the definition of local variance contraction is slightly stronger than its analog \cref{def:localentcontraction} for entropy, as here we consider all local functions $f_\tau^{(2)}$ while in \cref{def:localentcontraction} we only consider those induced by a global function $f^{(n)}$. 

Actually, the notion of local variance contraction is almost the same as local spectral expansion; this can be verified in the same way as for \cref{fact:var-contraction-gap}.
\begin{fact}\label{fact:var-expan}
Fix a pure $n$-dimensional weighted simplicial complex $(\SC,w)$. Then $(\SC,w)$ satisfies $(\alpha_0,\dots,\alpha_{n-2})$-local variance contraction iff it is a $(\zeta_0,\dots,\zeta_{n-2})$-local spectral expander, where $\zeta_k = \frac{1-\alpha_k}{1+\alpha_k}$ for all $k$.  
\end{fact}

Using the same proof approach, we can get the following analog of \cref{thm:local-global} for variance.

\begin{theorem}[Local-to-Global Variance Contraction]\label{thm:var-local-global}
If a pure $n$-dimensional weighted simplicial complex $(\SC,w)$ satisfies $(\alpha_0,\dots,\alpha_{n-2})$-local variance contraction, 
then it satisfies the order-$(r,s)$ global variance contraction with rate 
\[
\kappa = \frac{\sum_{k=r}^{s-1} \alpha_0  \cdots \alpha_{k-1}}{\sum_{k=0}^{s-1} \alpha_0 \cdots \alpha_{k-1}}, 
\]
and the spectral gaps of the order-$(s,r)$ down-up and order-$(r,s)$ up-down walks satisfy 
\[
\lambda(P_{s,r}^\vee) = \lambda(P_{r,s}^\wedge) \ge \frac{\sum_{k=r}^{s-1} \alpha_0  \cdots \alpha_{k-1}}{\sum_{k=0}^{s-1} \alpha_0 \cdots \alpha_{k-1}}. 
\]
\end{theorem}

Note that, \cref{thm:var-local-global}, combined with results from \cite{ALO20,CLV20,CGSV20,FGYZ20} and \cref{thm:spec-ind-matching}, can be used to establish uniform block factorization of variance for linear-sized blocks, with constant $C = O(1)$. 

To give a better illustration of \cref{thm:var-local-global}, we consider the order-$(s,s-1)$ down-up and order-$(s-1,s)$ up-down walks in a $(\zeta_0,\dots,\zeta_{n-2})$-local spectral expander. 
The main result of \cite{AL20} shows the following:
\begin{equation}\label{eq:var-AL20}
\lambda(P_{s,s-1}^\vee) = \lambda(P_{s-1,s}^\wedge) \ge \frac{1}{s} (1-\zeta_0) \cdots (1-\zeta_{s-2}). 
\end{equation}
Meanwhile, combining \cref{fact:var-contraction-gap}, \cref{fact:var-expan}, and \cref{thm:var-local-global}, we can get
\begin{equation}\label{eq:var-CLV}
\lambda(P_{s,s-1}^\vee) = \lambda(P_{s-1,s}^\wedge) \ge \frac{(1-\zeta_0) \cdots (1-\zeta_{s-2})} {\sum_{k=0}^{s-1} (1-\zeta_0) \cdots (1-\zeta_{k-1}) (1+\zeta_k) \cdots (1+\zeta_{s-2})}. 
\end{equation}
To understand the two bounds, consider two special cases: 
(1) If $\zeta_k = 0$ for all $k$ (corresponding to strongly log-concave distributions), then both \cref{eq:var-AL20,eq:var-CLV} give the same bound $\frac{1}{s}$; 
(2) If $\zeta_k = \frac{1}{k+2}$ for all $k$ (see Corollary 1.6 in \cite{AL20} and applications of this setting), both \cref{eq:var-AL20,eq:var-CLV} give the same bound $\frac{1}{s^2}$. 
However, it was shown in \cite{GM20} that the bound \cref{eq:var-AL20} is always no worse than \cref{eq:var-CLV} we get here, since the sequence of local variance contraction rates must be consistent with the trickling down theorem of \cite{Opp18}.

\end{appendices}
\end{document}